\documentclass[conference]{IEEEtran}
\usepackage{mathtools, cuted} 
\usepackage{lettrine} 
\usepackage{amsthm,amsmath} 
\usepackage{epsfig,amssymb,amsbsy,verbatim,subfigure,array}
\usepackage{pstricks,cite,url}
\usepackage{multicol,afterpage,wrapfig,paralist} 
\usepackage{hyperref}
\usepackage{breakurl}
\usepackage{verbatim,tikz,subfigure}
\usepackage{tikz,pgfplots}
\usepackage[english]{babel}
 
\newtheorem{lemma}{Lemma}{}
  
  \newtheorem{theorem}{Theorem}
  \newtheorem{prop}{Proposition}

\def\rk#1{{\red\sffamily\small\em $\Rightarrow$  (RK) #1 $\Leftarrow$}}
\def\d{\mathrm{d}}
\def\E{\mathrm{E}}
\renewcommand\qedsymbol{$\blacksquare$}

\title{Interference Characterization in Downlink \\ Li-Fi Optical Attocell Networks}  
\author{\IEEEauthorblockN{Atchutananda Surampudi and Radha Krishna Ganti}
\IEEEauthorblockA{Department of Electrical Engineering,\\
Indian Institute of Technology Madras,\\
Chennai, India 600036.\\
\{ee16s003, rganti\}@ee.iitm.ac.in}
}
\begin{document}
\maketitle

\begin{abstract}
Wireless access to data using visible light, popularly known as light-fidelity (Li-Fi), is one of the key emerging technologies which promises huge bandwidths and data rates. In Li-Fi, the data is modulated on optical intensities and transmitted and detected using light-emitting-diodes (LED) and photodiodes respectively. A network of such LED access points illuminates a given region in the form of attocells. Akin, to wireless networks, co-channel interference or simply interference is a major impediment in Li-Fi attocell networks. Also, when in such networks, the field-of-view (FOV) of a photodiode is limited, the network interference distribution gets affected significantly. So, for any given network scenario, interference characterization is critical for good system design. Currently, there are no good closed-form approximations to interference in Li-Fi attocell networks, that can be used for the analysis of signal-to-interference-plus-noise-ratio (or coverage), particularly for the case of limited FOVs. In this paper, using a technique from Fourier analysis, we provide a very close approximation to interference in one and two dimension Li-Fi attocell networks for any given finite inter-LED separation. We validate the interference approximation by providing theoretical error bounds using asymptotics and by performing numerical simulations. We show that our method of approximation can be extended to characterize interference in limited FOV scenarios as well.  
\end{abstract}

\section{Introduction}
Light-Fidelity (Li-Fi) is being seen as one of the key emerging technologies to provide wireless access of data using visible light at high data rates \cite{hass}. In Li-Fi, the data is usually intensity modulated to the visible light using light emitting diodes (LED), also called as downlink Li-Fi access points.  The modulated intensities travel through an optical channel and are detected by a receiver photodiode (PD). There have been several experiments conducted \cite{barrykahn1}, \cite{barrykahn2}, \cite{barrykahn3}, \cite{barrykahn4}, \cite{barrykahn5} to determine the optical wireless channel model and how it behaves with the transmitted visible light intensities. The channel is usually modelled as a linear and time invariant system \cite{heraldhass} and as a result, time varying fading on the line-of-sight links is absent. \\ \par

The LED access points are usually arranged  in a regular geometry  to form a Li-Fi attocell network. In such a network, the LEDs simultaneously transmit information packets on modulated intensities of different colours or light wavelengths. The LEDs transmitting on the same optical wavelength can be considered as co-channel interferers. Co-channel interference or simply interference in downlink of such networks, is one of the limiting factors which decreases the downlink system throughput. The interference experienced inside the attocell of the serving LED, depends on the location of the user relative to the interferers and the field-of-view (FOV) of the PD\footnote{When the FOV $<\frac{\pi}{2}$ radians, an interfering LED which does not have a line of sight link within the PD's FOV range, cannot be considered as a potential interferer.}. Additionally, the limitation of the FOV significantly affects the network interference distribution inside the serving attocell compared to the case when FOV is $\frac{\pi}{2}$ radians. So, for both the scenarios of FOV, the characterization of interference and Signal-to-Interference-plus-Noise-Ratio (SINR), is critical to understand the system performance and for good system design. Moreover, a simple closed form characterization, for both the cases of FOV, can be further used for simple analytical computation of other metrics like probability of coverage and area spectral efficiency. 

\subsection{Related works and common approaches}
In \cite{cheng1}, \cite{cheng2}, \cite{cheng3}, the SINR has been used to analyze fractional frequency reuse and angle diversity schemes, where the interference is calculated by numerical techniques. The order or number of terms of the interference summation increases linearly with the size of the network and one has to resort to simulations for understanding the behaviour of the system. 
In \cite{hass4}, the downlink system performance and interference are analyzed in Li-Fi optical attocell networks. There, for a deterministic hexagonal geometry, the interference in an infinite attocell network is approximated by only the first layer of hexagonal interferers around the central attocell using the flower model approximation \cite{flower}. Similarly in \cite{hass5}, the interference is obtained as a finite summation over the six interferers in the first layer of the hexagonal LED arrangement. But in a Li-Fi attocell network, when the inter LED separation reduces, more layers need to be considered into the interference approximation and hence the first layer approximation remains sub-optimal. Moreover, such approximations cannot be extended to any other deterministic lattice and any finite separation between the LEDs. Further, the analysis has been done only for the case of FOV = $\frac{\pi}{2}$ radians. In \cite{orient}, the problem of orientation and FOV of the PD in Li-Fi networks has been discussed to derive closed form expressions for the channel gain characteristics and probability of coverage. But, the characterization for both one and two dimension attocell networks and for any separation distance between the LEDs has not been shown. In \cite{pp1}, for the calculation of outage probability and SINR in a random deployment of LEDs, the interference is characterized by extracting its moments from its complementary function and approximation similar to the one in \cite{pp3}. But an explicit simple closed form expression for interference in a deterministic LED arrangement has not been provided.                 

\subsection{Our approach and contributions}
 The contributions of this paper are as follows:
\begin{itemize}
\item We assume a regular arrangement of LEDs in both one and two dimensions. For such an arrangement of LEDs, a close approximation to interference has been proposed for any given finite separation between the LEDs. Here we assume that the FOV of the PD used in the network is $=\frac{\pi}{2}$ radians. So, being a simple closed form expression, large scale network summations are shown to be circumvented using this characterization. 
\item The above results are generalised to  characterize the  interference when the photodiodes used in the environment have an FOV $<\frac{\pi}{2}$ radians. 
\item Theoretical error bounds have also been provided for the approximation using asymptotics, which give a clear idea on how good is the approximation for a given set of network parameters. The error bounds are validated through extensive numerical simulations. 

\end{itemize}  

This paper is arranged as follows. Section II describes the downlink system model and the arrangement of Li-Fi LEDs in both one and two dimension attocell network models. Section III is the main technical section of the paper, which describes our interference characterization (along with the FOV limitation case) in both one and two dimensions. The paper concludes with Section IV.

\section{Downlink System Model}
In this section, we describe the assumptions made for the line of sight channel model and derive the SINR at any location on the ground in such a communication scenario. Also, we describe the attocell network models, considered in this study, for both one and two dimensions. The attocell dimension or the attocell length, both refer to the inter-LED separation in the network.  
\subsection{Propagation channel  assumptions}
The optical wireless channel is considered as a linear time invariant attenuation channel \cite{dimitrov}. Further, for simplicity, the small scale path loss or fading due to multi path is neglected in this work. In Li-Fi, the baseband signal modulates the intensity of the optical signal, not the amplitude or phase. This is called the intensity-modulation and direct-detection (IM/DD). In \cite{islim}, various modulation techniques for Li-Fi have been discussed and compared. In this study, we consider a single carrier method of IM/DD, namely the non-return-zero-on-off-keying (NRZ-OOK)\footnote{While we choose NRZ-OOK for simplicity, the SINR expression holds true for other IM/DD modulation schemes as well with simple modifications.}. Moreover, we neglect any non-linear effects of the LED during intensity modulation.  
\begin{figure}[ht]
\centering
 \begin{tikzpicture}
 \draw [dashed,thick] (1,1) -- (4,7);
 \draw [fill] (2.5,1) circle [radius=0.05];
 \node [below] at (2.5,1) {\small $(0,0)$};
 \node [below] at (3.2,1) {\small $d$};
 \draw [dotted, thick ] (4,6) circle [radius=1];
 \draw [->, help lines ] (4,7) -- (2,6);
 \draw [->, help lines ] (4,7) -- (6,6);
 \draw [fill] (4,7) -- (4.5,6.5) -- (3.5,6.5) -- (4,7);
 \draw [fill] (0.75,1) -- (1.25,1) -- (1.25,1.25) -- (0.75,1.25) -- (0.75,1); 
 \draw [->, help lines ] (1,1) -- (2,2);
 \draw [->, help lines ] (1,1) -- (0,2);
 \draw [ dotted ] (1,1) -- (1,3);
 \draw [ dotted ] (4,7) -- (4,1);
 \draw (0,1) -- (5,1);
 \draw (5,7) -- (0,7);
 \draw (4,6.1055) arc [radius=0.894, start angle=270, end angle= 333.435];
 \draw (4,6) arc [radius=1, start angle=270, end angle=243.4349];
 \draw (1,1.8) arc [radius=0.8, start angle=90, end angle= 135];
 \draw (1,2) arc [radius=1, start angle=90, end angle=63.434];
 \node [align=center, below] at (1,1) {\small $(-z,0)$\\ \small Receiver photodiode (PD)};
 \node [align=center, above] at (4,7) {\small Transmitter LED\\ \small $(d,h)$};
 \node [right] at (4,3) {\small $h$};
 \node [below] at (1.7,1) {\small $z$};
 \node [align=center, right] at (4,5.85) {\small $\theta_{h}$ \\ \small (HPSA)};
 \node [left] at (4.139,5.7) {\small $\theta_{d,t}$};
 \node [above right] at (0.879,2.2) {\small $\theta_{d,r}$};
 \node [above left] at (1,1.8) {\small $\theta_{f}$};
 \end{tikzpicture} 
\caption{This figure shows the free space line-of-sight (LOS) light propagation geometry. The triangular shaped LED source is at a height $h$ and distance $d$ from the origin $(0,0)$ and is tagged to the PD at a distance $z$ on the ground. The PD has a given field-of-view (FOV) $\theta_{f}$. The free space LOS link from the LED to PD is shown by the dashed line. The angles $\theta_{d,t}$ and $\theta_{d,r}$ are respectively the transmission angle at the LED and incidence angle to the PD with respect to the normals drawn as dotted lines. We assume that the PD has no orientation towards the LED and its surface is parallel to the ground. So, we have $\theta_{d,t}=\theta_{d,r}$. $\theta_{h}$ is the half power semi angle (HPSA) of the LED. In this figure, the distance $D_{d}$, on ground, between the LED and the PD is z+d. This is adapted from \cite{cheng2}.}
\label{los_model}
\end{figure}
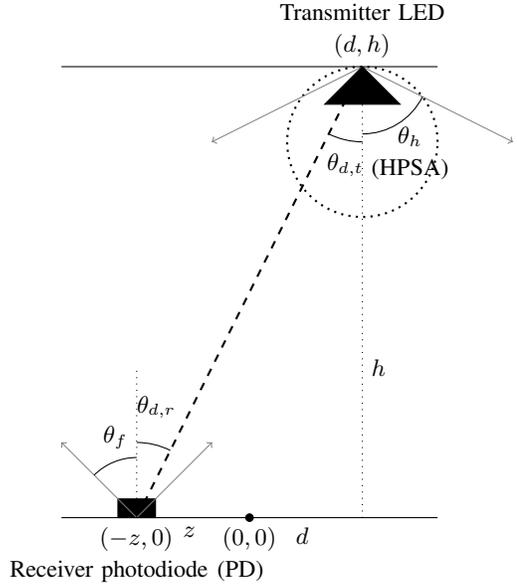

Consider the free space Li-Fi downlink of an LED-PD communication scenario shown in Fig.\ref{los_model} (dashed line). Let the light source be at an elevation height $h$ and distance $d$ from the origin $(0,0)$ and let the PD be at a distance $z$ on the ground. $\theta_{d,t}$ is the transmission angle from the LED which is at a distance $d$ from the origin and $\theta_{d,r}$ is the angle of incidence at the PD, from the same LED. We assume that the PD has no orientation towards the LED and its surface is parallel to the ground. So, we have $\theta_{d,t}=\theta_{d,r}$. $\theta_{f}$ denotes the FOV of the PD, which is the maximum angle to which the received rays can be detected. $\theta_{h}$ denotes the half-power-semi-angle (HPSA) of the transmitter LED, which is the angle at which the optical power becomes half of the power at normal. Let $A_{pd}$ be the light receiving cross sectional area of the PD. Let $D_{d}$ $($which equals $z+d$ in Fig.\ref{los_model}, but not shown explicitly$)$ be the distance on ground, between the PD $($located at a distance $z$ on the ground from $(0,0)$$)$ and the LED $($located at $(d,h)$ from $(0,0)$$)$. From \cite[Eqn. 1]{cheng2}, the channel gain from the LED to the PD with a given FOV $\theta_{f}$  is 
\begin{equation} 
   G_{d}(z) = \frac{(m+1)A_{pd}h^{m+1}}{2\pi}(D_{d}^{2} + h^{2})^{\frac{-(m+3)}{2}}\rho(D_{d}),
\label{eqn:gain}  
\end{equation}
where $m=-\frac{\ln(2)}{\ln(\cos(\theta_{h}))}$ is the Lambertian emission order of the LED and $\rho(D_{d})$ is the FOV constraint function defined as     
\begin{equation*}
\rho(D_{d}) = \left\{
               \begin{array}{ll}
                 1,& |D_{d}|\leq h\tan(\theta_{f}), \\
                 0,& |D_{d}|> h\tan(\theta_{f}). 
              \end{array}
              \right.
\end{equation*}
\subsection{The SINR expression } 
Extending the above discussion, we consider the downlink of a Li-Fi attocell network in one dimension to derive the SINR expression. 
In the attocell network, all the LEDs, as data access points, illuminate a given region in the form of attocells. An attocell is the region of data coverage due to illumination on the ground (or surface) by a particular LED, where, this LED becomes the nearest data source to a PD to be tagged upon, inside that region.
The optical attocell dimensions are in the range of metres. The co-channel LEDs, which illuminate at the same visible light wavelength, interfere. We consider interference at the PD only due to line of sight LEDs, fixed at a height $h$ and symmetrically arranged with uniform separation $a$ in an infinite one dimension corridor as shown in Fig. \ref{one_dim}. 
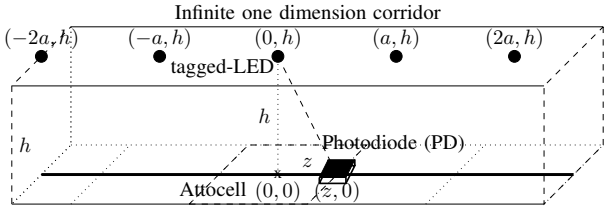
\begin{figure}[ht]
\centering
\resizebox{0.45\textwidth}{!}{%
 \begin{tikzpicture}
 \draw [dashed] (0,0) -- (0,2) -- (1,3);
 \node [right] at (0,1) {\normalsize $h$};
 \node [left] at (4.5,1.5) {\normalsize $h$};
 \node [above] at (5,0.5) {\normalsize $z$}; 
 \draw (1,3) -- (10,3);
 \draw [dashed] (10,3) -- (10,1) -- (9,0) -- (9,2) -- (10,3); 
 \draw (0,2) -- (9,2);
 \draw (0,0) -- (9,0);
 \draw [fill] (0.5,2.5) circle [radius=0.10];
 \draw [fill] (2.5,2.5) circle [radius=0.10];
 \draw [fill] (4.5,2.5) circle [radius=0.10];
 \draw [fill] (6.5,2.5) circle [radius=0.10];
 \draw [fill] (8.5,2.5) circle [radius=0.10];
 \draw [dotted] (1,3) -- (1,1) -- (0,0);
 \draw [dotted] (1,3) -- (1,1) -- (0,0);
 \draw [dotted] (1,1) -- (10,1);
 \draw [line width= 0.05cm] (0.5,0.5) -- (9.5,0.5);
 \node [above] at (5,3) {\normalsize Infinite one dimension corridor };
 \node [below] at (4.5,0.5) {\normalsize $(0,0)$};
 \node [below] at (5.5,0.5) {\normalsize $(z,0)$};
 \node [above] at (4.5,2.5) {\normalsize $(0,h)$};
 \node [above] at (6.5,2.5) {\normalsize $(a,h)$};
 \node [above] at (8.5,2.5) {\normalsize $(2a,h)$};
 \node [above] at (2.5,2.5) {\normalsize $(-a,h)$};
 \node [above] at (0.5,2.5) {\normalsize $(-2a,h)$};
 \node at (4.5,0.5) {\footnotesize x};
 \draw [thick] (5.2,0.35) -- (5.65,0.35) -- (5.8,0.65) -- (5.35,0.65) -- (5.2,0.35);
 \draw [fill] (5.2,0.45) -- (5.65,0.45) -- (5.8,0.75) -- (5.35,0.75) -- (5.2,0.45);
 \draw [thick] (5.2,0.35) -- (5.2,0.45);
 \draw [thick] (5.65,0.35) -- (5.65,0.45);
 \draw [thick] (5.8,0.65) -- (5.8,0.75);
 \draw [thick] (5.35,0.65) -- (5.35,0.75);
 \draw [dashed] (4.5,2.5) -- (5.5,0.5);
 \draw [dotted] (4.5,0.5) -- (4.5,2.5);
 \draw [dotted] (1,0) -- (2,1);
 \draw [dotted] (3,0) -- (4,1);
 \draw [dotted] (5,0) -- (6,1);
 \draw [dotted] (7,0) -- (8,1);
 \draw [dotted] (9,0) -- (10,1);
 \draw [dashed] (3,0) -- (5,0) -- (6,1) -- (4,1) -- (3,0);
 \node [above] at (3.4,0.001) {\normalsize Attocell};
 \node [left] at (4.55,2.3) {\normalsize tagged-LED};
 \node [above] at (6.45,0.75) {\normalsize Photodiode (PD)};
 \end{tikzpicture}
 }%
\caption{\textit{(One dimension model)} This figure shows the infinite one dimension corridor. There are infinite number of LEDs (circular dots) arranged at an equal interval $a$, all along the corridor, installed at a height $h$. The rectangular dotted regions on ground depict the attocells corresponding to each LED above. The user PD (small cuboid) at $(z,0)$ (inside one of the attocell), receives data wirelessly from the tagged-LED corresponding to the attocell in which it is located. Here, that attocell is highlighted as dash-dot. All other LEDs are co-channel interferers. Here, we assume that the user PD moves only along the thick line on ground, i.e length of the corridor.}
\label{one_dim}
\end{figure} 

We assume that all the LEDs operate at the same optical wavelength and transmit at same average optical power $P_{o}$. So, all the LEDs, other than the tagged-LED at $(0,h)$, are interferers, as shown in Fig. \ref{one_dim}. We calculate the SINR $\gamma(z)$, at every PD location $z$, inside the attocell. Let $x_{i}(t)$ be the baseband signal, during the time slot $t$, from each $i^{th}$ LED in the network before transmission. Let $s_{i}(t)$ be the optical IM signal on baseband signal $x_{i}(t)$, during the time slot $t$. Using the gain expression in \eqref{eqn:gain} and the geometry of the links in Fig. \ref{los_model}, we can modify $G_{d}(z)$ $($and distance $D_{d}$ from every other LED at $(ia,h)$$)$ as 
\begin{equation}  
G_{ia}(z) = \frac{(m+1)A_{pd}h^{m+1}}{2\pi}((z+ia)^{2} + h^{2} )^{\frac{-(m+3)}{2}}\rho(D_{ia}).
\label{eqn:gain2}  
\end{equation}
Now, the signal current $I(z,t)$ (in amperes), received at the PD, at $(z,0)$ with responsivity $R_{pd}$, during the time slot $t$ is given as 
\begin{align}
I(z,t) =&  \ s_{0}(t) G_{0}(z) R_{pd} \nonumber\\ 
&+ \sum_{i = -\infty \setminus 0}^{+\infty}s_{i}(t) G_{ia}(z) R_{pd} + n(t).
\label{eqn:recsig1} 
\end{align}
In \eqref{eqn:recsig1}, $n(t)$ is the noise current at the PD, which is modelled as additive white Gaussian noise, has a noise power spectral density of $N_{o}$. If the total IM bandwidth of the receiver PD is $W$ (which can be assumed as the total system bandwidth), then the total receiver noise variance $\sigma^{2}$, at the PD is
\begin{equation*} 
\sigma^{2} = N_{o}W. 
\end{equation*} 
From \cite{armstrong}, the average transmit optical power $P_{o}$, for every $i^{th}$ LED can be defined as
\begin{equation*}
P_{o} = \mathbb{E} [s_{i}(t)],
\end{equation*} 
where $\mathbb{E}[.]$ is the expectation operator over time slot $t$. The average received current $I_{i}(z)=\mathbb{E}[s_{i}(t)G_{ia}(z)R_{pd}]$, at the PD from the $i^{th}$ LED, after suffering through the channel gain $G_{ia}(z)$, is
\begin{equation*}
I_{i}(z)= P_{o}G_{ia}(z)R_{pd}.  
\end{equation*} 
So, $\gamma(z)$, at user position $z$ is
\begin{align}
\gamma(z) &= \frac{I_{0}^{2}(z)}{\sum_{i = -\infty \setminus 0}^{+\infty}I_{i}^{2}(z)+\sigma^{2}},\nonumber\\
&=\frac{P_{o}^{2}G_{0}^{2}(z)R_{pd}^{2}}{\sum_{i = -\infty \setminus 0}^{+\infty}P_{o}^{2}G_{ia}^{2}(z)R_{pd}^{2} + \sigma^{2}}. 
\label{eqn:sinrf1}  
\end{align}
Now, substituting for $G_{ia}(z)$ from \eqref{eqn:gain2} into \eqref{eqn:sinrf1} and further rearranging the constants, we have
\begin{equation}
\gamma(z) = \frac{(z^{2}+h^{2})^{-m-3}\rho(D_{0})}{ \sum_{i = -\infty \setminus 0}^{+\infty}( (ia + z)^{2} + h^{2} )^{-m-3}\rho(D_{ia})  + \Omega},
\label{eqn:sinrl1}  
\end{equation}
where $\Omega$ is given as
\begin{equation*}
\Omega = \frac{4\pi^{2}N_{0}W}{P_{o}^{2}(m+1)^{2}A_{pd}^{2}R_{pd}^{2}h^{2m+2}}.
\end{equation*}  

\subsection{Attocell network models} 
In this work, we consider two cases of lighting described below.
\subsubsection{One dimension infinite corridor network}
We consider an infinite length corridor, along which an infinite number of LEDs are arranged with uniform spacing $a$, as shown in Fig. \ref{one_dim}. Importantly, we also assume that all the LEDs are Li-Fi capable and all transmit data at the same time along with illumination. The corresponding derivation for SINR was shown in the previous subsection and was derived in \eqref{eqn:sinrl1} as 
\begin{equation}
\gamma(z) = \frac{(z^{2}+h^{2})^{-m-3}\rho(D_{0})}{ \mathbb{I}_{\infty}(z) + \Omega},
\label{eqn:sinrl11}  
\end{equation}
where the interference term\footnote{In this work, we characterize the normalized interference power $\mathbb{I}_{\infty}(z)$ (for one dimension) and $\mathbb{I}_{\infty}(d_{x},d_{y})$ (for two dimension model), normalized by the average optical power $P_{o}$. This we simply call the interference. So, all the assumed practical dimensions and further derived theoretical expressions for interference get linearly scaled by $P_{o}$, if it has to be introduced.} $\mathbb{I}_{\infty}(z)$ in \eqref{eqn:sinrl11}, is given as  
\begin{equation}
\mathbb{I}_{\infty}(z) =  \sum_{i = -\infty \setminus 0}^{+\infty}( (ia + z)^{2} + h^{2} )^{-m-3}\rho(D_{ia}).
\label{eqn:interf1}  
\end{equation}
Also, for a PD with an FOV $\theta_{f} = \frac{\pi}{2}$ radians, $\mathbb{I}_{\infty}(z)$ in \eqref{eqn:interf1} can be written as 
\begin{equation}
\mathbb{I}_{\infty}(z) =  \sum_{i = -\infty \setminus 0}^{+\infty}( (ia + z)^{2} + h^{2} )^{-m-3}.
\label{eqn:interf1a}  
\end{equation}

\subsubsection{Two dimension infinitely spread square grid network}
The two dimension network model is shown in Fig. \ref{two_dim}. Let the user PD be located at distance  $z = \sqrt{d_{x}^{2} + d_{y}^{2}}$ from the origin inside the respective attocell of the LED. Here, the tagged-LED, considered at $(0,0,h)$, has an attocell symmetrically around it on the ground, as a square of dimension $a$. Similar to the one dimension model, importantly, we here too assume that all the LEDs are Li-Fi capable and all transmit data at the same time along with illumination. From the one dimension case, the same expression for the SINR can be extended to a two dimension scenario. Let the interfering LEDs, indexed by ($u,v$), be located at $(u,v,h)$.
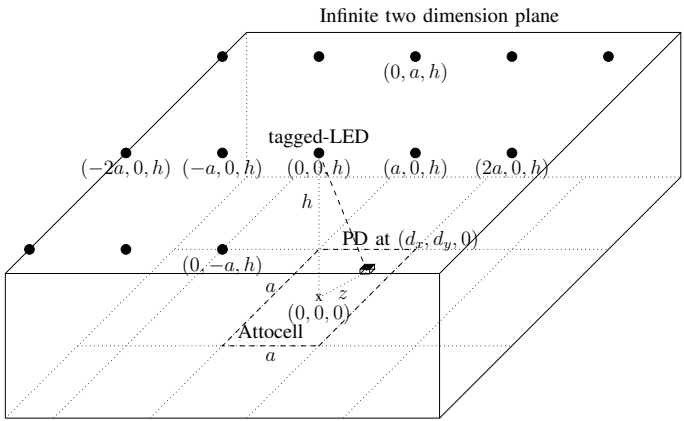
\begin{figure}[ht] 
\centering
 \resizebox{0.5\textwidth}{!}{%
 \begin{tikzpicture}
 \draw (0,0) -- (0,3) -- (5,8) -- (14,8) -- (14,5);
 \draw (14,5) -- (9,0) -- (0,0);
 \draw (14,8) -- (9,3) -- (0,3);
 \draw (9,0) -- (9,3);
 \draw [dotted] (0,0) -- (5,5) -- (14,5);
 \draw [dotted] (5,5) -- (5,8);
 
 \draw [fill] (0.5,3.5) circle [radius=0.10];
 
 \draw [fill] (2.5,3.5) circle [radius=0.10];
 
 \draw [fill] (4.5,3.5) circle [radius=0.10];
 \node [below] at (4.5,3.5) {\large $(0,-a,h)$};
 
 
 
 \draw [fill] (2.5,5.5) circle [radius=0.10];
 \node [below] at (2.5,5.5) {\large $(-2a,0,h)$};
 \node at (5.5,1.8) {\large Attocell};
 
 \draw [fill] (4.5,5.5) circle [radius=0.10];
 \node [below] at (4.5,5.5) {\large $(-a,0,h)$};
 
 \draw [fill] (6.5,5.5) circle [radius=0.10];
 \node [below] at (6.5,5.5) {\large $(0,0,h)$};
 \node [above] at (6.5,5.5) {\large tagged-LED};
 
 \draw [fill] (8.5,5.5) circle [radius=0.10];
 \node [below] at (8.5,5.5) {\large $(a,0,h)$};
 
 \draw [fill] (10.5,5.5) circle [radius=0.10];
 \node [below] at (10.5,5.5) {\large $(2a,0,h)$};
 
 \draw [fill] (4.5,7.5) circle [radius=0.10];
 
 \draw [fill] (6.5,7.5) circle [radius=0.10];
 
 \draw [fill] (8.5,7.5) circle [radius=0.10];
 \node [below] at (8.5,7.5) {\large $(0,a,h)$};
 
 \draw [fill] (10.5,7.5) circle [radius=0.10];
 
 \draw [fill] (12.5,7.5) circle [radius=0.10]; 
 
 \draw [dotted] (1,0) -- (6,5);
 \draw [dotted] (3,0) -- (8,5);
 \draw [dotted] (5,0) -- (10,5);
 \draw [dotted] (7,0) -- (12,5);
 \draw [dotted] (1.5,1.5) -- (10.5,1.5);
 \draw [dotted] (3.5,3.5) -- (12.5,3.5);
 \draw [dotted] (6.5,2.5) -- (6.5,5.5);
 \node [align=center] at (6.5,2.5) {\small x};
 \node [below] at (6.5,2.5) {\large $(0,0,0)$};
 \node [left] at (6.5,4.5) {\large $h$}; 
 \draw (7.35,2.99) -- (7.55,2.99) -- (7.65,3.09) -- (7.45,3.09) -- (7.35,2.99);
 \draw [fill] (7.35,3.09) -- (7.55,3.09) -- (7.65,3.19) -- (7.45,3.19) -- (7.35,3.09);
 \draw (7.35,2.99) -- (7.35,3.09);
 \draw (7.55,2.99) -- (7.55,3.09);
 \draw (7.65,3.09) -- (7.65,3.19);
 \draw (7.45,3.09) -- (7.45,3.19);
 \draw [dotted] (6.5,2.5) -- (7.5,3.04);
\node [below] at (8.4,4.04) {\large PD at $(d_{x},d_{y},0)$};
 \node [below] at (7,2.77) {\large $z$};
 \draw [dashed] (6.5,5.5) -- (7.5,3.04); 
 \draw [dashed] (6.5,1.5) -- (8.5,3.5) -- (6.5,3.5) -- (4.5,1.5) -- (6.5,1.5);
 \node [above] at (5.5,2.5) {\large $a$};
 \node [below] at (5.5,1.5) {\large $a$};
 \node [above] at (9,8) {\large Infinite two dimension plane}; 
 
 \end{tikzpicture}
}%
\caption{\textit{(Two dimension model)} This figure shows the infinite two dimension model. There are infinite number of LEDs (circular dots) arranged symmetrically at regular intervals of $a$ as a uniform square grid, all over the plane, installed at a height $h$. The rectangular dotted regions on ground depict the attocells corresponding to each LED above. The user PD (small cuboid) at $(d_{x},d_{y},0)$ (inside one of the attocell), receives data wirelessly from the tagged-LED corresponding to the attocell in which it is located. Here, that attocell is highlighted as dash-dot. All other LEDs are co-channel interferers. Here we assume that the user PD can move anywhere on the ground plane.}
\label{two_dim}
\end{figure} 
Now, for $D_{u,v} =\sqrt{(ua+d_{x})^{2}+(va+d_{y})^{2}}$, the FOV constraint function $\rho(.)$, for two dimensions is defined as
\begin{equation*}
\rho(D_{u,v}) = \left\{
               \begin{array}{ll}
                 1,& |D_{u,v}|\leq h\tan(\theta_{f}), \\
                 0,&|D_{u,v}|> h\tan(\theta_{f}). 
              \end{array}
              \right.
\end{equation*}
The SINR $\gamma(d_{x},d_{y})$, at a distance $z$ from origin is 
\begin{equation*}
\gamma(d_{x},d_{y}) = \frac{(z^{2}+h^{2})^{-m-3}\rho(D_{u,v})}{\mathbb{I}_{\infty}(d_{x},d_{y})  + \Omega}, 
\end{equation*}
where the interference term $\mathbb{I}_{\infty}(d_{x},d_{y})$  is given by
\begin{align}
\mathbb{I}_{\infty}(d_{x},d_{y}) =\sum_{u = -\infty}^{+\infty}\sum_{v = -\infty \setminus (0,0)}^{+\infty} &((ua+d_{x})^{2}+(va+d_{y})^{2} \nonumber\\
&+h^{2} )^{-m-3}\rho(D_{u,v}). 
\label{eqn:interf2} 
\end{align}
For a PD of FOV $\theta_{f}=\frac{\pi}{2}$ radians, $\mathbb{I}_{\infty}(d_{x},d_{y})$ in \eqref{eqn:interf2} can be written as  
\begin{align}
\mathbb{I}_{\infty}(d_{x},d_{y}) =\sum_{u = -\infty}^{+\infty}\sum_{v = -\infty \setminus (0,0)}^{+\infty} &((ua+d_{x})^{2}+(va+d_{y})^{2} \nonumber\\
&+h^{2} )^{-m-3}. 
\label{eqn:interf2a}  
\end{align}
A closed form expression for the interference term in \eqref{eqn:interf1} and \eqref{eqn:interf2} (or \eqref{eqn:interf1a} and \eqref{eqn:interf2a} for FOV=$\frac{\pi}{2}$ radians) is required in both one and two dimension scenarios, which is discussed in the following section.
\section{Interference Characterization}
In this Section, we characterize the interference as a closed form approximation using the Poisson summation theorem \cite{poisson}, which is stated for reference: 
\begin{theorem}[Poisson summation theorem]
Let $q(x)$ be a continuous function. Under some mild regularity conditions we have   
\begin{equation}
\sum_{i = -\infty}^{+ \infty}q(i) = \sum_{w=-\infty}^{+\infty}  Q(w),
\label{eqn:intb1}  
\end{equation} 
where
\[Q(w) = \int_{-\infty}^{\infty}q(t)e^{-2\pi \iota wt} \d t,\] is the Fourier transform of $q(x)$. 
\label{theorem0}   
\end{theorem}
%
In the following subsections for one and two dimension network models in succession, we first proceed with our interference characterization for FOV $\theta_{f}=\frac{\pi}{2}$ radians. In a simultaneous subsection, we show that our method of interference characterization using Fourier analysis, can be extended for the case of FOV $\theta_{f}<\frac{\pi}{2}$ radians.  

\subsection{One dimension model with FOV = $\frac{\pi}{2}$ radians} 
We now look at the interference characterization using the above Poisson summation theorem.
\begin{theorem}
Consider a photodiode, with $\theta_{f}=\frac{\pi}{2}$ radians, situated at a distance $z$ (inside an attocell) from the origin, in an infinite one dimension corridor network of Li-Fi LEDs, emitting light with a Lambertian emission order $m$, installed at a height $h$ with uniform inter-LED separation distance $a$. Then, for a wavelength reuse factor of unity, the interference $\mathbb{I}_{\infty}(z)$, caused by the co-channel interferers at the photodiode is
\begin{align*}
&\mathbb{I}_{\infty}(z) = \frac{h^{1-2\beta}\sqrt{\pi}\Gamma(\beta-0.5)}{a\Gamma(\beta)} -\frac{1}{(z^{2}+h^{2})^{\beta}}+ \sum_{w=1}^\infty g(w),
\end{align*}
where
\[g(w) =\frac{2^{2-\beta}\sqrt{2\pi}h^{0.5-\beta}(2\pi w)^{\beta-0.5} \mathbb{K}_{\beta-0.5}(\frac{2\pi hw}{a})\cos(\frac{2\pi wz}{a})}{a^{0.5+\beta}\Gamma(\beta)}.\]
Here $\Gamma(x)=\int_{0}^{\infty}t^{x-1}e^{-t}\d t$ denotes the standard Gamma function, $\beta=m+3$ and $\mathbb{K}_{v}(y)=\frac{\Gamma(v+\frac{1}{2})(2y^{v})}{\sqrt{\pi}}\int_{0}^{\infty}\frac{\cos(t)\d t}{(t^{2}+y^{2})^{v+\frac{1}{2}}}$ is the modified bessel function of second kind.
 
\label{theorem1}
\end{theorem}
\begin{proof}
The proof is provided in Appendix \ref{app:theorem1}.
\end{proof}
In the next proposition, we quantify the error when the summation in the above infinite series is truncated after\footnote{Over-usage: This variable has been used twice in the paper, but in two completely disjoint and separate contexts. In the one dimension model, $k$ represents the number of terms in the approximation that needs to be considered. Again, in the description of the two dimension model, we have used $k$ to denote the frequency term. This is due to lack of variables and the authors assure that this, in no way affects the understanding of the paper.} $k$ terms using the asymptotic notation\footnote{The asymptotic notation $f(n)=O(g(n))$ is defined as, $\exists n_{o}$ and $\exists k_{1}>0 \ni \forall n>n_{o}, f(n)\leq k_{1}\times g(n).$} $O(.)$. 
\begin{prop}
From Thm. \ref{theorem1}, for a finite integer $k$, the interference inside an attocell can be approximated to a closed form expression as
\begin{align}   
&\mathbb{I}_{\infty}(z)=\hat {\mathtt{I}}_{k}(z)+ O((k+1)^{\beta-2}e^{\frac{-2\pi h (k+1)}{a}}),
\label{eqn:newprop1}
\end{align}
where
\[\hat {\mathtt{I}}_{k}(z) \triangleq \frac{h^{1-2\beta}\sqrt{\pi}\Gamma(\beta-0.5)}{a\Gamma(\beta)}-\frac{1}{(z^{2}+h^{2})^{\beta}}+ \sum_{w=1}^k g(w).\]
\label{prop1}
\end{prop} 
\begin{proof}
The proof is provided in Appendix \ref{app:theorem1a}.
\end{proof}

Since in practice, the number of LEDs are finite, we also look at $ \mathbb{I}_{n}(z)$, i.e., looking at interference by a finite number of LEDs in \eqref{eqn:interf1a}. In Fig. \ref{one5h2}, we observe that as the number of interferers $n$ increases, the interference $\mathbb{I}_{n}(z)$, saturates to a constant value which is $\mathbb{I}_{\infty}(z)$. So, the approximation results in Prop. \ref{prop1} hold true for finite number of LEDs as well, even though the results are derived for an infinite corridor. \\ \par

We now try to understand the interference characterization in Prop. \ref{prop1} by taking a few theoretical examples and further validation through numerical simulations. Firstly, in Prop. \ref{prop1}, the interference for any position $z$ of the user inside the attocell, always has a constant term given as 
\begin{equation*}
\frac{h^{1-2\beta}\sqrt{\pi}\Gamma(\beta-0.5)}{a\Gamma(\beta)}.
\end{equation*}
This term represents the average spatial interference seen at all locations.\\ \par

\begin{figure}[ht]
\centering
\definecolor{mycolor1}{rgb}{0.00000,0.44700,0.74100}%
\definecolor{mycolor2}{rgb}{0.85000,0.32500,0.09800}%
\definecolor{mycolor3}{rgb}{0.92900,0.69400,0.12500}%
\begin{tikzpicture}
\begin{axis}[%
width=7cm,
height=5cm,
scale only axis,
xmin=0,
xmax=50,
xlabel style={font=\color{white!15!black}},
xlabel={Number of interferers ($n$)},
ymode=log,
ymin=1e-05,
ymax=0.01,
yminorticks=true,
ylabel style={font=\color{white!15!black}},
ylabel={Interference $\mathbb{I}_{n}(z)$},
axis background/.style={fill=white},
title style={font=\bfseries},
title={},
xmajorgrids,
ymajorgrids,
yminorgrids,
legend style={font=\fontsize{7}{5}\selectfont,at={(0.5410,0.108)}, anchor=south west, legend cell align=left, align=left, draw=white!15!black}
]
\addplot [color=black, line width=1.0pt]
  table[row sep=crcr]{%
1	0.00109406162488395\\
2	0.00182677062720763\\
3	0.00222817032364015\\
4	0.00242219580571311\\
5	0.00251073034781809\\
6	0.00255071979976105\\
7	0.00256911012091141\\
8	0.00257784636897835\\
9	0.00258215947683047\\
10	0.0025843755581385\\
11	0.00258555928529778\\
12	0.00258621516782645\\
13	0.00258659113425352\\
14	0.0025868134771304\\
15	0.00258694877602272\\
16	0.00258703328208266\\
17	0.00258708733552277\\
18	0.0025871226713215\\
19	0.0025871462363024\\
20	0.00258716224173445\\
21	0.00258717329722973\\
22	0.00258718105306508\\
23	0.0025871865726539\\
24	0.00258719055329508\\
25	0.00258719345968523\\
26	0.00258719560621235\\
27	0.00258719720858401\\
28	0.00258719841675672\\
29	0.00258719933626616\\
30	0.00258720004224573\\
31	0.00258720058876606\\
32	0.00258720101513732\\
33	0.00258720135021362\\
34	0.00258720161536641\\
35	0.00258720182656063\\
36	0.00258720199581916\\
37	0.00258720213226573\\
38	0.00258720224287444\\
39	0.00258720233301315\\
40	0.00258720240684054\\
41	0.00258720246759873\\
42	0.00258720251783031\\
43	0.00258720255954068\\
44	0.0025872025943201\\
45	0.00258720262343621\\
46	0.00258720264790449\\
47	0.00258720266854229\\
48	0.00258720268601051\\
49	0.00258720270084587\\
50	0.00258720271348607\\
};
\addlegendentry{$h = 2.5$m}

\addplot [color=black, dashed, line width=1.0pt]
  table[row sep=crcr]{%
1	0.000267849838808135\\
2	0.000467785402582155\\
3	0.000595359201877916\\
4	0.000668163889856196\\
5	0.000707024142689262\\
6	0.000727141411738609\\
7	0.00073750443260462\\
8	0.000742903664757206\\
9	0.000745775679538449\\
10	0.000747342913097121\\
11	0.000748221948073973\\
12	0.000748728812568853\\
13	0.000749029044295774\\
14	0.000749211494030345\\
15	0.00074932507146974\\
16	0.000749397383618117\\
17	0.000749444396530831\\
18	0.000749475560742764\\
19	0.000749496594434183\\
20	0.000749511029795084\\
21	0.000749521091480712\\
22	0.000749528206352609\\
23	0.000749533305299011\\
24	0.000749537005394575\\
25	0.000749539721844652\\
26	0.000749541737961963\\
27	0.00074954324961782\\
28	0.000749544393909852\\
29	0.000749545267916579\\
30	0.000749545941133248\\
31	0.000749546463823432\\
32	0.000749546872695607\\
33	0.000749547194805659\\
34	0.000749547450268793\\
35	0.000749547654163425\\
36	0.000749547817880788\\
37	0.000749547950090468\\
38	0.000749548057437228\\
39	0.000749548145048096\\
40	0.000749548216904322\\
41	0.000749548276116274\\
42	0.000749548325128067\\
43	0.00074954836587095\\
44	0.000749548399879024\\
45	0.000749548428377111\\
46	0.000749548452347861\\
47	0.000749548472583314\\
48	0.000749548489724735\\
49	0.000749548504293541\\
50	0.000749548516715458\\
};
\addlegendentry{$h = 3.0$m}

\addplot [color=black, line width=1.0pt, mark=o, mark size=3pt]
  table[row sep=crcr]{%
1	8.06203375688705e-05\\
2	0.000145201376444395\\
3	0.000190863960459213\\
4	0.000220184927625978\\
5	0.00023780658411884\\
6	0.000247988954703946\\
7	0.000253767531306547\\
8	0.000257038149956584\\
9	0.000258903451796995\\
10	0.000259982349742006\\
11	0.000260617592823085\\
12	0.000260999047213645\\
13	0.000261232812925775\\
14	0.000261379002993408\\
15	0.000261472246332511\\
16	0.000261532854435036\\
17	0.000261572963971205\\
18	0.000261599962246326\\
19	0.000261618427750227\\
20	0.000261631248091284\\
21	0.000261640275154892\\
22	0.000261646715659221\\
23	0.000261651367898322\\
24	0.000261654767593061\\
25	0.000261657279158883\\
26	0.000261659153677991\\
27	0.000261660566252323\\
28	0.000261661640401003\\
29	0.000261662464203676\\
30	0.000261663101114792\\
31	0.000261663597293453\\
32	0.000261663986627236\\
33	0.000261664294212206\\
34	0.000261664538787961\\
35	0.000261664734457942\\
36	0.000261664891915964\\
37	0.000261665019328345\\
38	0.000261665122973643\\
39	0.000261665207710295\\
40	0.000261665277321058\\
41	0.000261665334768546\\
42	0.000261665382386136\\
43	0.000261665422021525\\
44	0.000261665455145382\\
45	0.000261665482934078\\
46	0.000261665506333038\\
47	0.000261665526105551\\
48	0.000261665542870547\\
49	0.000261665557132013\\
50	0.000261665569301996\\
};
\addlegendentry{$h = 3.5$m}

\end{axis}
\end{tikzpicture}%
\caption{\textit{(One Dimension Model)} Here the variation of interference $\mathbb{I}_{n}(z)$, with respect to the number of interferers ($n$) is drawn for different height $h$ of LED installation. We consider $a=0.5$m, the half-power-semi-angle (HPSA) $\theta_{h}$ of the LED as $\frac{\pi}{3}$ radians and the position $z$ of the receiver photodiode (PD) at half the attocell length $\frac{a}{2}$.}
\label{one5h2}
\end{figure}
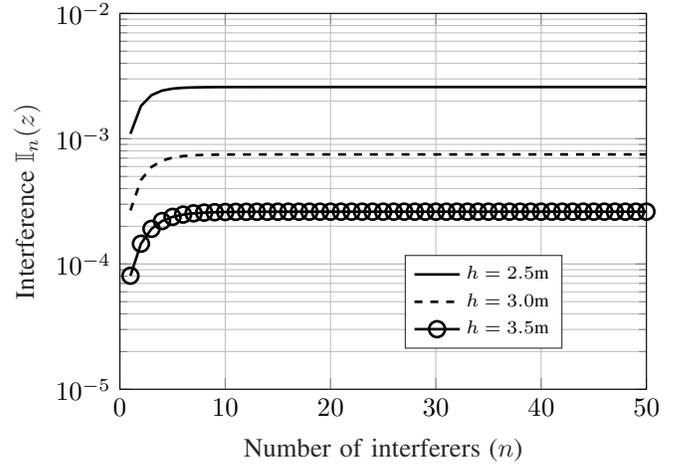 
We see that  the asymptotic error in \eqref{eqn:newprop1} becomes exponentially small when $\frac{h}{a}$ is large. Hence the interference can be well approximated with small values of $k$ as long as the ratio $\frac{h}{a}$ is large. Hence, considering only $k=0$ term, the interference can be approximated as  
\begin{equation*}
\mathbb{I}_{\infty}(z) \approx \hat {\mathtt{I}}_{0}(z) =  \frac{h^{1-2\beta}\sqrt{\pi}\Gamma(\beta-0.5)}{a\Gamma(\beta)} -\frac{1}{(z^{2}+h^{2})^{\beta}}.
\end{equation*}
For example, if we consider $a=0.2$m and $h=2.5$m, leading to $\frac{h}{a}=12.5$, we can choose $k=0$ and have a theoretical error bound  of $O(e^{-25\pi})$. 
This can be verified from Fig. \ref{space1s}. We see that all the terms from $w=1$ have negligible contribution. \\ \par 

When $\frac{h}{a}$ is not large, a few more terms $(w)$ are necessary to improve the approximation accuracy. For example, in Fig. \ref{space11}, when we consider $h=2.5$m and $a=0.5$m, leading to $\frac{h}{a}=5$; $w=0$ and $w=1$ are significant, with an error bound on $w>1$ as $O(e^{-20\pi})$. So, $k=1$ or $\hat{\mathtt{I}}_{1}(z)$ is a good approximation for this case. Further, in most practical cases, the ratio $\frac{h}{a}$ varies between $2.5$ to $5$. So, the above approximation to $\hat{\mathtt{I}}_{1}(z)$ i.e.
\[\mathbb{I}_{\infty}(z) \approx \hat {\mathtt{I}}_{1}(z) =  \frac{h^{1-2\beta}\sqrt{\pi}\Gamma(\beta-0.5)}{a\Gamma(\beta)} -\frac{1}{(z^{2}+h^{2})^{\beta}}+g(1),\]
can be extended in general to this practically seen range\footnote{For lower values of $\frac{h}{a}$ i.e. $<2.5$, $k>1$ may have to be considered to improve the approximation accuracy. Also, from the proof of Prop. \ref{prop1},  $k$ should be chosen such that  $k \geq \lceil\frac{a (\beta-2)}{2\pi h}\rceil$ for a good approximation.} of $\frac{h}{a}$ because we still have a theoretical asymptotic error bound on $w>1$ as $O(e^{-10\pi})$. \\ \par 
 
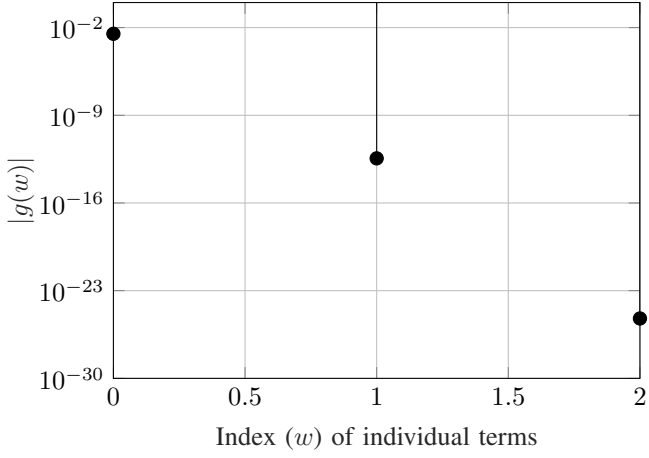
\begin{figure}[ht]
\centering
\definecolor{mycolor1}{rgb}{0.00000,0.44700,0.74100}%
\begin{tikzpicture}
\begin{axis}[%
width=7cm,
height=5cm,
scale only axis,
xmin=0,
xmax=2,
xlabel style={font=\color{white!15!black}},
xlabel={Index ($w$) of individual terms},
ymode=log,
ymin=1e-30,
ymax=1,
yminorticks=true,
ylabel style={font=\color{white!15!black}},
ylabel={$|g(w)|$},
axis background/.style={fill=white},
title style={font=\bfseries},
title={},
xmajorgrids,
ymajorgrids,
yminorgrids
]
\addplot[ycomb, mark size=2.5pt, mark=*, mark options={solid, black}, forget plot] table[row sep=crcr] {%
0	0.00321699087727595\\
1	3.64470513731168e-13\\
2	6.03306845195216e-26\\
3	4.48137428497783e-39\\
4	2.37474454654118e-52\\
5	1.04343692568582e-65\\
};
\end{axis}
\end{tikzpicture}%
\caption{\textit{(One Dimension Model)} (Large $\frac{h}{a}$ case) This graph shows the magnitude of the individual terms of $|g(w)|$ for a height $h$ of the LED = $2.5$m and $a=0.2$m. The half-power-semi-angle (HPSA) $\theta_{h}=\frac{\pi}{3}$ radians and $z=\frac{a}{2}$.}
\label{space1s}
\end{figure} 
\begin{figure}[ht]
\centering
\definecolor{mycolor1}{rgb}{0.00000,0.44700,0.74100}%
\begin{tikzpicture}
\begin{axis}[%
width=7cm,
height=5cm,
scale only axis,
xmin=0,
xmax=3,
xlabel style={font=\color{white!15!black}},
xlabel={Index ($w$) of individual terms},
ymode=log,
ymin=1e-9,
ymax=0.001,
yminorticks=true,
ylabel style={font=\color{white!15!black}},
ylabel={$|g(w)|$},
axis background/.style={fill=white},
title style={font=\bfseries},
title={},
xmajorgrids,
ymajorgrids,
yminorgrids
]
\addplot[ycomb, mark size=2.5pt, mark=*, mark options={solid, fill=black, black},forget plot] table[row sep=crcr] {%
0	0.000536165146212658\\
1	0.000152774076325489\\
2	4.00532421302961e-06\\
3	6.04051647203377e-08\\
4	6.96471370692569e-10\\
5	6.8534844415462e-12\\
};
\end{axis}
\end{tikzpicture}%
\caption{\textit{(One Dimension Model)} This graph shows the magnitude of the individual terms of $|g(w)|$ for a height $h$ of the LED = $2.5$m and $a=0.5$m. The half-power-semi-angle (HPSA) $\theta_{h}=\frac{\pi}{3}$ radians and $z=\frac{a}{2}$.}
\label{space11}
\end{figure}
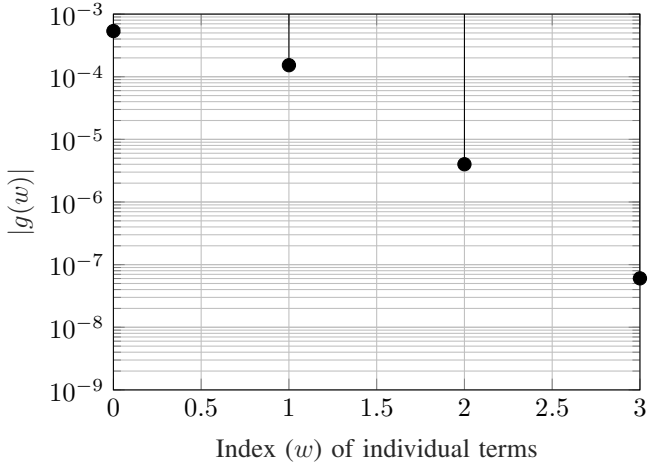 

\begin{table}
\caption{PARAMETERS : This table shows the Parameters considered in this study.}
\label{abc}
\begin{tabular}{ | m{3cm} | m{1cm}| m{2cm} | m{1.0cm} |} 
\hline
Parameter& Symbol & Value & Unit \\ 
\hline
Temperature of Operation & $T$ & $300$ & K \\ 
\hline
Noise power spectral density at Photodiode & $N_{o}$ & $4.14\times10^{-21}$ & WHz$^{-1}$ \\
\hline
Modulation bandwidth of LED & $W$ & $40\times10^{6}$ & Hz \\
\hline
Area of Photodiode & $A_{pd}$ & $10^{-4}$ & m$^{2}$ \\
\hline
Responsivity of PD & $R_{pd}$ & $0.1$ & AW$^{-1}$ \\
\hline 
\end{tabular}
\end{table}
 
For numerical validation, firstly, from Fig. \ref{one5a1}, we see the tightness of this approximation, for the above given range of $\frac{h}{a}$. We proceed by considering $h=2.5$m and plotting the interference w.r.t. the variation of $a$ from $0.1$m to $1$m $($i.e. $\frac{h}{a}$ in the range of $25$ to $2.5$$)$. We see that $\mathbb{I}_{n}(z)$ $($for $n=\{4,10,20,40\}$$)$ and $\hat{\mathtt{I}}_{1}(z)$ are tightly bounded with each other, which validates our approximation. Now, from the above numerical validation for $k=1$, we take a given value of $a=0.5$m and proceed for further numerical validation w.r.t various system parameters $h, \theta_{h}$ and $z$ in Fig. \ref{one5h1}, \ref{one5t1} and \ref{one5z1} respectively. The corresponding graphs for the approximation error $\hat{e}=|\mathbb{I}_{n}(z)-\hat{\mathtt{I}}_{1}(z)|$ are respectively shown in Fig. \ref{one5h3}, \ref{one5t3} and \ref{one5z3} for different number of interferers $n$. All the simulations are obtained using the parameter values given in Table I. \\ \par   

\begin{figure}[ht]
\centering
\definecolor{mycolor1}{rgb}{0.00000,0.44700,0.74100}%
\definecolor{mycolor2}{rgb}{0.85000,0.32500,0.09800}%
\definecolor{mycolor3}{rgb}{0.92900,0.69400,0.12500}%
\definecolor{mycolor4}{rgb}{0.49400,0.18400,0.55600}%
\definecolor{mycolor5}{rgb}{0.46600,0.67400,0.18800}%
\begin{tikzpicture}
\begin{axis}[%
width=7cm,
height=5cm,
scale only axis,
xmin=0.1,
xmax=1,
xlabel style={font=\color{white!15!black}},
xlabel={Inter-LED spacing ($a$) in metres},
ymode=log,
ymin=0.0003469446952,
ymax=0.016,
yminorticks=true,
ylabel style={font=\color{white!15!black}},
ylabel={Interference},
axis background/.style={fill=white},
title style={font=\bfseries},
title={},
xmajorgrids,
ymajorgrids,
yminorgrids,
legend style={font=\fontsize{7}{5}\selectfont,at={(0.375,0.585)}, anchor=south west, legend cell align=left, align=left, draw=white!15!black}
]
\addplot [color=black, dashed, line width=1.0pt]
  table[row sep=crcr]{%
0.1	0.015430641914642\\
0.14	0.0108359443201864\\
0.18	0.00828411217738649\\
0.22	0.00666103354200665\\
0.26	0.00553820163210812\\
0.3	0.00471564431916954\\
0.34	0.00408749163621144\\
0.38	0.00359244802880124\\
0.42	0.00319256618835287\\
0.46	0.00286309541875749\\
0.5	0.00258720279857029\\
0.54	0.0023530405044477\\
0.58	0.00215202666559296\\
0.62	0.00197779170010181\\
0.66	0.00182550780154909\\
0.7	0.00169144830659468\\
0.74	0.00157268995102281\\
0.78	0.0014669067118484\\
0.82	0.00137222395768493\\
0.86	0.00128711327412017\\
0.9	0.00121031531829372\\
0.94	0.00114078236729324\\
0.98	0.00107763494984047\\
};
\addlegendentry{$\hat{\mathtt{I}}_{1}(z)$}

\addplot [color=black, dashdotted, line width=1.0pt]
  table[row sep=crcr]{%
0.1	0.0105474897229063\\
0.14	0.0089482725799505\\
0.18	0.00751426440753899\\
0.22	0.00633278172032493\\
0.26	0.00539182895093805\\
0.3	0.00464733776439886\\
0.34	0.00405415440607391\\
0.38	0.00357546368301102\\
0.42	0.00318355700979503\\
0.46	0.00285813454165485\\
0.5	0.0025843755581385\\
0.54	0.002351377716053\\
0.58	0.0021510201883171\\
0.62	0.00197716624087504\\
0.66	0.00182510957285153\\
0.7	0.00169118882410668\\
0.74	0.0015725166873458\\
0.78	0.00146678720152878\\
0.82	0.00137213679555667\\
0.86	0.00128704270678001\\
0.9	0.00121024773237223\\
0.94	0.00114070379523714\\
0.98	0.00107752914339315\\
};
\addlegendentry{Number of interferers $(n) = 4$}

\addplot [color=black, line width=1.0pt, draw=none, mark=o, mark size=3pt]
  table[row sep=crcr]{%
0.1	0.0143561539893994\\
0.14	0.0106177495057785\\
0.16	0.00929548470785414\\
0.2	0.00736378223129982\\
0.24	0.00604431463331208\\
0.28	0.00509450768414893\\
0.32	0.00438068763770346\\
0.36	0.00382561927504751\\
0.4	0.00338217148275891\\
0.44	0.00302011769029225\\
0.48	0.00271922943672549\\
0.52	0.00246547351951668\\
0.56	0.00224881373413174\\
0.6	0.00206188369448185\\
0.64	0.00189915444254775\\
0.68	0.00175639472559139\\
0.72	0.00163031069341522\\
0.76	0.00151829911649289\\
0.8	0.00141827446780668\\
0.84	0.00132854528343034\\
0.88	0.00124772415365379\\
0.92	0.0011746611463856\\
0.96	0.00110839387108003\\
1	0.00104810956741082\\
};
\addlegendentry{Number of interferers $(n) = 10$}

\addplot [color=black, draw=none, mark=asterisk, mark size=3pt]
  table[row sep=crcr]{%
0.1	0.0154304129460149\\
0.14	0.0108359273407387\\
0.18	0.00828410981573109\\
0.22	0.00666103305848401\\
0.26	0.00553820150361882\\
0.3	0.00471564427798455\\
0.34	0.00408749162100866\\
0.38	0.0035924480225362\\
0.42	0.00319256618553183\\
0.46	0.00286309541737342\\
0.5	0.00258720279764997\\
0.54	0.00235304050239371\\
0.58	0.00215202665588482\\
0.62	0.00197779165763252\\
0.66	0.0018255076460979\\
0.7	0.00169144782110479\\
0.74	0.00157268862449983\\
0.78	0.00146690347360019\\
0.82	0.00137221677404748\\
0.86	0.00128709859117316\\
0.9	0.00121028735226879\\
0.94	0.00114073226469994\\
0.98	0.00107754985776401\\
};
\addlegendentry{Number of interferers $(n) = 20$}

\addplot [color=black, draw=none, mark=+, mark size=3pt]
  table[row sep=crcr]{%
0.1	0.0154306398237468\\
0.14	0.0108359441751488\\
0.18	0.008284112157774\\
0.22	0.00666103353804874\\
0.26	0.00553820163106511\\
0.3	0.00471564431883697\\
0.34	0.00408749163608911\\
0.38	0.00359244802875092\\
0.42	0.0031925661883292\\
0.46	0.00286309541872696\\
0.5	0.00258720279834561\\
0.54	0.00235304050276996\\
0.58	0.00215202665609743\\
0.62	0.00197779165775731\\
0.66	0.00182550764617362\\
0.7	0.0016914478211521\\
0.74	0.00157268862453018\\
0.78	0.00146690347362011\\
0.82	0.00137221677406084\\
0.86	0.00128709859118229\\
0.9	0.00121028735227513\\
0.94	0.00114073226470442\\
0.98	0.00107754985776722\\
};
\addlegendentry{Number of interferers $(n) = 40$}

\end{axis}
\end{tikzpicture}%
\caption{\textit{(One Dimension Model)} Here the variation of interference $\mathbb{I}_{n}(z)$, is drawn with respect to a linear variation of the inter-LED spacing $a$ for different number interferers ($n$) in the network. The graph for the proposed interference expression $\hat{\mathtt{I}}_{1}(z)$ is also drawn. We consider height $h$ of the LEDs as $2.5$m, the half-power-semi-angle (HPSA) $\theta_{h}$ of the LED as $\frac{\pi}{3}$ radians and the position $z=0.25$m.}
\label{one5a1}  
\end{figure}
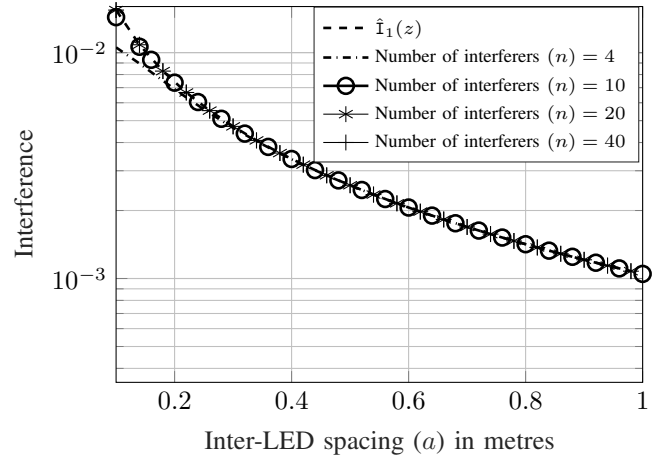 

\begin{figure}[ht]
\centering
\begin{tikzpicture}
\draw[help lines] (0,0);
\begin{axis}[%
width=7cm,
height=5cm,
scale only axis,
xmin=2.5,
xmax=5,
xlabel style={font=\color{white!15!black}},
xlabel={Height ($h$) of the LED installation in metres},
ymode=log,
ymin=1e-05,
ymax=0.01,
yminorticks=true,
ylabel style={font=\color{white!15!black}, align=center},
ylabel={Interference},
axis background/.style={fill=white},
title style={font=\bfseries, align=center},
xmajorgrids,
ymajorgrids,
yminorgrids,
legend style={font=\fontsize{7}{5}\selectfont,at={(0.376,0.551)}, anchor=south west, legend cell align=left, align=left, draw=white!15!black}
]
\addplot [color=black, dashed, line width=1.0pt]
  table[row sep=crcr]{%
2.5	0.00258720279835122\\
2.6	0.00198309116153975\\
2.7	0.00153489914571403\\
2.8	0.00119878915246052\\
2.9	0.000944198477099194\\
3	0.000749548600461203\\
3.1	0.000599423909529509\\
3.2	0.000482689878483216\\
3.3	0.000391221013271258\\
3.4	0.000319030179883406\\
3.5	0.00026166565174971\\
3.6	0.000215788943006176\\
3.7	0.000178876148885675\\
3.8	0.000149004589844549\\
3.9	0.000124698984077566\\
4	0.000104819568330052\\
4.1	8.84800528023656e-05\\
4.2	7.49869808767264e-05\\
4.3	6.37945741130931e-05\\
4.4	5.44708688742534e-05\\
4.5	4.66721487941069e-05\\
4.6	4.01235159624639e-05\\
4.7	3.46040357857947e-05\\
4.8	2.99353118253406e-05\\
4.9	2.59726490342095e-05\\
5	2.25981820252323e-05\\
};
\addlegendentry{$\hat{\mathtt{I}}_{1}(z)$}

\addplot [color=black, dashdotted, line width=1.0pt]
  table[row sep=crcr]{%
2.5	0.00251073034781809\\
2.6	0.00191508598585208\\
2.7	0.00147443372801194\\
2.8	0.00114502772550451\\
2.9	0.000896390974212617\\
3	0.000707024142689262\\
3.1	0.000561584141191273\\
3.2	0.000449002378637945\\
3.3	0.000361212942284716\\
3.4	0.000292282266428237\\
3.5	0.00023780658411884\\
3.6	0.000194490197469524\\
3.7	0.000159847265033928\\
3.8	0.000131988886717298\\
3.9	0.000109469699894576\\
4	9.1176399610267e-05\\
4.1	7.62460645710232e-05\\
4.2	6.40058577346179e-05\\
4.3	5.39281815534167e-05\\
4.4	4.55970949043933e-05\\
4.5	3.86829972905183e-05\\
4.6	3.29234249876371e-05\\
4.7	2.81083961304358e-05\\
4.8	2.40691631660226e-05\\
4.9	2.06695332191469e-05\\
5	1.77991350530624e-05\\
};
\addlegendentry{Number of interferers $(n) = 4$}

\addplot [color=black, line width=1.0pt, draw=none, mark=o, mark size=3pt]
  table[row sep=crcr]{%
2.5	0.0025843755581385\\
2.6	0.00198039518863349\\
2.7	0.00153233121273368\\
2.8	0.00119634575674219\\
2.9	0.000941875889884441\\
3	0.000747342913097121\\
3.1	0.00059733107641311\\
3.2	0.00048070575657739\\
3.3	0.000389341398103636\\
3.4	0.000317250838012163\\
3.5	0.000259982349742006\\
3.6	0.000214197473025786\\
3.7	0.000177372351026986\\
3.8	0.000147584371381337\\
3.9	0.000123358335817701\\
4	0.000103554578282498\\
4.1	8.72869173716002e-05\\
4.2	7.38620138029025e-05\\
4.3	6.27342133749674e-05\\
4.4	5.34716817668808e-05\\
4.5	4.57308353854122e-05\\
4.6	3.92369111227425e-05\\
4.7	3.37691099725811e-05\\
4.8	2.9149170793458e-05\\
4.9	2.52325326276172e-05\\
5	2.19014621938979e-05\\
};
\addlegendentry{Number of interferers $(n) = 10$}

\addplot [color=black, draw=none, mark=asterisk, mark size=3pt]
  table[row sep=crcr]{%
2.5	0.00258720271348607\\
2.6	0.00198309107688368\\
2.7	0.00153489906127452\\
2.8	0.00119878906824497\\
2.9	0.00094419839311493\\
3	0.000749548516715458\\
3.1	0.000599423826029433\\
3.2	0.000482689795235869\\
3.3	0.000391220930283611\\
3.4	0.000319030097162335\\
3.5	0.000261665569301996\\
3.6	0.000215788860838505\\
3.7	0.000178876067004632\\
3.8	0.000149004508256619\\
3.9	0.000124698902789132\\
4	0.000104819487347392\\
4.1	8.84799721316513e-05\\
4.2	7.49869005240238e-05\\
4.3	6.37944940843587e-05\\
4.4	5.44707891753336e-05\\
4.5	4.66720694307366e-05\\
4.6	4.01234369402653e-05\\
4.7	3.46039571102762e-05\\
4.8	2.99352335018958e-05\\
4.9	2.59725710681159e-05\\
5	2.25981044216505e-05\\
};
\addlegendentry{Number of interferers $(n) = 20$}

\addplot [color=black, draw=none, mark=+, mark size=3pt]
  table[row sep=crcr]{%
2.5	0.00258720279764997\\
2.6	0.00198309116083894\\
2.7	0.00153489914501368\\
2.8	0.00119878915176064\\
2.9	0.00094419847639981\\
3	0.000749548599762325\\
3.1	0.000599423908831155\\
3.2	0.000482689877785401\\
3.3	0.000391221012574\\
3.4	0.000319030179186721\\
3.5	0.000261665651053614\\
3.6	0.000215788942310686\\
3.7	0.000178876148190808\\
3.8	0.00014900458915032\\
3.9	0.000124698983383992\\
4	0.000104819567637149\\
4.1	8.84800521101491e-05\\
4.2	7.4986980185213e-05\\
4.3	6.37945734222987e-05\\
4.4	5.44708681841939e-05\\
4.5	4.66721481047981e-05\\
4.6	4.01235152739216e-05\\
4.7	3.46040350980345e-05\\
4.8	2.99353111383783e-05\\
4.9	2.59726483480605e-05\\
5	2.2598181339912e-05\\
};
\addlegendentry{Number of interferers $(n) = 40$}

\end{axis}
\end{tikzpicture}%
\caption{\textit{(One Dimension Model)} Here the variation of interference $\mathbb{I}_{n}(z)$, is drawn with respect to a linear variation of the height $h$ of installation of the LED for different number interferers $(n)$ in the network. The graph for the proposed interference expression $\hat{\mathtt{I}}_{1}(z)$ is also drawn. We consider $a=0.5$m, the half-power-semi-angle (HPSA) $\theta_{h}$ of the LED as $\frac{\pi}{3}$ radians and the position $z$ of the receiver photodiode (PD) at half the attocell length $\frac{a}{2}$.}
\label{one5h1}
\end{figure}
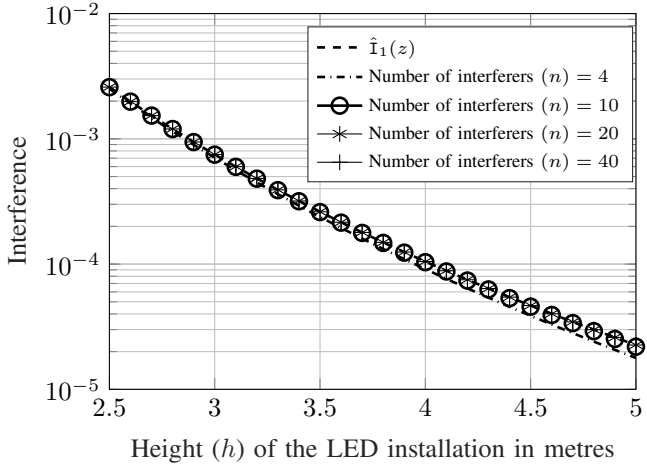 

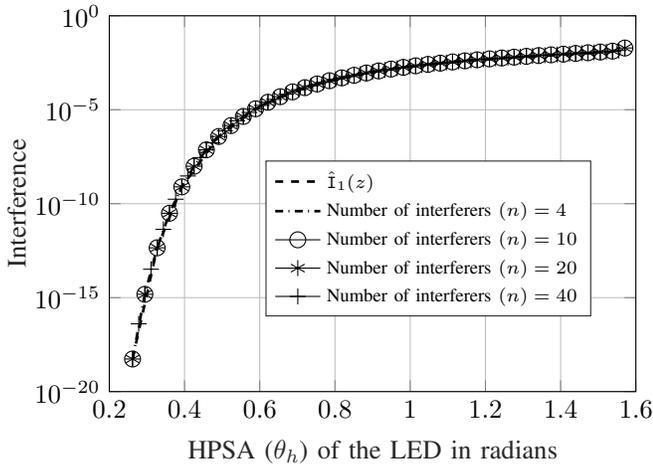
\begin{figure}[ht]
\centering
\begin{tikzpicture}
\begin{axis}[%
width=7cm,
height=5cm,
scale only axis,
xmin=0.2,
xmax=1.6,
xlabel style={font=\color{white!15!black}},
xlabel={HPSA ($\theta_{h}$) of the LED in radians},
ymode=log,
ymin=1e-20,
ymax=1,
yminorticks=true,
ylabel style={font=\color{white!15!black}},
ylabel={Interference},
axis background/.style={fill=white},
title style={font=\bfseries},
title={},
xmajorgrids,
ymajorgrids,
yminorgrids,
legend style={font=\fontsize{7}{5}\selectfont,at={(0.2966,0.200)}, anchor=south west, legend cell align=left, align=left, draw=white!15!black}
]
\addplot [color=black, dashed, line width=1.0pt]
  table[row sep=crcr]{%
0.261799387799149	5.42594176485996e-19\\
0.327249234748937	4.49530039353952e-13\\
0.392699081698724	7.85960969938092e-10\\
0.490873852123405	3.76863622651642e-07\\
0.589048622548086	1.12563409527614e-05\\
0.67086093123532	6.73730453229294e-05\\
0.752673239922555	0.00023553539354545\\
0.818123086872342	0.000498856315692319\\
0.899935395559576	0.00103038871983863\\
0.98174770424681	0.0018024107445792\\
1.06356001293404	0.00280512306249874\\
1.12900985988383	0.00375719554519806\\
1.21082216857107	0.00511365446178422\\
1.2926344772583	0.00664290618539656\\
1.37444678594553	0.0083621983573636\\
1.43989663289532	0.00993550007585437\\
1.50534647984511	0.0118736896792379\\
1.55443386505745	0.0141805247116131\\
1.5707963267949	0.0193839295109491\\
};
\addlegendentry{$\hat{\mathtt{I}}_{1}(z)$}

\addplot [color=black, dashdotted, line width=1.0pt]
  table[row sep=crcr]{%
0.261799387799149	5.42594170315058e-19\\
0.294524311274043	1.50889933670242e-15\\
0.327249234748937	4.49528706743244e-13\\
0.35997415822383	3.0984295533413e-11\\
0.392699081698724	7.85912690559892e-10\\
0.425424005173618	9.83711396793448e-09\\
0.458148928648511	7.36863607290185e-08\\
0.490873852123405	3.76577388237086e-07\\
0.523598775598299	1.43899577753915e-06\\
0.556323699073193	4.39044077595417e-06\\
0.589048622548086	1.12218222914358e-05\\
0.62177354602298	2.49050343877873e-05\\
0.654498469497873	4.93155117352412e-05\\
0.687223392972767	8.89696292687109e-05\\
0.719948316447661	0.000148646779938162\\
0.752673239922555	0.000232985407373526\\
0.785398163397448	0.000346128829233071\\
0.818123086872342	0.000491467162088964\\
0.850848010347236	0.000671491206138808\\
0.883572933822129	0.000887751234576427\\
0.916297857297023	0.00114090054961481\\
0.949022780771917	0.00143079904305651\\
0.98174770424681	0.00175665311970776\\
1.0144726277217	0.00211717256946051\\
1.0471975511966	0.0025107303478181\\
1.07992247467149	0.0029355165562124\\
1.11264739814639	0.00338968265678517\\
1.14537232162128	0.00387147603241312\\
1.17809724509617	0.00437936866849763\\
1.21082216857107	0.00491218754382424\\
1.24354709204596	0.00546925923618013\\
1.27627201552085	0.00605058901823846\\
1.30899693899575	0.00665710883419428\\
1.34172186247064	0.00729105676649084\\
1.37444678594553	0.00795661198698998\\
1.40717170942043	0.00866105607340532\\
1.43989663289532	0.00941712867920694\\
1.47262155637022	0.0102485154831973\\
1.50534647984511	0.0112056306784905\\
1.53807140332	0.0124324526350411\\
1.5707963267949	0.0180414341147544\\
};
\addlegendentry{Number of interferers $(n) = 4$}

\addplot [color=black, draw=none, mark=o, mark size=3pt]
  table[row sep=crcr]{%
0.261799387799149	5.42594176486317e-19\\
0.294524311274043	1.50889977703472e-15\\
0.327249234748937	4.49530039352806e-13\\
0.35997415822383	3.09848077374093e-11\\
0.392699081698724	7.8596096842886e-10\\
0.425424005173618	9.83880230583391e-09\\
0.458148928648511	7.37150953295169e-08\\
0.490873852123405	3.76863451554735e-07\\
0.523598775598299	1.44088979497328e-06\\
0.556323699073193	4.39958361129509e-06\\
0.589048622548086	1.12562380651398e-05\\
0.62177354602298	2.50112531963384e-05\\
0.654498469497873	4.95947319504305e-05\\
0.687223392972767	8.96134036702095e-05\\
0.719948316447661	0.000149978805667287\\
0.752673239922555	0.000235503962036752\\
0.785398163397448	0.000350543925874361\\
0.818123086872342	0.000498727536784687\\
0.850848010347236	0.00068279823359791\\
0.883572933822129	0.00090455909198066\\
0.916297857297023	0.00116490387236538\\
0.949022780771917	0.00146391065177555\\
0.98174770424681	0.00180097509331082\\
1.0144726277217	0.0021749640944435\\
1.0471975511966	0.0025843755581385\\
1.07992247467149	0.00302749516944942\\
1.11264739814639	0.00350254576672861\\
1.14537232162128	0.00400782906252933\\
1.17809724509617	0.00454186333209677\\
1.21082216857107	0.00510352480720475\\
1.24354709204596	0.00569220589972903\\
1.27627201552085	0.00630801190225065\\
1.30899693899575	0.00695203325571574\\
1.34172186247064	0.0076267612710764\\
1.37444678594553	0.00833678208949232\\
1.40717170942043	0.00909004364202938\\
1.43989663289532	0.00990042277628961\\
1.47262155637022	0.0107937025276209\\
1.50534647984511	0.0118247586956112\\
1.53807140332	0.0131503474567452\\
1.5707963267949	0.0192618379133941\\
};
\addlegendentry{Number of interferers $(n) = 10$}

\addplot [color=black, draw=none, mark=asterisk, mark size=3pt]
  table[row sep=crcr]{%
0.261799387799149	5.42594176486317e-19\\
0.294524311274043	1.50889977703475e-15\\
0.327249234748937	4.49530039353951e-13\\
0.35997415822383	3.09848077407933e-11\\
0.392699081698724	7.85960969938092e-10\\
0.425424005173618	9.83880248310225e-09\\
0.458148928648511	7.37151032288097e-08\\
0.490873852123405	3.76863622651642e-07\\
0.523598775598299	1.44089193674461e-06\\
0.556323699073193	4.39960115077983e-06\\
0.589048622548086	1.12563409527605e-05\\
0.62177354602298	2.50117157711388e-05\\
0.654498469497873	4.95964097578271e-05\\
0.687223392972767	8.96185108628758e-05\\
0.719948316447661	0.000149992260281621\\
0.752673239922555	0.000235535393518694\\
0.785398163397448	0.000350610313971676\\
0.818123086872342	0.000498856315354853\\
0.850848010347236	0.000683030591401511\\
0.883572933822129	0.000904953162232047\\
0.916297857297023	0.00116553752710108\\
0.949022780771917	0.00146488369100166\\
0.98174770424681	0.00180241071925658\\
1.0144726277217	0.00217700966594172\\
1.0471975511966	0.00258720271348607\\
1.07992247467149	0.00303129948054459\\
1.11264739814639	0.00350754614278577\\
1.14537232162128	0.00401426716413144\\
1.17809724509617	0.00455000329646212\\
1.21082216857107	0.00511365363910946\\
1.24354709204596	0.00570463497332142\\
1.27627201552085	0.00632308018674242\\
1.30899693899575	0.00697011319730441\\
1.34172186247064	0.00764826889055077\\
1.37444678594553	0.00836219415009751\\
1.40717170942043	0.00911992763069733\\
1.43989663289532	0.00993549263369783\\
1.47262155637022	0.0108349340358125\\
1.50534647984511	0.0118736762727974\\
1.53807140332	0.0132100397061788\\
1.5707963267949	0.0193838622496067\\
};
\addlegendentry{Number of interferers $(n) = 20$}

\addplot [color=black, draw=none, mark=+, mark size=3pt]
  table[row sep=crcr]{%
0.261799387799149	5.42594176486317e-19\\
0.278161849536596	4.04145214056654e-17\\
0.294524311274043	1.50889977703475e-15\\
0.31088677301149	3.25067304315718e-14\\
0.327249234748937	4.49530039353951e-13\\
0.343611696486384	4.33095730302587e-12\\
0.35997415822383	3.09848077407933e-11\\
0.376336619961277	1.73096307434838e-10\\
0.392699081698724	7.85960969938092e-10\\
0.409061543436171	2.99566793340904e-09\\
0.425424005173618	9.83880248310225e-09\\
0.441786466911065	2.84489046924108e-08\\
0.458148928648511	7.37151032288097e-08\\
0.474511390385958	1.73704628770216e-07\\
0.490873852123405	3.76863622651642e-07\\
0.507236313860852	7.60648897706194e-07\\
0.523598775598299	1.44089193674461e-06\\
0.539961237335746	2.58094841607562e-06\\
0.556323699073193	4.39960115077986e-06\\
0.572686160810639	7.17677328167437e-06\\
0.589048622548086	1.12563409527614e-05\\
0.605411084285533	1.70456539768714e-05\\
0.62177354602298	2.5011715771152e-05\\
0.638136007760427	3.56742856119568e-05\\
0.654498469497873	4.95964097579627e-05\\
0.67086093123532	6.73730453229291e-05\\
0.687223392972767	8.96185108638875e-05\\
0.703585854710214	0.00011695349137274\\
0.719948316447661	0.000149992260287421\\
0.736310778185108	0.000189330676982902\\
0.752673239922555	0.000235535393545407\\
0.769035701660001	0.000289134574990126\\
0.785398163397448	0.000350610314074116\\
0.801760625134895	0.000420392813184027\\
0.818123086872342	0.000498856315691414\\
0.834485548609789	0.000586316699161781\\
0.850848010347236	0.000683030592371196\\
0.867210472084682	0.000789195845413562\\
0.883572933822129	0.000904953164729319\\
0.899935395559576	0.00103038871982152\\
0.916297857297023	0.00116553753294132\\
0.93266031903447	0.00131038747454132\\
0.949022780771917	0.00146488370356894\\
0.965385242509364	0.00162893341084882\\
0.98174770424681	0.00180241074441553\\
0.998110165984257	0.00198516181660705\\
1.0144726277217	0.00217700971323915\\
1.03083508945915	0.00237775944475764\\
1.0471975511966	0.00258720279764998\\
1.06356001293404	0.0028051230615282\\
1.07992247467149	0.00303129962327399\\
1.09628493640894	0.00326551243471098\\
1.11264739814639	0.0035075463748203\\
1.12900985988383	0.00375719554206272\\
1.14537232162128	0.00401426752759885\\
1.16173478335873	0.00427858773699383\\
1.17809724509617	0.0045500038475296\\
1.19445970683362	0.00482839051210121\\
1.21082216857107	0.00511365445101724\\
1.22718463030851	0.00540574011292453\\
1.24354709204596	0.00570463613996458\\
1.25990955378341	0.00601038294668957\\
1.27627201552085	0.006323081827111\\
1.2926344772583	0.00664290615479862\\
1.30899693899575	0.00697011546133947\\
1.32535940073319	0.00730507350782342\\
1.34172186247064	0.00764827196742708\\
1.35808432420809	0.0080003621261512\\
1.37444678594553	0.00836219828084584\\
1.39080924768298	0.00873489863192062\\
1.40717170942043	0.00911993312906264\\
1.42353417115788	0.00951925433531499\\
1.43989663289532	0.00993549992405665\\
1.45625909463277	0.0103723207659818\\
1.47262155637022	0.0108349437314009\\
1.48898401810766	0.0113312103200686\\
1.50534647984511	0.0118736893714893\\
1.52170894158256	0.0124846066338377\\
1.53807140332	0.0132100582941887\\
1.55443386505745	0.0141805240902788\\
1.5707963267949	0.019383927378817\\
};
\addlegendentry{Number of interferers $(n) = 40$}

\end{axis}
\end{tikzpicture}%
\caption{\textit{(One Dimension Model)} Here the variation of interference $\mathbb{I}_{n}(z)$, is drawn with respect to a linear variation of the half-power-semi-angle (HPSA) $\theta_{h}$ of the LED for different number interferers $(n)$ in the network. The graph for the proposed interference expression $\hat{\mathtt{I}}_{1}(z)$ is also drawn. We consider the attocell length $a=0.5$m, the height $h$ of the LED as $2.5$m and the position $z$ of the receiver photodiode (PD) at half the attocell length $\frac{a}{2}$.}
\label{one5t1}
\end{figure} 

\begin{figure}[ht]
\centering
\begin{tikzpicture}
\begin{axis}[%
width=7cm,
height=5cm,
scale only axis,
xmin=0,
xmax=0.25,
xlabel style={font=\color{white!15!black}},
xlabel={Distance ($z$) of the PD from origin in metres},
ymin=0.002555,
ymax=0.00259,
ylabel style={font=\color{white!15!black}},
ylabel={Interference},
axis background/.style={fill=white},
title style={font=\bfseries},
title={},
xmajorgrids,
ymajorgrids,
legend style={font=\fontsize{7}{5}\selectfont,at={(0.056,0.581)}, anchor=south west, legend cell align=left, align=left, draw=white!15!black}
]
\addplot [color=black, dashed, line width=1.0pt]
  table[row sep=crcr]{%
0	0.00256163087764042\\
0.011	0.00256168162625908\\
0.022	0.00256183384264463\\
0.033	0.00256208743841991\\
0.044	0.00256244226640371\\
0.055	0.00256289812078149\\
0.066	0.00256345473734409\\
0.077	0.00256411179379372\\
0.088	0.00256486891011693\\
0.099	0.00256572564902354\\
0.11	0.00256668151645102\\
0.122	0.00256783668850457\\
0.133	0.00256899798124273\\
0.144	0.00257025652292118\\
0.155	0.00257161159382064\\
0.166	0.00257306242099021\\
0.177	0.00257460817912469\\
0.188	0.00257624799149811\\
0.199	0.00257798093095161\\
0.211	0.00257997647260739\\
0.222	0.00258190091996064\\
0.233	0.00258391532122353\\
0.244	0.00258601855428932\\
};
\addlegendentry{$\hat{\mathtt{I}}_{1}(z)$}

\addplot [color=black, dashdotted, line width=1.0pt]
  table[row sep=crcr]{%
0	0.00255892120439153\\
0.011	0.00255897172774742\\
0.022	0.00255912326829029\\
0.033	0.00255937573747998\\
0.044	0.0025597289878635\\
0.055	0.00256018281324573\\
0.066	0.00256073694892795\\
0.077	0.00256139107201377\\
0.088	0.00256214480178179\\
0.099	0.00256299770012441\\
0.11	0.0025639492720519\\
0.122	0.00256509923803109\\
0.133	0.00256625528055797\\
0.144	0.00256750811352601\\
0.155	0.00256885701583629\\
0.166	0.00257030121304667\\
0.177	0.00257183987824866\\
0.188	0.00257347213300037\\
0.199	0.00257519704831397\\
0.211	0.00257718330201467\\
0.222	0.00257909874336891\\
0.233	0.00258110366573978\\
0.244	0.00258319694472242\\
};
\addlegendentry{Number of interferers $(n) = 4$}

\addplot [color=black, line width=1.0pt, draw=none, mark=o, mark size=3pt]
  table[row sep=crcr]{%
0	0.00256163079300338\\
0.011	0.00256168154162161\\
0.022	0.00256183375800583\\
0.033	0.00256208735377891\\
0.044	0.00256244218175961\\
0.055	0.00256289803613342\\
0.066	0.00256345465269116\\
0.077	0.00256411170913506\\
0.088	0.00256486882545165\\
0.099	0.00256572556435076\\
0.11	0.00256668143176985\\
0.122	0.00256783660381325\\
0.133	0.00256899789654117\\
0.144	0.0025702564382085\\
0.155	0.00257161150909596\\
0.166	0.00257306233625264\\
0.177	0.00257460809437336\\
0.188	0.00257624790673212\\
0.199	0.00257798084617008\\
0.211	0.0025799763878079\\
0.222	0.00258190083514376\\
0.233	0.00258391523638837\\
0.244	0.002586018469435\\
};
\addlegendentry{Number of interferers $(n) = 10$}

\addplot [color=black, draw=none, mark=asterisk, mark size=3pt]
  table[row sep=crcr]{%
0	0.00256163087693966\\
0.011	0.00256168162555832\\
0.022	0.00256183384194386\\
0.033	0.00256208743771914\\
0.044	0.00256244226570293\\
0.055	0.00256289812008071\\
0.066	0.00256345473664329\\
0.077	0.00256411179309292\\
0.088	0.00256486890941611\\
0.099	0.0025657256483227\\
0.11	0.00256668151575016\\
0.122	0.0025678366878037\\
0.133	0.00256899798054183\\
0.144	0.00257025652222026\\
0.155	0.0025716115931197\\
0.166	0.00257306242028924\\
0.177	0.00257460817842369\\
0.188	0.00257624799079708\\
0.199	0.00257798093025055\\
0.211	0.00257997647190629\\
0.222	0.0025819009192595\\
0.233	0.00258391532052235\\
0.244	0.0025860185535881\\
};
\addlegendentry{Number of interferers $(n) = 20$}

\addplot [color=black, draw=none, mark=+, mark size=3pt]
  table[row sep=crcr]{%
0	0.00256163087763481\\
0.01	0.00256167281899427\\
0.02	0.00256179862294325\\
0.03	0.0025620082291129\\
0.04	0.0025623015369528\\
0.05	0.00256267840582748\\
0.06	0.00256313865515114\\
0.07	0.00256368206456068\\
0.08	0.00256430837412653\\
0.09	0.00256501728460106\\
0.11	0.00256668151644541\\
0.12	0.00256763604548294\\
0.13	0.00256867159151722\\
0.14	0.00256978766372057\\
0.15	0.00257098373420068\\
0.16	0.002572259238498\\
0.17	0.00257361357611598\\
0.18	0.00257504611108343\\
0.19	0.00257655617254809\\
0.2	0.00257814305540073\\
0.21	0.00257980602092841\\
0.22	0.00258154429749642\\
0.23	0.00258335708125746\\
0.24	0.00258524353688729\\
0.25	0.00258720279834561\\
};
\addlegendentry{Number of interferers $(n) = 40$}
\end{axis}
\end{tikzpicture}%
\caption{\textit{(One Dimension Model)} Here the variation of interference $\mathbb{I}_{n}(z)$, is drawn with respect to a linear variation of the position $z$ of the receiver photodiode (PD) inside the attocell for different number interferers ($n$) in the network. The graph for the proposed interference expression $\hat{\mathtt{I}}_{1}(z)$ is also drawn. We consider $a=0.5$m, the half-power-semi-angle (HPSA) $\theta_{h}$ of the LED as $\frac{\pi}{3}$ radians and the height $h$ of LED as $2.5$m.}
\label{one5z1}
\end{figure}
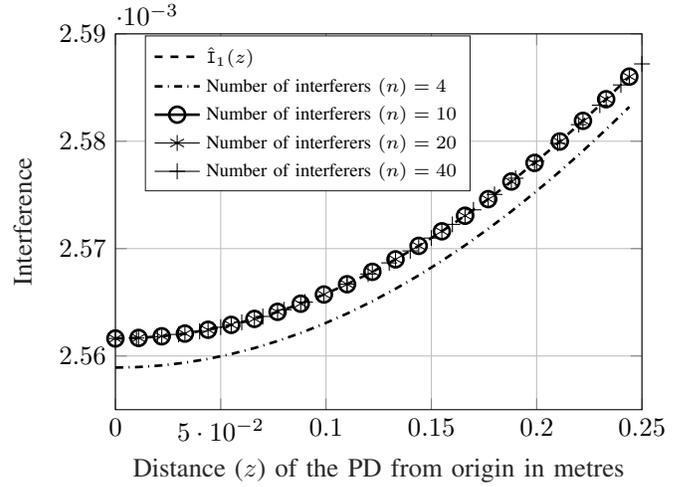 

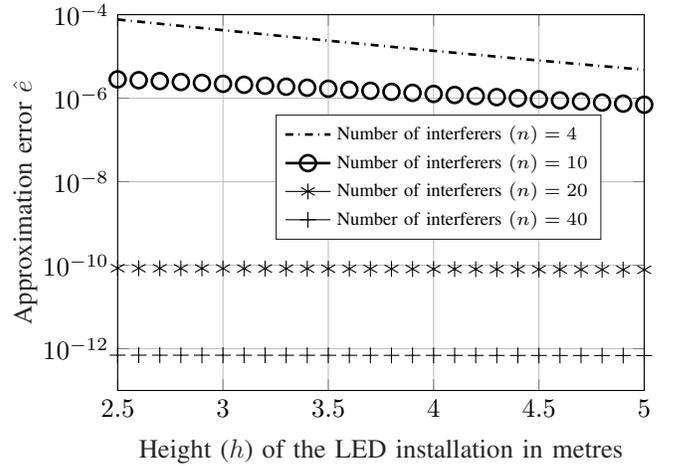
\begin{figure}[ht]
\centering
\begin{tikzpicture}
\begin{axis}[%
width=7cm,
height=5cm,
scale only axis,
xmin=2.5,
xmax=5,
xlabel style={font=\color{white!15!black}},
xlabel={Height ($h$) of the LED installation in metres},
ymode=log,
ymin=1e-13,
ymax=0.0001,
yminorticks=true,
ylabel style={font=\color{white!15!black}},
ylabel={Approximation error $\hat{e}$},
axis background/.style={fill=white},
title style={font=\bfseries},
title={},
xmajorgrids,
ymajorgrids,
yminorgrids,
legend style={font=\fontsize{7}{5}\selectfont,at={(0.30,0.401)}, anchor=south west, legend cell align=left, align=left, draw=white!15!black}
]  
\addplot [color=black, dashdotted, line width=1.0pt]
  table[row sep=crcr]{%
2.5	7.64724505331244e-05\\
2.6	6.8005175687664e-05\\
2.7	6.04654177020899e-05\\
2.8	5.37614269560023e-05\\
2.9	4.78075028865774e-05\\
3	4.25244577719415e-05\\
3.1	3.78397683382366e-05\\
3.2	3.36874998452714e-05\\
3.3	3.00080709865415e-05\\
3.4	2.67479134551687e-05\\
3.5	2.38590676308697e-05\\
3.6	2.12987455366524e-05\\
3.7	1.90288838517467e-05\\
3.8	1.70157031272512e-05\\
3.9	1.52292841829903e-05\\
4	1.36431687197854e-05\\
4.1	1.22339882313424e-05\\
4.2	1.09811231421085e-05\\
4.3	9.86639255967634e-06\\
4.4	8.87377396986015e-06\\
4.5	7.98915150358861e-06\\
4.6	7.20009097482683e-06\\
4.7	6.49563965535892e-06\\
4.8	5.86614865931805e-06\\
4.9	5.30311581506257e-06\\
5	4.79904697216989e-06\\
};
\addlegendentry{Number of interferers $(n) = 4$}

\addplot [color=black, line width=1.0pt, draw=none, mark=o, mark size=3pt]
  table[row sep=crcr]{%
2.5	2.82724021271729e-06\\
2.6	2.6959729062569e-06\\
2.7	2.56793298034508e-06\\
2.8	2.4433957183218e-06\\
2.9	2.3225872147533e-06\\
3	2.20568736408169e-06\\
3.1	2.09283311639941e-06\\
3.2	1.98412190582627e-06\\
3.3	1.87961516762158e-06\\
3.4	1.77934187124278e-06\\
3.5	1.68330200770367e-06\\
3.6	1.59146998038979e-06\\
3.7	1.50379785868909e-06\\
3.8	1.42021846321212e-06\\
3.9	1.34064825986594e-06\\
4	1.26499004755401e-06\\
4.1	1.1931354307654e-06\\
4.2	1.12496707382395e-06\\
4.3	1.06036073812572e-06\\
4.4	9.99187107372641e-07\\
4.5	9.41313408694738e-07\\
4.6	8.86604839721437e-07\\
4.7	8.34925813213581e-07\\
4.8	7.86141031882633e-07\\
4.9	7.40116406592322e-07\\
5	6.96719831334413e-07\\
};
\addlegendentry{Number of interferers $(n) = 10$}

\addplot [color=black, draw=none, mark=asterisk, mark size=3pt]
  table[row sep=crcr]{%
2.5	8.48651478951856e-11\\
2.6	8.46560699145626e-11\\
2.7	8.44395096379069e-11\\
2.8	8.42155479467005e-11\\
2.9	8.39842644213828e-11\\
3	8.3745744930766e-11\\
3.1	8.35000763194441e-11\\
3.2	8.3247346949894e-11\\
3.3	8.29876474072073e-11\\
3.4	8.27210712580316e-11\\
3.5	8.24477135326872e-11\\
3.6	8.21676712672685e-11\\
3.7	8.18810432054885e-11\\
3.8	8.1587930123939e-11\\
3.9	8.12884345881455e-11\\
4	8.09826607357268e-11\\
4.1	8.06707143306048e-11\\
4.2	8.03527026003743e-11\\
4.3	8.00287343989334e-11\\
4.4	7.96989198405653e-11\\
4.5	7.93633703541482e-11\\
4.6	7.90221985815117e-11\\
4.7	7.86755184451989e-11\\
4.8	7.83234448232064e-11\\
4.9	7.79660936506277e-11\\
5	7.760358178074e-11\\
};
\addlegendentry{Number of interferers $(n) = 20$}

\addplot [color=black, draw=none, mark=+, mark size=3pt]
  table[row sep=crcr]{%
2.5	7.01244617928864e-13\\
2.6	7.00804431846835e-13\\
2.7	7.00347982732219e-13\\
2.8	6.99875270585015e-13\\
2.9	6.99384560681748e-13\\
3	6.98877696166111e-13\\
3.1	6.98354568617887e-13\\
3.2	6.97814852776424e-13\\
3.3	6.97257898120418e-13\\
3.4	6.96684788852042e-13\\
3.5	6.96095524971296e-13\\
3.6	6.95489618587203e-13\\
3.7	6.94867476275576e-13\\
3.8	6.94228935406091e-13\\
3.9	6.93574375449507e-13\\
4	6.92903444040119e-13\\
4.1	6.92216466438578e-13\\
4.2	6.91513429092357e-13\\
4.3	6.90794399764091e-13\\
4.4	6.90059581741689e-13\\
4.5	6.89308792066034e-13\\
4.6	6.88542308563932e-13\\
4.7	6.87760151564173e-13\\
4.8	6.86962348171814e-13\\
4.9	6.86149047464651e-13\\
5	6.85320246054553e-13\\
};
\addlegendentry{Number of interferers $(n) = 40$}

\end{axis}
\end{tikzpicture}%
\caption{\textit{(One Dimension Model)} Here the variation of interference approximation error $\hat{e}=|\mathbb{I}_{n}(z)-\hat{\mathtt{I}}_{1}(z)|$ is drawn for a linear variation of the height $h$ of installation of the LED for different number interferers ($n$) in the network. We consider $a=0.5$m, the half-power-semi-angle (HPSA) $\theta_{h}$ of the LED as $\frac{\pi}{3}$ radians and the position of the receiver photodiode (PD) at half the attocell length $\frac{a}{2}$.}
\label{one5h3}
\end{figure} 

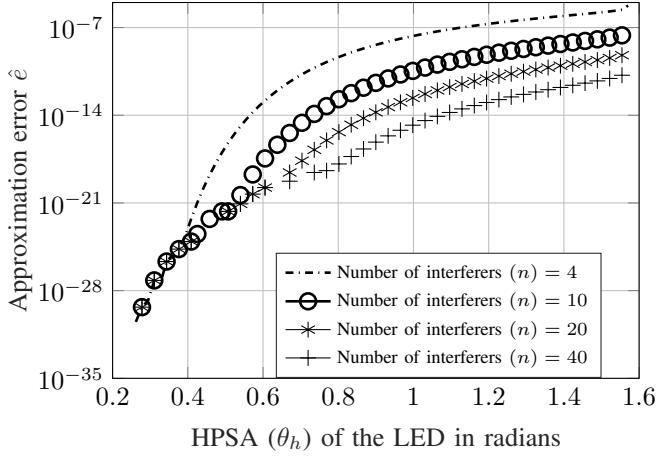
\begin{figure}[ht]
\centering
\begin{tikzpicture}
\begin{axis}[%
width=7cm,
height=5cm,
scale only axis,
xmin=0.2,
xmax=1.6,
xlabel style={font=\color{white!15!black}},
xlabel={HPSA ($\theta_{h}$) of the LED in radians},
ymode=log,
ymin=1e-35,
ymax=1e-05,
yminorticks=true,
ylabel style={font=\color{white!15!black}},
ylabel={Approximation error $\hat{e}$},
axis background/.style={fill=white},
title style={font=\bfseries},
title={},
xmajorgrids,
ymajorgrids,
yminorgrids,
legend style={font=\fontsize{7}{5}\selectfont,at={(0.31,0.001)}, anchor=south west, legend cell align=left, align=left, draw=white!15!black}
]
\addplot [color=black, dashdotted, line width=1.0pt]
  table[row sep=crcr]{%
0.261799387799149	3.20474742746036e-31\\
0.278161849536596	4.80712114119054e-30\\
0.294524311274043	1.65660790096412e-29\\
0.31088677301149	6.56332273143882e-28\\
0.327249234748937	4.03896783473158e-28\\
0.343611696486384	2.10026327406042e-26\\
0.35997415822383	0\\
0.376336619961277	2.58493941422821e-25\\
0.392699081698724	2.06795153138257e-24\\
0.409061543436171	5.45939204284998e-23\\
0.425424005173618	1.08195224121936e-21\\
0.441786466911065	1.57759886426113e-20\\
0.458148928648511	1.73443230840119e-19\\
0.474511390385958	1.50110120121223e-18\\
0.490873852123405	1.05339134902913e-17\\
0.507236313860852	6.15655309799948e-17\\
0.523598775598299	3.06456520316606e-16\\
0.539961237335746	1.3241352662236e-15\\
0.556323699073193	5.04785896784411e-15\\
0.572686160810639	1.7217609919718e-14\\
0.589048622548086	5.31841857066926e-14\\
0.605411084285533	1.50346149240107e-13\\
0.62177354602298	3.92527529306363e-13\\
0.638136007760427	9.54080008898871e-13\\
0.654498469497873	2.17410044502965e-12\\
0.67086093123532	4.67340764340445e-12\\
0.687223392972767	9.52824971700572e-12\\
0.703585854710214	1.85145962688753e-11\\
0.719948316447661	3.44349561343625e-11\\
0.736310778185108	6.15358988808115e-11\\
0.752673239922555	1.06019058929065e-10\\
0.769035701660001	1.76641905638188e-10\\
0.785398163397448	2.85397642638911e-10\\
0.801760625134895	4.48257023622119e-10\\
0.818123086872342	6.85949356994328e-10\\
0.834485548609789	1.02475604933129e-09\\
0.850848010347236	1.49728809017112e-09\\
0.867210472084682	2.14321899485483e-09\\
0.883572933822129	3.00994681891494e-09\\
0.899935395559576	4.1531626737961e-09\\
0.916297857297023	5.63730830245579e-09\\
0.93266031903447	7.53591129667823e-09\\
0.949022780771917	9.93179299110629e-09\\
0.965385242509364	1.2917150523685e-08\\
0.98174770424681	1.65935207670741e-08\\
0.998110165984257	2.10716393674037e-08\\
1.0144726277217	2.64712130804666e-08\\
1.03083508945915	3.29206276353938e-08\\
1.0471975511966	4.0556616767197e-08\\
1.06356001293404	4.95239207047929e-08\\
1.07992247467149	5.99749646699316e-08\\
1.09628493640894	7.20695898662541e-08\\
1.11264739814639	8.59748714122574e-08\\
1.12900985988383	1.01865060081474e-07\\
1.14537232162128	1.19921687765853e-07\\
1.16173478335873	1.40333880974119e-07\\
1.17809724509617	1.63298933097966e-07\\
1.19445970683362	1.89023195095415e-07\\
1.21082216857107	2.17723357353455e-07\\
1.22718463030851	2.49628213898345e-07\\
1.24354709204596	2.84981026520137e-07\\
1.25990955378341	3.24042643857526e-07\\
1.27627201552085	3.67095584744362e-07\\
1.2926344772583	4.14449373539044e-07\\
1.30899693899575	4.66447531370651e-07\\
1.32535940073319	5.23476801114753e-07\\
1.34172186247064	5.85979451003804e-07\\
1.35808432420809	6.54469920546158e-07\\
1.37444678594553	7.29557749686374e-07\\
1.39080924768298	8.1197986051551e-07\\
1.40717170942043	9.02647216055597e-07\\
1.42353417115788	1.00271441478422e-06\\
1.43989663289532	1.11368751001643e-06\\
1.45625909463277	1.23759897011078e-06\\
1.47262155637022	1.3773084556086e-06\\
1.48898401810766	1.53705955332692e-06\\
1.50534647984511	1.72361684232583e-06\\
1.52170894158256	1.9489335542737e-06\\
1.53807140332	2.23790216650685e-06\\
1.55443386505745	2.6617465198886e-06\\
1.5707963267949	5.71205126436575e-06\\
};
\addlegendentry{Number of interferers $(n) = 4$}

\addplot [color=black, line width=1.0pt, draw=none, mark=o, mark size=3pt]
  table[row sep=crcr]{%
0.278161849536596	4.80712114119054e-30\\
0.31088677301149	6.56332273143882e-28\\
0.343611696486384	2.10026327406042e-26\\
0.376336619961277	2.06795153138257e-25\\
0.409061543436171	8.27180612553028e-25\\
0.425424005173618	3.30872245021211e-24\\
0.441786466911065	0\\
0.458148928648511	5.29395592033938e-23\\
0.474511390385958	0\\
0.490873852123405	2.11758236813575e-22\\
0.507236313860852	2.11758236813575e-22\\
0.539961237335746	4.2351647362715e-21\\
0.572686160810639	1.88041314290455e-19\\
0.605411084285533	3.57786716920216e-18\\
0.638136007760427	4.41947910559404e-17\\
0.67086093123532	3.84281907510331e-16\\
0.703585854710214	2.49816443573248e-15\\
0.736310778185108	1.27788946958918e-14\\
0.769035701660001	5.35190923696416e-14\\
0.801760625134895	1.89458908630979e-13\\
0.834485548609789	5.81730518790791e-13\\
0.867210472084682	1.5822642675245e-12\\
0.899935395559576	3.87901256254652e-12\\
0.93266031903447	8.69562473242158e-12\\
0.965385242509364	1.80403838569154e-11\\
0.998110165984257	3.4991141061036e-11\\
1.03083508945915	6.39986440578721e-11\\
1.06356001293404	1.11191498283064e-10\\
1.09628493640894	1.84672621812709e-10\\
1.12900985988383	2.94808945484892e-10\\
1.16173478335873	4.54532499312432e-10\\
1.19445970683362	6.79692782959429e-10\\
1.22718463030851	9.89530119302473e-10\\
1.25990955378341	1.40738512163785e-09\\
1.2926344772583	1.96183750494905e-09\\
1.32535940073319	2.6886165360035e-09\\
1.35808432420809	3.63393603576567e-09\\
1.39080924768298	4.86061754040312e-09\\
1.42353417115788	6.46017651872999e-09\\
1.45625909463277	8.57937691722821e-09\\
1.48898401810766	1.14890195412887e-08\\
1.52170894158256	1.58194554712721e-08\\
1.55443386505745	2.4071242256396e-08\\
};
\addlegendentry{Number of interferers $(n) = 10$}

\addplot [color=black, draw=none, mark=asterisk, mark size=3pt]
  table[row sep=crcr]{%
0.278161849536596	4.80712114119054e-30\\
0.31088677301149	6.56332273143882e-28\\
0.343611696486384	2.10026327406042e-26\\
0.376336619961277	2.06795153138257e-25\\
0.409061543436171	8.27180612553028e-25\\
0.441786466911065	0\\
0.474511390385958	0\\
0.507236313860852	2.11758236813575e-22\\
0.539961237335746	8.470329472543e-22\\
0.572686160810639	5.0821976835258e-21\\
0.605411084285533	1.6940658945086e-20\\
0.638136007760427	0\\
0.67086093123532	2.71050543121376e-19\\
0.703585854710214	2.35813972515597e-18\\
0.736310778185108	1.75098650856409e-17\\
0.769035701660001	9.8282926935811e-17\\
0.801760625134895	4.50702843102224e-16\\
0.834485548609789	1.74209605074971e-15\\
0.867210472084682	5.81240784669479e-15\\
0.899935395559576	1.71106618457317e-14\\
0.93266031903447	4.521578424177e-14\\
0.965385242509364	1.08832214074095e-13\\
0.998110165984257	2.41528151645465e-13\\
1.03083508945915	4.99384388008561e-13\\
1.06356001293404	9.70542222977766e-13\\
1.09628493640894	1.78651711479905e-12\\
1.12900985988383	3.13533660839527e-12\\
1.16173478335873	5.27657952054117e-12\\
1.19445970683362	8.55907792557398e-12\\
1.22718463030851	1.34428440601297e-11\\
1.25990955378341	2.05282249532424e-11\\
1.2926344772583	3.05979477865925e-11\\
1.32535940073319	4.46826204508755e-11\\
1.35808432420809	6.41704762244233e-11\\
1.39080924768298	9.10051270452961e-11\\
1.42353417115788	1.28074089528196e-10\\
1.45625909463277	1.80064387012013e-10\\
1.48898401810766	2.55691825351567e-10\\
1.52170894158256	3.75395303384152e-10\\
1.55443386505745	6.21334355493941e-10\\
};
\addlegendentry{Number of interferers $(n) = 20$}

\addplot [color=black, draw=none, mark=+, mark size=3pt]
  table[row sep=crcr]{%
0.278161849536596	4.80712114119054e-30\\
0.31088677301149	6.56332273143882e-28\\
0.343611696486384	2.10026327406042e-26\\
0.376336619961277	2.06795153138257e-25\\
0.409061543436171	8.27180612553028e-25\\
0.441786466911065	0\\
0.474511390385958	0\\
0.507236313860852	2.11758236813575e-22\\
0.539961237335746	8.470329472543e-22\\
0.572686160810639	5.0821976835258e-21\\
0.605411084285533	1.6940658945086e-20\\
0.638136007760427	0\\
0.67086093123532	5.42101086242752e-20\\
0.703585854710214	0\\
0.736310778185108	2.71050543121376e-19\\
0.769035701660001	3.79470760369927e-19\\
0.801760625134895	1.30104260698261e-18\\
0.834485548609789	5.09575021068187e-18\\
0.867210472084682	2.03830008427275e-17\\
0.899935395559576	7.17741838185404e-17\\
0.93266031903447	2.27031934918465e-16\\
0.965385242509364	6.34475111338517e-16\\
0.998110165984257	1.61199179005145e-15\\
1.03083508945915	3.76781938982163e-15\\
1.06356001293404	8.19656842399041e-15\\
1.09628493640894	1.67322752875343e-14\\
1.12900985988383	3.2309224740068e-14\\
1.16173478335873	5.93977991791839e-14\\
1.19445970683362	1.04576070025786e-13\\
1.22718463030851	1.77279198265712e-13\\
1.25990955378341	2.90812512959704e-13\\
1.2926344772583	4.63715871257264e-13\\
1.32535940073319	7.21894766186892e-13\\
1.35808432420809	1.10206635484111e-12\\
1.39080924768298	1.6578197148398e-12\\
1.42353417115788	2.47144492371287e-12\\
1.45625909463277	3.68002503470244e-12\\
1.48898401810766	5.54347540093758e-12\\
1.52170894158256	8.68193190950439e-12\\
1.55443386505745	1.56402182871496e-11\\
};
\addlegendentry{Number of interferers $(n) = 40$}
\end{axis}
\end{tikzpicture}%
\caption{\textit{(One Dimension Model)} Here the variation of interference approximation error $\hat{e}=|\mathbb{I}_{n}(z)-\hat{\mathtt{I}}_{1}(z)|$ is drawn for a linear variation of the half-power-semi-angle (HPSA) $\theta_{h}$ of the LED for different number interferers $(n)$ in the network. We consider the attocell length $a=0.5$m, the height $h$ of the LED as $2.5$m and the position $z$ of the receiver photodiode (PD) at half the attocell length $\frac{a}{2}$.}
\label{one5t3}
\end{figure} 

\begin{figure}[ht]
\centering
\begin{tikzpicture}
\begin{axis}[%
width=7cm,
height=5cm,
scale only axis,
xmin=0,
xmax=0.25,
xlabel style={font=\color{white!15!black}},
xlabel={Distance ($z$) of the PD from origin in metres},
ymode=log,
ymin=1e-15,
ymax=1e-05,
yminorticks=true,
ylabel style={font=\color{white!15!black}},
ylabel={Approximation error $\hat{e}$},
axis background/.style={fill=white},
title style={font=\bfseries},
title={},
xmajorgrids,
ymajorgrids,
yminorgrids,
legend style={font=\fontsize{7}{5}\selectfont,at={(0.32,0.541)}, anchor=south west, legend cell align=left, align=left, draw=white!15!black}
]
\addplot [color=black, dashdotted, line width=1.0pt]
  table[row sep=crcr]{%
0	2.70967324888551e-06\\
0.001	2.70967511052395e-06\\
0.002	2.70968069544407e-06\\
0.003	2.70969000365669e-06\\
0.004	2.70970303517873e-06\\
0.005	2.70971979003925e-06\\
0.006	2.70974026826947e-06\\
0.007	2.7097644699119e-06\\
0.008	2.70979239501206e-06\\
0.009	2.70982404362808e-06\\
0.01	2.70985941582067e-06\\
0.011	2.70989851166442e-06\\
0.012	2.70994133123393e-06\\
0.013	2.70998787461419e-06\\
0.014	2.71003814190018e-06\\
0.015	2.71009213319079e-06\\
0.016	2.71014984859312e-06\\
0.017	2.71021128822428e-06\\
0.018	2.71027645220526e-06\\
0.019	2.71034534066444e-06\\
0.02	2.71041795374362e-06\\
0.021	2.71049429158374e-06\\
0.022	2.71057435433748e-06\\
0.023	2.71065814216528e-06\\
0.024	2.71074565523369e-06\\
0.025	2.71083689371834e-06\\
0.026	2.71093185779964e-06\\
0.027	2.71103054766669e-06\\
0.028	2.71113296351853e-06\\
0.029	2.71123910555726e-06\\
0.03	2.71134897399453e-06\\
0.031	2.71146256905063e-06\\
0.032	2.71157989095196e-06\\
0.033	2.71170093993009e-06\\
0.034	2.71182571622702e-06\\
0.035	2.71195422009516e-06\\
0.036	2.71208645178734e-06\\
0.037	2.71222241156723e-06\\
0.038	2.71236209970632e-06\\
0.039	2.71250551648391e-06\\
0.04	2.71265266218578e-06\\
0.041	2.71280353710294e-06\\
0.042	2.7129581415407e-06\\
0.043	2.71311647580394e-06\\
0.044	2.7132785402075e-06\\
0.045	2.71344433507709e-06\\
0.046	2.71361386074229e-06\\
0.047	2.71378711754269e-06\\
0.048	2.71396410582046e-06\\
0.049	2.71414482593164e-06\\
0.05	2.7143292782336e-06\\
0.051	2.71451746309626e-06\\
0.052	2.71470938089389e-06\\
0.053	2.71490503201072e-06\\
0.054	2.71510441683404e-06\\
0.055	2.71530753576284e-06\\
0.056	2.7155143892039e-06\\
0.057	2.71572497756705e-06\\
0.058	2.71593930127427e-06\\
0.059	2.71615736075186e-06\\
0.06	2.71637915643394e-06\\
0.061	2.71660468876504e-06\\
0.062	2.71683395819356e-06\\
0.063	2.7170669651749e-06\\
0.064	2.7173037101761e-06\\
0.065	2.71754419366946e-06\\
0.066	2.71778841613246e-06\\
0.067	2.71803637805559e-06\\
0.068	2.71828807992891e-06\\
0.069	2.71854352225853e-06\\
0.07	2.71880270555098e-06\\
0.071	2.71906563032494e-06\\
0.072	2.71933229710345e-06\\
0.073	2.71960270642035e-06\\
0.074	2.71987685881124e-06\\
0.075	2.72015475482733e-06\\
0.076	2.72043639502027e-06\\
0.077	2.72072177995341e-06\\
0.078	2.72101091019356e-06\\
0.079	2.72130378631882e-06\\
0.08	2.72160040891423e-06\\
0.081	2.72190077857002e-06\\
0.082	2.72220489588642e-06\\
0.083	2.72251276146971e-06\\
0.084	2.72282437593311e-06\\
0.085	2.72313973990026e-06\\
0.086	2.72345885399825e-06\\
0.087	2.7237817188633e-06\\
0.088	2.72410833514246e-06\\
0.089	2.72443870348409e-06\\
0.09	2.72477282454867e-06\\
0.091	2.72511069900495e-06\\
0.092	2.72545232752296e-06\\
0.093	2.72579771078748e-06\\
0.094	2.7261468494846e-06\\
0.095	2.72649974431341e-06\\
0.096	2.72685639597777e-06\\
0.097	2.7272168051872e-06\\
0.098	2.72758097266332e-06\\
0.099	2.72794889913169e-06\\
0.1	2.72832058532523e-06\\
0.101	2.72869603198901e-06\\
0.102	2.72907523986851e-06\\
0.103	2.72945820972226e-06\\
0.104	2.72984494231569e-06\\
0.105	2.73023543841686e-06\\
0.106	2.73062969880812e-06\\
0.107	2.73102772427487e-06\\
0.108	2.73142951561249e-06\\
0.109	2.73183507362284e-06\\
0.11	2.73224439911388e-06\\
0.111	2.73265749290439e-06\\
0.112	2.7330743558175e-06\\
0.113	2.73349498868501e-06\\
0.114	2.73391939234653e-06\\
0.115	2.73434756765164e-06\\
0.116	2.73477951545122e-06\\
0.117	2.73521523661002e-06\\
0.118	2.73565473199672e-06\\
0.119	2.73609800248865e-06\\
0.12	2.73654504897268e-06\\
0.121	2.73699587233786e-06\\
0.122	2.73745047348756e-06\\
0.123	2.73790885332642e-06\\
0.124	2.73837101276951e-06\\
0.125	2.73883695274275e-06\\
0.126	2.7393066741742e-06\\
0.127	2.73978017799933e-06\\
0.128	2.74025746517048e-06\\
0.129	2.74073853663353e-06\\
0.13	2.74122339335124e-06\\
0.131	2.74171203629289e-06\\
0.132	2.74220446643252e-06\\
0.133	2.7427006847559e-06\\
0.134	2.74320069225224e-06\\
0.135	2.74370448991989e-06\\
0.136	2.74421207876498e-06\\
0.137	2.74472345980275e-06\\
0.138	2.74523863405193e-06\\
0.139	2.74575760254338e-06\\
0.14	2.74628036631272e-06\\
0.141	2.74680692640252e-06\\
0.142	2.74733728386966e-06\\
0.143	2.74787143976886e-06\\
0.144	2.74840939516741e-06\\
0.145	2.74895115114169e-06\\
0.146	2.74949670877114e-06\\
0.147	2.75004606914995e-06\\
0.148	2.75059923337054e-06\\
0.149	2.75115620254184e-06\\
0.15	2.75171697777538e-06\\
0.151	2.75228156019049e-06\\
0.152	2.75284995091604e-06\\
0.153	2.75342215108826e-06\\
0.154	2.75399816184818e-06\\
0.155	2.75457798435026e-06\\
0.156	2.75516161975025e-06\\
0.157	2.75574906921389e-06\\
0.158	2.75634033391689e-06\\
0.159	2.75693541504059e-06\\
0.16	2.75753431377458e-06\\
0.161	2.75813703131539e-06\\
0.162	2.75874356886431e-06\\
0.163	2.75935392763998e-06\\
0.164	2.75996810885628e-06\\
0.165	2.76058611374313e-06\\
0.166	2.76120794353782e-06\\
0.167	2.7618335994798e-06\\
0.168	2.76246308282199e-06\\
0.169	2.76309639482119e-06\\
0.17	2.76373353674505e-06\\
0.171	2.76437450986728e-06\\
0.172	2.76501931546509e-06\\
0.173	2.76566795483386e-06\\
0.174	2.76632042926685e-06\\
0.175	2.76697674006943e-06\\
0.176	2.76763688855262e-06\\
0.177	2.76830087603958e-06\\
0.178	2.76896870385519e-06\\
0.179	2.76964037333434e-06\\
0.18	2.77031588582057e-06\\
0.181	2.77099524266826e-06\\
0.182	2.77167844523266e-06\\
0.183	2.77236549488073e-06\\
0.184	2.77305639298595e-06\\
0.185	2.77375114093176e-06\\
0.186	2.77444974010463e-06\\
0.187	2.77515219190751e-06\\
0.188	2.77585849774034e-06\\
0.189	2.77656865901822e-06\\
0.19	2.77728267716274e-06\\
0.191	2.77800055359942e-06\\
0.192	2.77872228976591e-06\\
0.193	2.77944788710724e-06\\
0.194	2.78017734707363e-06\\
0.195	2.78091067112485e-06\\
0.196	2.78164786072892e-06\\
0.197	2.78238891735991e-06\\
0.198	2.78313384250274e-06\\
0.199	2.78388263764496e-06\\
0.2	2.78463530428794e-06\\
0.201	2.78539184393614e-06\\
0.202	2.78615225810311e-06\\
0.203	2.78691654831539e-06\\
0.204	2.78768471609694e-06\\
0.205	2.78845676298777e-06\\
0.206	2.78923269053438e-06\\
0.207	2.79001250028804e-06\\
0.208	2.79079619381045e-06\\
0.209	2.79158377267065e-06\\
0.21	2.79237523844508e-06\\
0.211	2.7931705927197e-06\\
0.212	2.79396983708614e-06\\
0.213	2.7947729731425e-06\\
0.214	2.79558000250034e-06\\
0.215	2.79639092677295e-06\\
0.216	2.79720574758663e-06\\
0.217	2.7980244665716e-06\\
0.218	2.79884708536803e-06\\
0.219	2.79967360562218e-06\\
0.22	2.80050402899114e-06\\
0.221	2.80133835713764e-06\\
0.222	2.80217659173135e-06\\
0.223	2.80301873445583e-06\\
0.224	2.80386478699202e-06\\
0.225	2.80471475103865e-06\\
0.226	2.80556862829661e-06\\
0.227	2.80642642047938e-06\\
0.228	2.80728812930347e-06\\
0.229	2.80815375649406e-06\\
0.23	2.8090233037889e-06\\
0.231	2.80989677292837e-06\\
0.232	2.81077416566322e-06\\
0.233	2.81165548375118e-06\\
0.234	2.81254072895991e-06\\
0.235	2.81342990306188e-06\\
0.236	2.81432300784038e-06\\
0.237	2.81522004508565e-06\\
0.238	2.81612101659529e-06\\
0.239	2.81702592417428e-06\\
0.24	2.81793476963887e-06\\
0.241	2.81884755481011e-06\\
0.242	2.81976428151594e-06\\
0.243	2.82068495159778e-06\\
0.244	2.8216095668992e-06\\
0.245	2.8225381292755e-06\\
0.246	2.82347064058672e-06\\
0.247	2.82440710270551e-06\\
0.248	2.82534751750796e-06\\
0.249	2.82629188688016e-06\\
0.25	2.82724021271729e-06\\
};
\addlegendentry{Number of interferers $(n) = 4$}

\addplot [color=black, line width=1.0pt, draw=none, mark=o, mark size=3pt]
  table[row sep=crcr]{%
0	8.46370334928181e-11\\
0.01	8.46373969173864e-11\\
0.02	8.46384923952614e-11\\
0.03	8.46403151559538e-11\\
0.04	8.46428673678679e-11\\
0.05	8.46461498983653e-11\\
0.06	8.46501623137652e-11\\
0.07	8.46549063487911e-11\\
0.08	8.46603767992726e-11\\
0.09	8.46665762672949e-11\\
0.1	8.4673507788624e-11\\
0.11	8.46811691948557e-11\\
0.12	8.46895617870325e-11\\
0.13	8.46986838304309e-11\\
0.14	8.47085327229657e-11\\
0.15	8.47191149698501e-11\\
0.16	8.47304292700413e-11\\
0.17	8.4742473455135e-11\\
0.18	8.4755244923046e-11\\
0.19	8.47687532147534e-11\\
0.2	8.47829870545547e-11\\
0.21	8.47979577181524e-11\\
0.22	8.48136560982482e-11\\
0.23	8.48300856642892e-11\\
0.24	8.4847250753084e-11\\
0.25	8.48651478951856e-11\\
};
\addlegendentry{Number of interferers $(n) = 10$}

\addplot [color=black, draw=none, mark=asterisk, mark size=3pt]
  table[row sep=crcr]{%
0	7.00759762717329e-13\\
0.01	7.00759762717329e-13\\
0.02	7.00763232164281e-13\\
0.03	7.00764533206888e-13\\
0.04	7.0077233946253e-13\\
0.05	7.00777543632958e-13\\
0.06	7.00786217250338e-13\\
0.07	7.007979266338e-13\\
0.08	7.00809202336394e-13\\
0.09	7.00819176996381e-13\\
0.1	7.00835223188534e-13\\
0.11	7.00852570423294e-13\\
0.12	7.00870785019791e-13\\
0.13	7.00892035382372e-13\\
0.14	7.00909816298001e-13\\
0.15	7.00931500341451e-13\\
0.16	7.00959255917066e-13\\
0.17	7.00984409407468e-13\\
0.18	7.01008695536132e-13\\
0.19	7.01039920558699e-13\\
0.2	7.01067242453446e-13\\
0.21	7.01102370603834e-13\\
0.22	7.01132728264664e-13\\
0.23	7.01166989053315e-13\\
0.24	7.01203851927179e-13\\
0.25	7.01244617928864e-13\\
};
\addlegendentry{Number of interferers $(n) = 20$}

\addplot [color=black, draw=none, mark=+, mark size=3pt]
  table[row sep=crcr]{%
0.007	5.60402418914308e-15\\
0.015	5.60749363609503e-15\\
0.023	5.6022894656671e-15\\
0.031	5.6022894656671e-15\\
0.039	5.60012106132213e-15\\
0.047	5.6022894656671e-15\\
0.055	5.6018557847981e-15\\
0.063	5.6018557847981e-15\\
0.071	5.60142210392911e-15\\
0.079	5.60272314653609e-15\\
0.087	5.60012106132213e-15\\
0.095	5.60575891261905e-15\\
0.103	5.60575891261905e-15\\
0.111	5.6022894656671e-15\\
0.119	5.60359050827408e-15\\
0.127	5.60402418914308e-15\\
0.135	5.59925369958414e-15\\
0.143	5.60705995522603e-15\\
0.151	5.60662627435704e-15\\
0.159	5.60272314653609e-15\\
0.167	5.60662627435704e-15\\
0.175	5.60272314653609e-15\\
0.183	5.60575891261905e-15\\
0.191	5.60315682740509e-15\\
0.199	5.60142210392911e-15\\
0.207	5.60749363609503e-15\\
0.215	5.60359050827408e-15\\
0.223	5.60359050827408e-15\\
0.231	5.60098842306012e-15\\
0.239	5.60142210392911e-15\\
0.247	5.60749363609503e-15\\
};
\addlegendentry{Number of interferers $(n) = 40$}

\end{axis}
\end{tikzpicture}%
\caption{\textit{(One Dimension Model)} Here the variation of interference approximation error $\hat{e}=|\mathbb{I}_{n}(z)-\hat{\mathtt{I}}_{1}(z)|$ is drawn for a linear variation of the position $z$ of the receiver photodiode (PD) inside the attocell for different number interferers $(n)$ in the network. We consider $a=0.5$m, the half-power-semi-angle (HPSA) $\theta_{h}$ of the LED as $\frac{\pi}{3}$ radians and the height $h$ of LED as $2.5$m.}
\label{one5z3}
\end{figure}    

From Fig. \ref{one5h1} and it's corresponding approximation error plot in Fig. \ref{one5h3}, we observe that for any given height $h$, as the number of interferers increase, the error $\hat{e}$, decreases.  On the log axis, we observe a maximum error $\hat{e}_{max}$ in the order of $10^{-4}$ with respect to $\hat{\mathtt{I}}_{1}(z)$, that is in the order of $10^{-2}$. This error further reduces as the number of interferers is increased. The same can be observed with the variation of HPSA in graphs of Fig. \ref{one5t1} and the error plot in Fig. \ref{one5t3}, where $\hat{e}_{max}$ is in the order of $10^{-7}$, for $\hat{\mathtt{I}}_{1}(z)$ in the order of $10^{-3}$. Again, this error reduces as the number of interferers increases. Similarly, in graphs of Fig. \ref{one5z1} and Fig. \ref{one5z3}, we observe $\hat{e}_{max}$ in the order of $10^{-5}$, for $\hat{\mathtt{I}}_{1}(z)$ in the order of $10^{-2}$. So, when compared with the interference values, these errors are small, which numerically validates the approximation to $\hat{\mathtt{I}}_{1}(z)$. \\ \par

As seen in the above example, Prop. \ref{prop1} essentially implies that for a given value of $\frac{h}{a}$ the approximation to $\hat {\mathtt{I}}_{k}(z)$ is tight and very close to the actual interference $\mathbb{I}_{\infty}(z)$ in \eqref{eqn:interf1a}, with an approximation error bounded by an exponential decay.
So, this characterization can be summarized as   
\begin{align*}
\mathbb{I}_{n}(z) < \mathbb{I}_{\infty}(z) \approx \hat{\mathtt{I}}_{k}(z). 
\end{align*}  
This also implies that our characterization provides closed form analytical bounds for interference in finite LED networks.
\subsection{One dimension model with FOV $\theta_{f}<\frac{\pi}{2}$ radians} 
We now look at the interference characterization when $\theta_{f}<\frac{\pi}{2}$ radians. Here we show that, the Fourier analysis method can be used to give a suitable interference approximation for such cases as well.
The infinite summation in \eqref{eqn:interf1a} becomes a finite summation, when the FOV constraint function $\rho(D_{d})$, acts on every interferer. From the proof of Thm. \ref{theorem1}, we can modify the function $q(.)$ in \eqref{eqn:intc} as 
 \begin{equation*}
q'(x) = ( x^{2} + h^{2} )^{-\beta}\rho(D_{d}).
\end{equation*}
The Poisson summation theorem can be used to obtain a similar result as in the previous subsection if the Fourier transform of $q'(x)$ can be obtained. 

Hence the Fourier transform of  $q'(x) $ equals
\begin{align*}
Q'(w) &=\int_{-\infty}^{\infty}q'(x)e^{-\iota2\pi wx} \d x, \nonumber\\
&=\int_{-\infty}^{\infty}(x^{2}+h^{2})^{-\beta}\rho(D_{d})e^{-\iota2\pi wx} \d x, \nonumber\\
&=\int_{0}^{h\tan(\theta_{f})}\frac{2\cos(2\pi wx)}{(x^{2}+h^{2})^{\beta}} \d x.
\end{align*}
Hence we have the following Lemma. 
\begin{lemma}
For an FOV $\theta_{f}<\frac{\pi}{2}$ radians and a finite integer $k\geq1$ we have 
\begin{align}
\mathbb{I}_{\infty}(z) \approx \hat {\mathtt{I}}'_{k}(z) = &\frac{1}{a}\left[Q'(0)+\sum_{w=1}^{k}2Q'\bigg(\frac{w}{a}\bigg)\cos\bigg(\frac{2\pi w z}{a}\bigg)\right] \nonumber\\
&-\frac{1}{(z^{2}+h^{2})^{\beta}}.
\label{eqn:fovfov1}
\end{align}
\label{lem1}
\end{lemma}
\begin{proof}
Follows from the Poisson summation theorem and approximations. 
\end{proof}
The constant term evaluated at $w=0$ is 
\begin{align*}
Q'(0) &= \int_{0}^{h\tan(\theta_{f})}\frac{2}{(x^{2}+h^{2})^{\beta}} \d x, \nonumber\\  
&= 2h^{1-2\beta}\tan(\theta_{f}){}_{2}F_{1}(0.5,\beta;1.5;-\tan^{2}(\theta_{f})), 
\end{align*}
where ${}_{2}F_{1}(.;.;.)$ is the generalized hypergeometric function. As earlier, this represents the average spatial interference seen at all locations. A closed form expression for $Q'\big(\frac{w}{a}\big)$ can be simply evaluated using numerical integration. \\ \par

We consider $h=2.5$m and $a=0.5$m, leading to $\frac{h}{a}=5$ to numerically validate \eqref{eqn:fovfov1} for $k=1$ over various values of $\theta_{f}$ and compare it with $\mathbb{I}_{\infty}(z)$ in \eqref{eqn:interf1}. In Li-Fi attocell networks, if the FOV $\theta_{f}<\theta_{o}\big(=\tan^{-1}\big(\frac{a}{h}\big)\big)$, the PD does not experience any interference. Here the ratio $\frac{a}{h}=0.2$ and $\theta_{o}=0.197$ radians. So, in Fig. \ref{onefov1}, we observe that both $\mathbb{I}_{\infty}(z)$ and $\hat{\mathtt{I}}'_{1}(z)$ drop down to zero once $\theta_{f}<\theta_{o}=0.197$ radians. Also, for $\theta_{f}>\theta_{o}$, both the graphs, $\mathbb{I}_{\infty}(z)$ and $\hat{\mathtt{I}}'_{1}(z)$ are tightly bounded, which numerically validates our proposition in Lem. \ref{lem1} for $k=1$. Also, as $\theta_{f}\to 1.57 (=\frac{\pi}{2})$ radians, the interference values converge to the earlier case of $\theta_{f}=\frac{\pi}{2}$ radians. \\ \par

So, the approximation above in Lem. \ref{lem1} is a good approximation for various practical parameter values based on the choice of $k$. As shown above, if we choose $h=2.5$m and $a=0.5$m, considering $k=1$ is sufficient. When $\frac{h}{a}$ becomes small, a few more terms are necessary to improve the approximation accuracy. \\ \par

\begin{figure}[ht]
\centering
\definecolor{mycolor1}{rgb}{0.00000,0.44700,0.74100}%
\definecolor{mycolor2}{rgb}{0.85000,0.32500,0.09800}%
\begin{tikzpicture}
\begin{axis}[%
width=7cm,
height=5cm,
scale only axis,
xmin=0.01,
xmax=1.6,
xlabel style={font=\color{white!15!black}},
xlabel={FOV $(\theta_{f})$ of the PD in radians},
ymin=0.000,
ymax=0.003,
ylabel style={font=\color{white!15!black}},
ylabel={Interference},
axis background/.style={fill=white},
xmajorgrids,
ymajorgrids,
legend style={font=\fontsize{7}{5}\selectfont,at={(0.42,0.241)}, anchor=south west, legend cell align=left, align=left, draw=white!15!black}
]
\addplot [color=black, draw=none, mark=+, mark size=3pt]
  table[row sep=crcr]{%
0.0145444104332861	0\\
0.0436332312998582	0\\
0.0727220521664304	0\\
0.101810873033003	0\\
0.130899693899575	0\\
0.159988514766147	0\\
0.189077335632719	0\\
0.218166156499291	0.00112041\\
0.247254977365863	0.00112041\\
0.276343798232435	0.00112041\\
0.305432619099008	0.00112041\\
0.33452143996558	0.00112041\\
0.363610260832152	0.00112041\\
0.392699081698724	0.00184431\\
0.421787902565296	0.00184431\\
0.450876723431868	0.00184431\\
0.479965544298441	0.00184431\\
0.509054365165013	0.00184431\\
0.538143186031585	0.00184431\\
0.567232006898157	0.00222745\\
0.596320827764729	0.00222745\\
0.625409648631301	0.00222745\\
0.654498469497874	0.00222745\\
0.683587290364446	0.00240864\\
0.712676111231018	0.00240864\\
0.74176493209759	0.00240864\\
0.770853752964162	0.00240864\\
0.799942573830734	0.00249056\\
0.829031394697307	0.00249056\\
0.858120215563879	0.00249056\\
0.887209036430451	0.00252753\\
0.916297857297023	0.00252753\\
0.945386678163595	0.00252753\\
0.974475499030167	0.00254461\\
1.00356431989674	0.00254461\\
1.03265314076331	0.00255277\\
1.06174196162988	0.00255277\\
1.09083078249646	0.00255682\\
1.11991960336303	0.00255892\\
1.1490084242296	0.00256005\\
1.17809724509617	0.00256068\\
1.20718606596274	0.00256104\\
1.23627488682932	0.00256125\\
1.26536370769589	0.00256138\\
1.29445252856246	0.00256152\\
1.32354134942903	0.00256158\\
1.35263017029561	0.00256161\\
1.38171899116218	0.00256162\\
1.41080781202875	0.00256163\\
1.43989663289532	0.00256163\\
1.46898545376189	0.00256163\\
1.49807427462847	0.00256163\\
1.52716309549504	0.00256163\\
};
\addlegendentry{$\mathbb{I}_{\infty}(z)$}

\addplot [color=black, line width=1.0pt, draw=none, mark=o, mark size=3pt]
  table[row sep=crcr]{%
0	-0.00065536\\
0.0290888208665722	4.40698e-05\\
0.0581776417331443	-2.33366e-05\\
0.0872664625997165	1.64634e-05\\
0.116355283466289	-1.25903e-05\\
0.145444104332861	8.07848e-06\\
0.174532925199433	7.36684e-06\\
0.203621746066005	0.0011531\\
0.232710566932577	0.0010943\\
0.261799387799149	0.00113625\\
0.290888208665722	0.00111746\\
0.319977029532294	0.00110845\\
0.349065850398866	0.00113679\\
0.378154671265438	0.00128546\\
0.40724349213201	0.00185852\\
0.436332312998582	0.00185462\\
0.465421133865155	0.00184965\\
0.494509954731727	0.0018476\\
0.523598775598299	0.00185266\\
0.552687596464871	0.00221505\\
0.581776417331443	0.00222154\\
0.610865238198015	0.00222768\\
0.639954059064587	0.00223119\\
0.66904287993116	0.00221052\\
0.698131700797732	0.00240471\\
0.727220521664304	0.00240902\\
0.756309342530876	0.00241058\\
0.785398163397448	0.00244991\\
0.81448698426402	0.00248978\\
0.843575805130593	0.00249137\\
0.872664625997165	0.00248746\\
0.901753446863737	0.00252793\\
0.930842267730309	0.00252794\\
0.959931088596881	0.00254513\\
0.989019909463453	0.00254469\\
1.01810873033003	0.00255249\\
1.0471975511966	0.00255269\\
1.07628637206317	0.00255678\\
1.10537519292974	0.00255663\\
1.13446401379631	0.00255894\\
1.16355283466289	0.00256004\\
1.19264165552946	0.00256067\\
1.22173047639603	0.00256104\\
1.2508192972626	0.00256138\\
1.27990811812917	0.00256147\\
1.30899693899575	0.00256155\\
1.33808575986232	0.0025616\\
1.36717458072889	0.00256162\\
1.39626340159546	0.00256163\\
1.42535222246204	0.00256163\\
1.45444104332861	0.00256163\\
1.48352986419518	0.00256163\\
1.51261868506175	0.00256163\\
};
\addlegendentry{$\hat{\mathtt{I}}_{1}(z)$}

\end{axis}
\end{tikzpicture}%
\caption{\textit{(One Dimension Model)} ($\theta_{f}<\frac{\pi}{2}$ radians) Here the variation of $\hat{\mathtt{I}}'_{1}(z)$ is drawn for a linear variation of the FOV $\theta_{f}$ of the receiver photodiode (PD). $\mathbb{I}_{\infty}(z)$ from \eqref{eqn:interf1} (or $\mathbb{I}_{n}(z)$ for $n=20$) is also drawn to validate the same. We consider $a=0.5$m, the half-power-semi-angle (HPSA) $\theta_{h}$ of the LED as $\frac{\pi}{3}$ radians, the height $h$ of the LED as $2.5$m and $z=\frac{a}{2}$.}
\label{onefov1}
\end{figure}
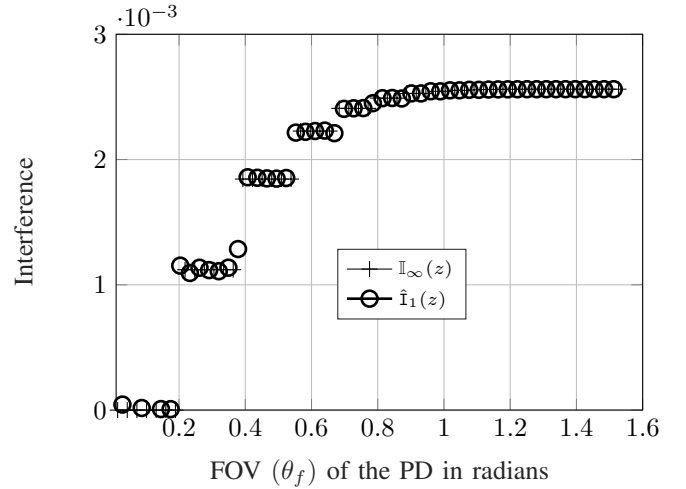    

\subsection{Two dimension model with FOV $\theta_{f}= \frac{\pi}{2}$ radians}
We now extend the result for two dimensions. 
\begin{theorem}
Consider a photodiode, with $\theta_{f}=\frac{\pi}{2}$ radians, situated at a distance $z=\sqrt{d_{x}^{2}+d_{y}^{2}}$ (inside an attocell) from the origin, in an infinite two dimension plane network of Li-Fi LEDs arranged as a regular square lattice of dimension $a$, emitting light with a Lambertian emission order $m$ and installed at a height $h$. Then, for a wavelength reuse factor of unity, the interference $\mathbb{I}_{\infty}(z)$, caused by the co-channel interferers at the photodiode is
\begin{align*}
\mathbb{I}_{\infty}(d_{x},d_{y}) =& \frac{h^{2-2\beta}\pi}{a^{2}(\beta-1)} -\frac{1}{(d_{x}^{2} + d_{y}^{2} + h^{2})^{\beta}} \nonumber\\
&+ \sum_{w=0}^{\infty}\sum_{k=0\setminus(0,0)}^{\infty} g(w,k),
\end{align*}
where 
\begin{align*}
&g(w,k)= \nonumber\\
&\frac{\big(\frac{h}{2\pi\sqrt{k^{2}+w^{2}}}\big)^{1-\beta}\mathbb{K}_{\beta-1}\big(\frac{2\pi h\sqrt{k^{2}+w^{2}}}{a}\big)\cos\big(\frac{2\pi wd_{x}}{a}\big)\cos\big(\frac{2\pi kd_{y}}{a}\big)}{2^{\beta-4} a^{\beta+1}\frac{\Gamma(\beta)}{\pi}}.
\end{align*}
Here $\Gamma(x)=\int_{0}^{\infty}t^{x-1}e^{-t}dt$ denotes the standard Gamma function, $\beta=m+3$ and $\mathbb{K}_{v}(y)=\frac{\Gamma(v+\frac{1}{2})(2y^{v})}{\sqrt{\pi}}\int_{0}^{\infty}\frac{\cos(t)dt}{(t^{2}+y^{2})^{v+\frac{1}{2}}}$ is the modified bessel function of second kind. 
\label{theorem2}
\end{theorem}

\begin{proof}
The proof is provided in Appendix \ref{app:theorem2}.
\end{proof}
In the next proposition, we quantify the error when the summation in the above infinite series is truncated after $j \times l$ terms using the asymptotic notation $O(.)$, similar to the one dimension case.
\begin{prop}
From Thm. \ref{theorem2}, for finite integers $j\geq 0$ and $l\geq 0$, the interference inside an attocell can be approximated to a closed form expression as  
\begin{align}  
\mathbb{I}_{\infty}(d_{x},d_{y})&=\hat{\mathtt{I}}_{j,l}(d_{x},d_{y}) \nonumber\\
&\ \ \ +O((\sqrt{j^{2}+l^{2}}+1)^{\beta-2.5}e^{\frac{-2\pi h(\sqrt{j^{2}+l^{2}}+1)}{a}}),
\label{eqn:newprop12d}
\end{align}
where
\begin{align*}
\hat{\mathtt{I}}_{j,l}(d_{x},d_{y})\triangleq& \frac{h^{2-2\beta}\pi}{a^{2}(\beta-1)} -\frac{1}{(d_{x}^{2} + d_{y}^{2} + h^{2})^{\beta}} \nonumber\\
&+ \sum_{(w,k)\in\mathbb{A}}g(w,k),
\end{align*} 
and the set $\mathbb{A} \triangleq (\mathbb{Z}^{2}\cap([0,j]\times[0,l]))\setminus\{(0,0)\}$ over the set of integers $\mathbb{Z}^{2}$.
\label{prop2}
\end{prop}  
\begin{proof}
The proof is provided in Appendix \ref{app:theorem2a}.
\end{proof}

We give similar theoretical and numerical validations to the two dimension model, as that of the one dimension model. \\ \par

Since in practice, the number of LEDs are finite, we also look at $ \mathbb{I}_{n}(d_{x},d_{y})$, i.e., looking at interference by a finite number of LEDs in \eqref{eqn:interf2a}. In Fig. \ref{two5h2}, we observe that as the number of interferers $n$ increases, the interference $\mathbb{I}_{n}(d_{x},d_{y})$, saturates to a constant value which is $\mathbb{I}_{\infty}(d_{x},d_{y})$. So, the approximation results in Prop. \ref{prop2} hold true for finite number of LEDs as well, even though the results are derived for an infinite plane. \\ \par

We now try to understand the interference characterization in Prop. \ref{prop2} by taking a few theoretical examples and further validation through numerical simulations. Firstly, in Prop. \ref{prop2}, the interference for any position $(d_{x},d_{y})$ of the user inside the attocell, always has a constant term given as 
\begin{equation*} 
\frac{h^{2-2\beta}\pi}{a^{2}(\beta-1)}.
\end{equation*}
This term represents the average spatial interference seen at all locations.\\ \par

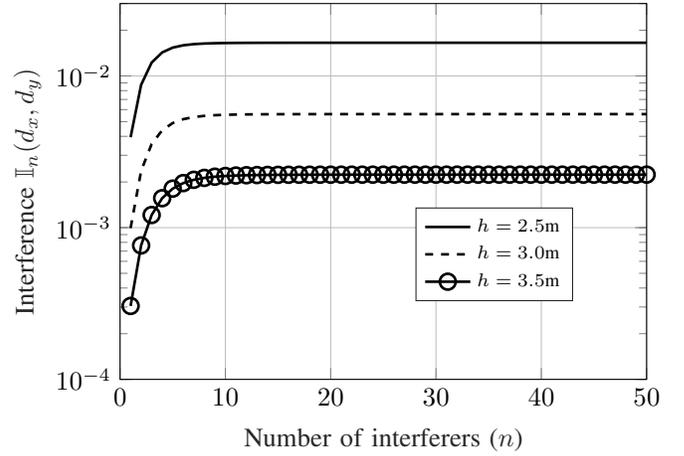
\begin{figure}[ht]
\centering
\definecolor{mycolor1}{rgb}{0.00000,0.44700,0.74100}%
\definecolor{mycolor2}{rgb}{0.85000,0.32500,0.09800}%
\definecolor{mycolor3}{rgb}{0.92900,0.69400,0.12500}%
\begin{tikzpicture}
\begin{axis}[%
width=7cm,
height=5cm,
scale only axis,
xmin=0,
xmax=50,
xlabel style={font=\color{white!15!black}},
xlabel={Number of interferers ($n$)},
ymode=log,
ymin=1e-04,
ymax=0.03,
ylabel style={font=\color{white!15!black}},
ylabel={Interference $\mathbb{I}_{n}(d_{x},d_{y})$},
axis background/.style={fill=white},
title style={font=\bfseries},
title={},
xmajorgrids,
ymajorgrids,
legend style={font=\fontsize{7}{5}\selectfont,at={(0.5610,0.208)}, anchor=south west, legend cell align=left, align=left, draw=white!15!black}
]
\addplot [color=black, line width=1.0pt]
  table[row sep=crcr]{%
1	0.003944860204854\\
2	0.00873198995777312\\
3	0.0122405476424056\\
4	0.0142861369481787\\
5	0.0153659971558498\\
6	0.0159189137177137\\
7	0.0162035277779052\\
8	0.0163534236983724\\
9	0.0164348168659612\\
10	0.0164804945522309\\
11	0.0165069814242476\\
12	0.0165228267402207\\
13	0.0165325858349299\\
14	0.0165387600358441\\
15	0.016542763661778\\
16	0.0165454190135343\\
17	0.0165472168617521\\
18	0.0165484573409729\\
19	0.0165493282023934\\
20	0.0165499493841225\\
21	0.0165503990070681\\
22	0.0165507288788616\\
23	0.0165509739352266\\
24	0.0165511581026151\\
25	0.016551298005531\\
26	0.0165514053511619\\
27	0.0165514884878457\\
28	0.0165515534389018\\
29	0.0165516045980881\\
30	0.0165516452036148\\
31	0.0165516776652303\\
32	0.016551703792559\\
33	0.0165517249562531\\
34	0.0165517422029089\\
35	0.0165517563378142\\
36	0.0165517679850855\\
37	0.0165517776317549\\
38	0.0165517856603612\\
39	0.0165517923732321\\
40	0.0165517980107116\\
41	0.0165518027649412\\
42	0.0165518067903489\\
43	0.0165518102116853\\
44	0.0165518131302186\\
45	0.0165518156285379\\
46	0.0165518177742991\\
47	0.0165518196231621\\
48	0.0165518212211052\\
49	0.0165518226062584\\
50	0.0165518238103626\\
51	0.0165518248599358\\
52	0.01655182577721\\
53	0.0165518265808862\\
54	0.0165518272867464\\
55	0.0165518279081503\\
56	0.0165518284564416\\
57	0.0165518289412798\\
58	0.016551829370913\\
59	0.0165518297524037\\
60	0.0165518300918148\\
61	0.0165518303943643\\
62	0.0165518306645539\\
63	0.0165518309062765\\
64	0.0165518311229065\\
65	0.0165518313173751\\
66	0.0165518314922341\\
67	0.0165518316497098\\
68	0.0165518317917481\\
69	0.0165518319200537\\
70	0.0165518320361222\\
71	0.0165518321412687\\
72	0.0165518322366515\\
73	0.0165518323232926\\
74	0.0165518324020952\\
75	0.0165518324738591\\
76	0.0165518325392933\\
77	0.0165518325990276\\
78	0.0165518326536221\\
79	0.016551832703576\\
80	0.0165518327493343\\
81	0.0165518327912948\\
82	0.0165518328298134\\
83	0.0165518328652088\\
84	0.0165518328977671\\
85	0.0165518329277451\\
86	0.0165518329553737\\
87	0.0165518329808609\\
88	0.0165518330043943\\
89	0.0165518330261432\\
90	0.0165518330462604\\
91	0.0165518330648843\\
92	0.0165518330821403\\
93	0.0165518330981419\\
94	0.0165518331129922\\
95	0.016551833126785\\
96	0.0165518331396053\\
97	0.0165518331515308\\
98	0.016551833162632\\
99	0.0165518331729735\\
100	0.016551833182614\\
};
\addlegendentry{$h=2.5$m}

\addplot [color=black, dashed, line width=1.0pt]
  table[row sep=crcr]{%
1	0.00099393160440774\\
2	0.00236673024956608\\
3	0.00355822190812618\\
4	0.00438314246088818\\
5	0.00489124755511692\\
6	0.00518747800852782\\
7	0.0053569770849244\\
8	0.00545416391513168\\
9	0.00551065147386661\\
10	0.00554413171103763\\
11	0.00556442127741532\\
12	0.00557700292297608\\
13	0.00558498376854923\\
14	0.00559015772595595\\
15	0.0055935818650636\\
16	0.00559589221301877\\
17	0.00559747943251683\\
18	0.00559858830945688\\
19	0.00559937517222347\\
20	0.00559994167421762\\
21	0.00560035504992836\\
22	0.00560066048578792\\
23	0.00560088880983723\\
24	0.00560106135272274\\
25	0.00560119307009326\\
26	0.00560129457885406\\
27	0.00560137350427696\\
28	0.00560143538336324\\
29	0.00560148427844671\\
30	0.00560152319906043\\
31	0.00560155439532344\\
32	0.00560157956420135\\
33	0.00560159999601149\\
34	0.00560161667950249\\
35	0.0056016303779209\\
36	0.00560164168456056\\
37	0.00560165106366819\\
38	0.00560165888080522\\
39	0.00560166542555452\\
40	0.00560167092862447\\
41	0.00560167557482158\\
42	0.00560167951295425\\
43	0.00560168286344105\\
44	0.00560168572419072\\
45	0.00560168817517255\\
46	0.00560169028198839\\
47	0.00560169209867909\\
48	0.00560169366994056\\
49	0.00560169503288235\\
50	0.00560169621842934\\
51	0.00560169725244433\\
52	0.00560169815663079\\
53	0.00560169894926197\\
54	0.00560169964577193\\
55	0.00560170025923671\\
56	0.00560170080076735\\
57	0.00560170127983224\\
58	0.00560170170452216\\
59	0.00560170208176918\\
60	0.00560170241752799\\
61	0.0056017027169264\\
62	0.0056017029843909\\
63	0.00560170322375145\\
64	0.0056017034383295\\
65	0.00560170363101171\\
66	0.00560170380431232\\
67	0.00560170396042566\\
68	0.00560170410127076\\
69	0.00560170422852912\\
70	0.00560170434367694\\
71	0.00560170444801255\\
72	0.00560170454267989\\
73	0.00560170462868864\\
74	0.00560170470693148\\
75	0.00560170477819894\\
76	0.00560170484319224\\
77	0.00560170490253434\\
78	0.00560170495677951\\
79	0.00560170500642164\\
80	0.00560170505190147\\
81	0.00560170509361287\\
82	0.00560170513190829\\
83	0.00560170516710353\\
84	0.00560170519948195\\
85	0.00560170522929814\\
86	0.00560170525678111\\
87	0.00560170528213706\\
88	0.00560170530555199\\
89	0.00560170532719376\\
90	0.00560170534721411\\
91	0.00560170536575034\\
92	0.00560170538292678\\
93	0.00560170539885621\\
94	0.00560170541364098\\
95	0.00560170542737409\\
96	0.00560170544014012\\
97	0.00560170545201609\\
98	0.00560170546307218\\
99	0.00560170547337245\\
100	0.00560170548297532\\
};
\addlegendentry{$h=3.0$m}

\addplot [color=black, line width=1.0pt, mark=o, mark size=3pt]
  table[row sep=crcr]{%
1	0.000304738487122572\\
2	0.000763550624650119\\
3	0.00121149961334361\\
4	0.00156348840972001\\
5	0.00180809582329984\\
6	0.0019668926453536\\
7	0.00206653303732045\\
8	0.00212826657834849\\
9	0.00216653474579435\\
10	0.00219045936156176\\
11	0.00220561377169303\\
12	0.00221536328024466\\
13	0.00222174078794451\\
14	0.00222598360231445\\
15	0.00222885359532904\\
16	0.00223082643211871\\
17	0.00223220355924256\\
18	0.0022331789639048\\
19	0.00223387939971175\\
20	0.00223438893345123\\
21	0.00223476413012742\\
22	0.00223504357967293\\
23	0.00223525395854302\\
24	0.00223541394041038\\
25	0.00223553675349279\\
26	0.00223563187484378\\
27	0.00223570616703976\\
28	0.00223576465000396\\
29	0.00223581103114609\\
30	0.00223584807351737\\
31	0.00223587785416657\\
32	0.00223590194726277\\
33	0.0022359215551333\\
34	0.00223593760288561\\
35	0.00223595080732628\\
36	0.00223596172757488\\
37	0.00223597080252802\\
38	0.00223597837879918\\
39	0.0022359847317055\\
40	0.00223599008114046\\
41	0.00223599460365751\\
42	0.00223599844172724\\
43	0.00223600171087183\\
44	0.00223600450519535\\
45	0.00223600690169444\\
46	0.00223600896363602\\
47	0.00223601074321761\\
48	0.00223601228367285\\
49	0.00223601362094574\\
50	0.00223601478502808\\
51	0.00223601580103239\\
52	0.00223601669005635\\
53	0.0022360174698821\\
55	0.00223601875979173\\
57	0.00223601976582295\\
59	0.00223602055707258\\
61	0.00223602118430721\\
63	0.0022360216852029\\
65	0.00223602208798562\\
67	0.00223602241398975\\
69	0.00223602267947583\\
71	0.00223602289693504\\
73	0.00223602307603454\\
75	0.00223602322430861\\
77	0.00223602334766849\\
79	0.00223602345078133\\
81	0.0022360235373542\\
83	0.0022360236103484\\
85	0.0022360236721424\\
87	0.00223602372465638\\
89	0.00223602376944829\\
91	0.00223602380778814\\
93	0.00223602384071582\\
95	0.00223602386908656\\
97	0.00223602389360662\\
99	0.00223602391486152\\
};
\addlegendentry{$h=3.5$m}

\end{axis}
\end{tikzpicture}%
\caption{\textit{(Two Dimension Model)} Here the variation of interference $\mathbb{I}_{n}(d_{x},d_{y})$, with respect to the number of interferers $n$, is drawn for different height $h$ of the LED installation. We consider $a=0.5$m, the half-power-semi-angle (HPSA) $\theta_{h}$ of the LED as $\frac{\pi}{3}$ radians with $d_{x}=d_{y}=0$.}
\label{two5h2}
\end{figure} 

As in the one dimension model, the approximation error depends on the ratio $\frac{h}{a}$. For larger values of $\frac{h}{a}$, 
we can choose $j=l=0$ leading to 
\begin{equation*}
\mathbb{I}_{\infty}(d_{x},d_{y}) \approx \hat {\mathtt{I}}_{0,0}(d_{x},d_{y}) =  \frac{h^{2-2\beta}\pi}{a^{2}(\beta-1)} -\frac{1}{(d_{x}^{2}+d_{y}^{2}+h^{2})^{\beta}},
\end{equation*}
which can be verified from Fig. \ref{space212}. We see that all the terms from $w=k=1$ have negligible contribution. \\ \par 

When $\frac{h}{a}$ is not large, a few more terms $(w,k)$ are necessary to improve the approximation accuracy. For example, in Fig. \ref{space2s2}, when we consider $h=2.5$m and $a=0.5$m, leading to $\frac{h}{a}=5$; $w=k=0$ and $w=k=1$ are significant, with an error bound, similar to that in one dimension, for $w>1$ and $k>1$ as $O(e^{-24\pi})$. So, $j=l=1$ or $\hat{\mathtt{I}}_{1,1}(d_{x},d_{y})$ is a good approximation in this case. Further, in most practical cases, the ratio $\frac{h}{a}$ varies between $2.5$ to $5$. So, the above approximation to $\hat{\mathtt{I}}_{1,1}(d_{x},d_{y})$ i.e.  
\begin{align*}
\mathbb{I}_{\infty}(d_{x},d_{y}) & \approx \hat {\mathtt{I}}_{1,1}(d_{x},d_{y}) \nonumber\\
&= \frac{h^{2-2\beta}\pi}{a^{2}(\beta-1)} -\frac{1}{(d_{x}^{2}+d_{y}^{2}+h^{2})^{\beta}}+g(1,1),
\end{align*}
can be extended in general to this practically seen range\footnote{For lower values of $\frac{h}{a}$ i.e. $<2.5$, higher values of $(j,l)$ may have to be considered to improve the approximation accuracy. Also, from the proof of Prop. \ref{prop2},  $(j,l)$ should be chosen such that  $\sqrt{j^{2}+l^{2}} \geq \lceil\frac{a (\beta-2.5)}{2\pi h}\rceil$ for a good approximation.} of $\frac{h}{a}$ because we still have a theoretical asymptotic error bound on $w>1$ and $k>1$ as $O(e^{-12\pi})$. 

\begin{figure}[ht]
\centering
\definecolor{mycolor1}{rgb}{0.00000,0.44700,0.74100}%
\begin{tikzpicture}
\begin{axis}[%
width=7cm,
height=5cm,
scale only axis,
xmin=0,
xmax=2,
xlabel style={font=\color{white!15!black}},
xlabel={Index ($w=k$) of individual terms},
ymode=log,
ymin=1e-39,
ymax=1,
yminorticks=true,
ylabel style={font=\color{white!15!black}},
ylabel={$|g(w,k)|$},
axis background/.style={fill=white},
title style={font=\bfseries},
title={},
xmajorgrids,
ymajorgrids,
yminorgrids
]
\addplot[ycomb, mark size=2.5pt, mark=*, mark options={solid, fill=black, black}, forget plot] table[row sep=crcr] {%
0	0.0171572846788051\\
1	7.9017911796913e-18\\
2	2.15817588260631e-36\\
3	2.96523518214323e-55\\
4	3.05963922695074e-74\\
5	2.69539479710128e-93\\
};
\end{axis}
\end{tikzpicture}%
\caption{\textit{(Two Dimension Model)} (Large $\frac{h}{a}$ case) This graph shows the magnitude of the individual terms of $|g(w,k)|$ for a height $h$ of the LED =$2.5$m and $a=0.2$m. The half-power-semi-angle (HPSA) $\theta_{h}=\frac{\pi}{3}$ radians and $d_{x}=d_{y}=0$.}  
\label{space212}
\end{figure}
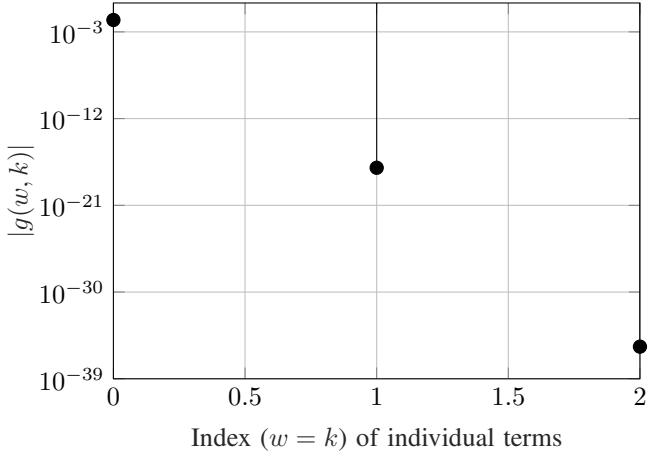 

\begin{figure}[ht]
\centering
\definecolor{mycolor1}{rgb}{0.00000,0.44700,0.74100}%
\begin{tikzpicture} 

\begin{axis}[%
width=7cm,
height=5cm,
scale only axis,
xmin=0,
xmax=3,
xlabel style={font=\color{white!15!black}},
xlabel={Index ($w=k$) of individual terms}, 
ymode=log,
ymin=1e-10,
ymax=0.01,
yminorticks=true,
ylabel style={font=\color{white!15!black}},
ylabel={$|g(w,k)|$},
axis background/.style={fill=white},
title style={font=\bfseries},
title={},
xmajorgrids,
ymajorgrids,
yminorgrids
]
\addplot[ycomb, mark size=2.5pt, mark=*, mark options={solid, fill=black, black}, forget plot] table[row sep=crcr] {%
0	0.000350148666914389\\
1	7.35483156610943e-05\\
2	5.38299696003349e-07\\
3	2.33483201755164e-09\\
4	7.95004105527043e-12\\
5	2.35374759441881e-14\\
};
\end{axis}
\end{tikzpicture}%
\caption{\textit{(Two Dimension Model)} This graph shows the magnitude of the individual terms of $|g(w,k)|$ for a height $h$ of the LED =$2.5$m and $a=0.5$m. The half-power-semi-angle (HPSA) $\theta_{h}=\frac{\pi}{3}$ radians and $d_{x}=d_{y}=0$.}
\label{space2s2}
\end{figure}
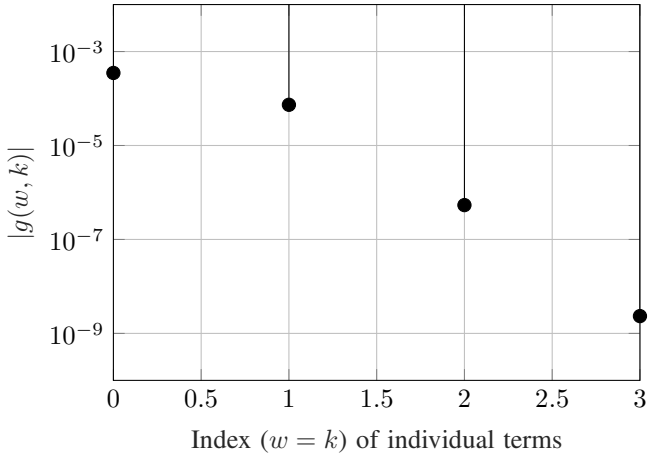 

For numerical validation, firstly, from Fig. \ref{two5a1}, we see the tightness of this approximation, for the above given range of $\frac{h}{a}$. We proceed by considering $h=2.5$m and plotting the interference w.r.t. the variation of $a$ from $0.1$m to $1$m $($i.e. $\frac{h}{a}$ in the range of $25$ to $2.5$$)$. We see that $\mathbb{I}_{n}(d_{x},d_{y})$ $($for $n=\{8, 15, 24, 35\}$$)$ and $\hat{\mathtt{I}}_{1,1}(d_{x},d_{y})$ are tightly bounded with each other, which validates our approximation. Now, from the above numerical validation for $j=l=1$, we take a given value of $a=0.5$m and proceed for further numerical validation w.r.t various system parameters $h, \theta_{h}$ and $z=\sqrt{d_{x}^{2}+d_{y}^{2}}$ in Fig. \ref{two5h1}, \ref{two5t1} and \ref{two5z1} respectively. The corresponding graphs for the approximation error $\hat{\xi}=|\mathbb{I}_{n}(d_{x},d_{y})-\hat{\mathtt{I}}_{1,1}(d_{x},d_{y})|$ are respectively shown in Fig. \ref{two5h3}, \ref{two5t3} and \ref{two5z3} for different number of interferers $n$. All the simulations are obtained using the parameter values given in Table I. \\ \par   

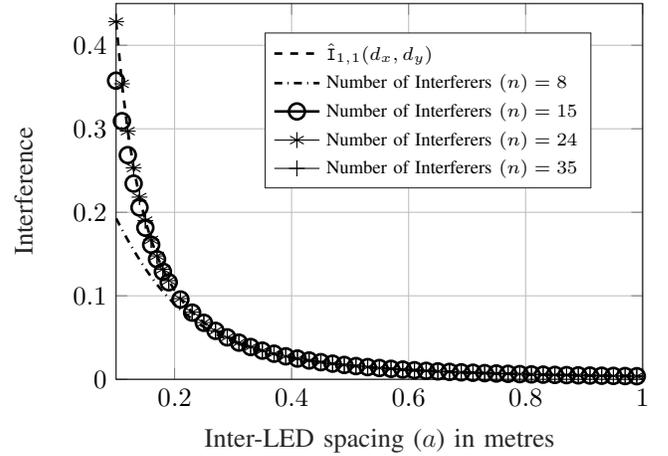
\begin{figure}[ht]
\centering
\definecolor{mycolor1}{rgb}{0.00000,0.44700,0.74100}%
\definecolor{mycolor2}{rgb}{0.85000,0.32500,0.09800}%
\definecolor{mycolor3}{rgb}{0.92900,0.69400,0.12500}%
\definecolor{mycolor4}{rgb}{0.49400,0.18400,0.55600}%
\definecolor{mycolor5}{rgb}{0.46600,0.67400,0.18800}%
\begin{tikzpicture}
\begin{axis}[%
width=7cm,
height=5cm,
scale only axis,
xmin=0.1,
xmax=1,
xlabel style={font=\color{white!15!black}},
xlabel={Inter-LED spacing ($a$) in metres},
ymin=0,
ymax=0.45,
ylabel style={font=\color{white!15!black}},
ylabel={Interference},
axis background/.style={fill=white},
title style={font=\bfseries},
title={},
xmajorgrids,
ymajorgrids,
legend style={font=\fontsize{7}{5}\selectfont,at={(0.2821,0.508)}, anchor=south west, legend cell align=left, align=left, draw=white!15!black}
]
\addplot [color=black, dashed, line width=1.0pt]
  table[row sep=crcr]{%
0.1	0.428276756970126\\
0.11	0.353833992867873\\
0.12	0.297214165673699\\
0.13	0.253150626372856\\
0.14	0.218187556821493\\
0.15	0.189981136431167\\
0.16	0.166896248191456\\
0.17	0.147764057636722\\
0.18	0.131731095854977\\
0.19	0.118162400933553\\
0.2	0.106577669242532\\
0.21	0.0966081585873302\\
0.22	0.0879669782169683\\
0.23	0.0804282159867914\\
0.24	0.0738120214184247\\
0.25	0.0679737787152202\\
0.26	0.062796136593214\\
0.27	0.0581830648244343\\
0.28	0.0540553692053733\\
0.29	0.0503472698418699\\
0.3	0.0470037641077918\\
0.31	0.0439785751685876\\
0.32	0.0412325420478639\\
0.33	0.0387323458742081\\
0.34	0.0364494944091805\\
0.35	0.0343595066914389\\
0.36	0.0324412539637443\\
0.37	0.0306764235624636\\
0.38	0.0290490802333883\\
0.39	0.0275453051525395\\
0.4	0.0261528973106329\\
0.41	0.0248611252450997\\
0.42	0.0236605196468326\\
0.43	0.0225426993277516\\
0.44	0.0215002245542421\\
0.45	0.0205264729367964\\
0.46	0.0196155339966979\\
0.47	0.0187621192652841\\
0.48	0.0179614853546062\\
0.49	0.0172093679037954\\
0.5	0.0165019246788051\\
0.51	0.0158356864040802\\
0.52	0.0152075141483035\\
0.53	0.0146145622844474\\
0.54	0.0140542462061087\\
0.55	0.0135242141147152\\
0.56	0.0130223223013439\\
0.57	0.0125466134370625\\
0.58	0.0120952974604694\\
0.59	0.0116667347133078\\
0.6	0.0112594210269538\\
0.61	0.0108719745060601\\
0.62	0.0105031237921636\\
0.63	0.0101516976208419\\
0.64	0.00981661551201041\\
0.65	0.00949687945498506\\
0.66	0.00919156646866323\\
0.67	0.00889982193312209\\
0.68	0.00862085360255774\\
0.69	0.00835392622115019\\
0.7	0.00809835667344816\\
0.71	0.00785350960947672\\
0.72	0.00761879349219252\\
0.73	0.0073936570213236\\
0.74	0.00717758589318271\\
0.75	0.00697009986085928\\
0.76	0.00677075006338225\\
0.77	0.00657911659609296\\
0.78	0.00639480629764955\\
0.79	0.00621745073186618\\
0.8	0.00604670434502747\\
0.81	0.00588224278145562\\
0.82	0.00572376134198669\\
0.83	0.00557097357166579\\
0.84	0.00542360996442912\\
0.85	0.0052814167738282\\
0.86	0.00514415491998999\\
0.87	0.00501159898401569\\
0.88	0.00488353628191483\\
0.89	0.00475976601096657\\
0.9	0.00464009846210652\\
0.91	0.00452435429256713\\
0.92	0.00441236385356124\\
0.93	0.00430396656829979\\
0.94	0.00419901035608319\\
0.95	0.00409735109860717\\
0.96	0.00399885214498396\\
0.97	0.00390338385230246\\
0.98	0.00381082315884127\\
0.99	0.00372105318730955\\
1	0.00363396287572548\\
};
\addlegendentry{$\hat{\mathtt{I}}_{1,1}(d_{x},d_{y})$}

\addplot [color=black, dashdotted, line width=1.0pt]
  table[row sep=crcr]{%
0.1	0.192717439305716\\
0.11	0.179513266746728\\
0.12	0.166788154129803\\
0.13	0.154670797154509\\
0.14	0.143244533753954\\
0.15	0.132555221488402\\
0.16	0.12261907008683\\
0.17	0.113429829467656\\
0.18	0.104965018482426\\
0.19	0.0971910851093122\\
0.2	0.0900675221816435\\
0.21	0.0835500389088892\\
0.22	0.077592923594259\\
0.23	0.0721507412109386\\
0.24	0.0671795018213884\\
0.25	0.0626374199777223\\
0.26	0.0584853663184354\\
0.27	0.0546870936496778\\
0.28	0.0512093025270683\\
0.29	0.0480215964731073\\
0.3	0.0450963646504129\\
0.31	0.0424086199246547\\
0.32	0.0399358125049774\\
0.33	0.0376576334038448\\
0.34	0.0355558174768412\\
0.35	0.0336139524828457\\
0.36	0.0318172981870547\\
0.37	0.0301526178002242\\
0.38	0.0286080228363955\\
0.39	0.0271728316442781\\
0.4	0.0258374413211581\\
0.41	0.0245932123742593\\
0.42	0.0234323652941138\\
0.43	0.0223478881041193\\
0.44	0.02133345391817\\
0.45	0.0203833475509026\\
0.46	0.0194924002661897\\
0.47	0.0186559318073737\\
0.48	0.0178696989192633\\
0.49	0.0171298496416481\\
0.5	0.0164328827233771\\
0.51	0.0157756115726134\\
0.52	0.015155132221343\\
0.53	0.0145687948398746\\
0.54	0.014014178389623\\
0.55	0.013489068049933\\
0.56	0.0129914350972578\\
0.57	0.0125194189529538\\
0.58	0.0120713111496404\\
0.59	0.0116455409958867\\
0.6	0.0112406627452929\\
0.61	0.0108553440992114\\
0.62	0.0104883558927323\\
0.63	0.0101385628314705\\
0.64	0.00980491516241229\\
0.65	0.00948644117588049\\
0.66	0.00918224044778803\\
0.67	0.00889147774197301\\
0.68	0.00861337750173418\\
0.69	0.00834721886787413\\
0.7	0.00809233116774655\\
0.71	0.00784808982612478\\
0.72	0.00761391265426475\\
0.73	0.00738925647842599\\
0.74	0.00717361407342127\\
0.75	0.00696651137056054\\
0.76	0.00676750491270401\\
0.77	0.00657617953209557\\
0.78	0.00639214622926113\\
0.79	0.00621504023356796\\
0.8	0.00604451922808823\\
0.81	0.00588026172322474\\
0.82	0.00572196556516692\\
0.83	0.00556934656667486\\
0.84	0.00542213724896166\\
0.85	0.00528008568457515\\
0.86	0.00514295443218938\\
0.87	0.00501051955511497\\
0.88	0.00488256971614096\\
0.89	0.004758905342038\\
0.9	0.00463933785169522\\
0.91	0.00452368894243798\\
0.92	0.00441178992958861\\
0.93	0.004303481134796\\
0.94	0.00419861131907379\\
0.95	0.0040970371568619\\
0.96	0.0039986227477611\\
0.97	0.00390323916289373\\
0.98	0.00381076402311668\\
0.99	0.00372108110655865\\
1	0.00363407998317669\\
};
\addlegendentry{Number of Interferers $(n) = 8$}

\addplot [color=black, line width=1.0pt, draw=none, mark=o, mark size=3pt]
  table[row sep=crcr]{%
0.1	0.357637400924452\\
0.11	0.309128881165398\\
0.12	0.268442936067229\\
0.13	0.234337307131605\\
0.14	0.205697786030066\\
0.15	0.18156851230257\\
0.16	0.161150889960059\\
0.17	0.143788112228852\\
0.18	0.128944737758266\\
0.19	0.116186119620051\\
0.21	0.0955798467133086\\
0.23	0.0798705538524652\\
0.25	0.067659837475335\\
0.27	0.0580002609366649\\
0.29	0.0502375267359334\\
0.31	0.0439108461560912\\
0.33	0.0386894840116805\\
0.35	0.0343317551945349\\
0.37	0.0306580770511102\\
0.39	0.0275329427335879\\
0.41	0.0248526480224367\\
0.43	0.0225367919290417\\
0.45	0.020522294734991\\
0.47	0.0187591232227014\\
0.49	0.017207192014101\\
0.51	0.0158340873450692\\
0.53	0.0146133741189413\\
0.55	0.0135233221406775\\
0.57	0.0125459373647494\\
0.59	0.0116662176819387\\
0.61	0.0108715758150815\\
0.63	0.0101513878603826\\
0.65	0.00949663720913748\\
0.67	0.00889963152976773\\
0.69	0.00835377618886677\\
0.71	0.00785339161992424\\
0.73	0.00739356517073691\\
0.75	0.00697003019278643\\
0.77	0.00657906679904674\\
0.79	0.00621741996617518\\
0.81	0.00588223160313221\\
0.83	0.00557098393000583\\
0.85	0.00528145206531735\\
0.87	0.00501166414898315\\
0.89	0.00475986766200234\\
0.91	0.00452450086546295\\
0.93	0.00430416848750544\\
0.95	0.00409762095011597\\
0.97	0.00390373655763172\\
0.99	0.00372150617290713\\
};
\addlegendentry{Number of Interferers $(n) = 15$}

\addplot [color=black, draw=none, mark=asterisk, mark size=3pt]
  table[row sep=crcr]{%
0.1	0.428224472227695\\
0.11	0.353808927850001\\
0.12	0.297201406041115\\
0.13	0.253143789559718\\
0.14	0.218183728138665\\
0.15	0.189978908305202\\
0.16	0.166894907019389\\
0.17	0.147763225891849\\
0.18	0.131730566146824\\
0.19	0.118162055453865\\
0.21	0.0966080021806876\\
0.23	0.0804281399778388\\
0.25	0.0679737395200422\\
0.27	0.0581830435690644\\
0.29	0.0503472578054394\\
0.31	0.0439785680916021\\
0.33	0.0387323415738555\\
0.35	0.034359504001143\\
0.37	0.0306764218352546\\
0.39	0.0275453040176224\\
0.41	0.0248611244836239\\
0.43	0.0225426988070924\\
0.45	0.0205264725746672\\
0.47	0.0187621190096019\\
0.49	0.0172093677211853\\
0.51	0.0158356862733967\\
0.53	0.0146145621935131\\
0.55	0.013524214059566\\
0.57	0.0125466134237077\\
0.59	0.0116667347625476\\
0.61	0.010871974662387\\
0.63	0.0101516979668447\\
0.65	0.009496880133025\\
0.67	0.00889982317628046\\
0.69	0.00835392839556157\\
0.71	0.007853513270138\\
0.73	0.00739366298324621\\
0.75	0.00697010928709807\\
0.77	0.00657913110365801\\
0.79	0.00621747251591194\\
0.81	0.00588227475743671\\
0.83	0.00557101953441233\\
0.85	0.00528148157177062\\
0.87	0.00501168870611105\\
0.89	0.00475988818329347\\
0.91	0.00452451808108472\\
0.93	0.00430418298392424\\
0.95	0.0040976332005677\\
0.97	0.0039037469456859\\
0.99	0.00372151501078669\\
};
\addlegendentry{Number of Interferers $(n) = 24$}

\addplot [color=black, draw=none, mark=+, mark size=3pt]
  table[row sep=crcr]{%
0.1	0.428275822339653\\
0.11	0.353833553740597\\
0.12	0.297213945564158\\
0.13	0.253150509858501\\
0.14	0.218187492202243\\
0.15	0.189981099120304\\
0.16	0.166896225877767\\
0.17	0.147764043872846\\
0.18	0.131731087128806\\
0.19	0.11816239526406\\
0.21	0.0966081560362628\\
0.23	0.0804282147527452\\
0.25	0.0679737780811191\\
0.27	0.0581830644815249\\
0.29	0.0503472696481244\\
0.31	0.0439785750548772\\
0.33	0.0387323458052167\\
0.35	0.0343595066483322\\
0.37	0.0306764235348185\\
0.39	0.0275453051343909\\
0.41	0.024861125232936\\
0.43	0.0225426993194555\\
0.45	0.0205264729310908\\
0.47	0.0187621192614758\\
0.49	0.0172093679017614\\
0.51	0.0158356864045854\\
0.53	0.0146145622899979\\
0.55	0.0135242141313369\\
0.57	0.0125466134776614\\
0.59	0.0116667348035069\\
0.61	0.0108719746937676\\
0.63	0.0101516979910938\\
0.65	0.00949688015191509\\
0.67	0.00889982319110708\\
0.69	0.0083539284072819\\
0.71	0.0078535132794648\\
0.73	0.00739366299071605\\
0.75	0.0069701092931164\\
0.77	0.00657913110853453\\
0.79	0.00621747251988438\\
0.81	0.00588227476068933\\
0.83	0.00557101953708838\\
0.85	0.00528148157398263\\
0.87	0.00501168870794735\\
0.89	0.00475988818482432\\
0.91	0.00452451808236617\\
0.93	0.00430418298500096\\
0.95	0.0040976332014757\\
0.97	0.00390374694645509\\
0.99	0.00372151501143993\\
};
\addlegendentry{Number of Interferers $(n) = 35$}

\end{axis}
\end{tikzpicture}%
\caption{\textit{(Two Dimension Model)} Here the variation of interference $\mathbb{I}_{n}(d_{x},d_{y})$, is drawn with respect to a linear variation of the inter-LED spacing $a$ for different number interferers $n$ in the network. The graph for the proposed interference expression $\hat{\mathtt{I}}_{1,1}(d_{x},d_{y})$ from approximation is also drawn. We consider height $h$ of the LEDs as $2.5$m, the half-power-semi-angle (HPSA) $\theta_{h}$ of the LED as $\frac{\pi}{3}$ radians and the position of the receiver photodiode (PD) as $d_{x}=d_{y}=0$.} 
\label{two5a1}
\end{figure} 

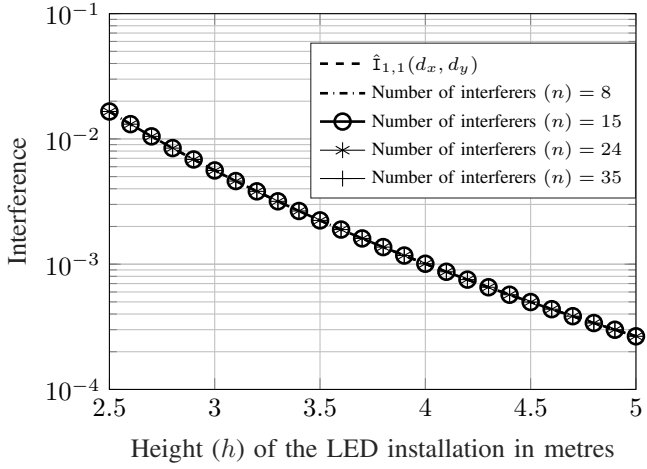
\begin{figure}[ht]
\centering
\begin{tikzpicture}
\begin{axis}[%
width=7cm,
height=5cm,
scale only axis,
xmin=2.5,
xmax=5,
xlabel style={font=\color{white!15!black}},
xlabel={Height ($h$) of the LED installation in metres},
ymode=log,
ymin=0.0001,
ymax=0.1,
yminorticks=true,
ylabel style={font=\color{white!15!black}},
ylabel={Interference},
axis background/.style={fill=white},
xmajorgrids,
ymajorgrids,
yminorgrids,
legend style={font=\fontsize{7}{5}\selectfont,at={(0.3815,0.508)}, anchor=south west, legend cell align=left, align=left, draw=white!15!black}
] 
\addplot [color=black, dashed, line width=1.0pt]
  table[row sep=crcr]{%
2.5	0.0165518333404043\\
2.6	0.0131146264956068\\
2.7	0.0104812073944158\\
2.8	0.00844395414516201\\
2.9	0.00685362905240715\\
3	0.0056017056387644\\
3.1	0.00460840276388968\\
3.2	0.00381447401409158\\
3.3	0.00317550323154716\\
3.4	0.00265789477758397\\
3.5	0.00223602407967442\\
3.6	0.00189019181899419\\
3.7	0.00160514083943355\\
3.8	0.00136897114967819\\
3.9	0.00117233929564977\\
4	0.0010078627370146\\
4.1	0.000869673299612636\\
4.2	0.000753079928950126\\
4.3	0.0006543122105183\\
4.4	0.000570324016709433\\
4.5	0.000498642232344533\\
4.6	0.000437249505422545\\
4.7	0.000384492845697551\\
4.8	0.000339011979981266\\
4.9	0.000299682897419951\\
5	0.000265573139450774\\
};
\addlegendentry{$\hat{\mathtt{I}}_{1,1}(d_{x},d_{y})$}

\addplot [color=black, dashdotted, line width=1.0pt]
  table[row sep=crcr]{%
2.5	0.0165517980107116\\
2.6	0.0131145912828143\\
2.7	0.0104811723016383\\
2.8	0.00844391917611778\\
2.9	0.0068535942109163\\
3	0.00560167092862447\\
3.1	0.00460836818884424\\
3.2	0.00381443957782103\\
3.3	0.00317546893766578\\
3.4	0.00265786062963791\\
3.5	0.00223599008114046\\
3.6	0.00189015797327852\\
3.7	0.00160510714987068\\
3.8	0.00136893761952987\\
3.9	0.00117230592810404\\
4	0.00100782953518489\\
4.1	0.000869640266536978\\
4.2	0.000753047067590351\\
4.3	0.000654279523759319\\
4.4	0.000570291507358673\\
4.5	0.000498609903131323\\
4.6	0.000437217358997667\\
4.7	0.00038446088463278\\
4.8	0.000338980206769036\\
4.9	0.000299651314473053\\
5	0.000265541749102118\\
};
\addlegendentry{Number of interferers $(n) = 8$}

\addplot [color=black, line width=1.0pt, draw=none, mark=o, mark size=3pt]
  table[row sep=crcr]{%
2.5	0.0165518238103626\\
2.6	0.0131146169869172\\
2.7	0.0104811979069293\\
2.8	0.00844394467938558\\
2.9	0.00685361960900367\\
3	0.00560169621842934\\
3.1	0.00460839336732144\\
3.2	0.00381446464198302\\
3.3	0.00317549388458416\\
3.4	0.00265788545644446\\
3.5	0.00223601478502808\\
3.6	0.00189018255150246\\
3.7	0.00160513159974918\\
3.8	0.00136896193844523\\
3.9	0.00117233011350337\\
4	0.00100785358458089\\
4.1	0.000869664177508527\\
4.2	0.000753070837783203\\
4.3	0.000654303150886653\\
4.4	0.000570314989201587\\
4.5	0.000498633237539223\\
4.6	0.000437240543888698\\
4.7	0.000384483917994114\\
4.8	0.000339003086657121\\
4.9	0.000299674039013804\\
5	0.00026556431649107\\
};
\addlegendentry{Number of interferers $(n) = 15$}

\addplot [color=black, draw=none, mark=asterisk, mark size=3pt]
  table[row sep=crcr]{%
2.5	0.016551833182614\\
2.6	0.0131146263390678\\
2.7	0.010481207238252\\
2.8	0.00844395398915979\\
2.9	0.00685362889651659\\
3	0.00560170548297532\\
3.1	0.00460840260820198\\
3.2	0.00381447385850759\\
3.3	0.00317550307607\\
3.4	0.00265789462221672\\
3.5	0.00223602392442014\\
3.6	0.00189019166385613\\
3.7	0.0016051406844148\\
3.8	0.00136897099478188\\
3.9	0.00117233914087906\\
4	0.00100786258237264\\
4.1	0.000869673145102489\\
4.2	0.000753079774574888\\
4.3	0.000654312056281042\\
4.4	0.000570323862613304\\
4.5	0.000498642078392522\\
4.6	0.00043724935161775\\
4.7	0.000384492692042977\\
4.8	0.000339011826479957\\
4.9	0.000299682744074909\\
5	0.000265572986265005\\
};
\addlegendentry{Number of interferers $(n) = 24$}

\addplot [color=black, draw=none, mark=+, mark size=3pt]
  table[row sep=crcr]{%
2.5	0.0165518333363651\\
2.6	0.013114626492735\\
2.7	0.0104812073918322\\
2.8	0.00844395414264954\\
2.9	0.00685362904991322\\
3	0.00560170563627497\\
3.1	0.00460840276140183\\
3.2	0.00381447401160437\\
3.3	0.00317550322906051\\
3.4	0.00265789477509781\\
3.5	0.00223602407718878\\
3.6	0.00189019181650897\\
3.7	0.00160514083694876\\
3.8	0.00136897114719384\\
3.9	0.00117233929316602\\
4	0.0010078627345314\\
4.1	0.000869673297129976\\
4.2	0.000753079926467985\\
4.3	0.000654312208036699\\
4.4	0.000570324014228418\\
4.5	0.000498642229864104\\
4.6	0.000437249502942759\\
4.7	0.00038449284321835\\
4.8	0.000339011977502705\\
4.9	0.000299682894942015\\
5	0.000265573136973497\\
};
\addlegendentry{Number of interferers $(n) = 35$}
\end{axis}
\end{tikzpicture}%
\caption{\textit{(Two Dimension Model)} Here the variation of interference $\mathbb{I}_{n}(d_{x},d_{y})$, is drawn with respect to a linear variation of the height $h$ of installation of the LED for different number interferers $n$ in the network. The graph for the proposed interference expression $\hat{\mathtt{I}}_{1,1}(d_{x},d_{y})$ from approximation is also drawn. We consider $a=0.5$m, the half-power-semi-angle (HPSA) $\theta_{h}$ of the LED as $\frac{\pi}{3}$ radians and $d_{x}=d_{y}=0$.}
\label{two5h1}
\end{figure} 

\begin{figure}[ht]
\centering
\begin{tikzpicture}
\begin{axis}[%
width=7cm,
height=5cm,
scale only axis,
xmin=0.2,
xmax=1.6,
xlabel style={font=\color{white!15!black}},
xlabel={HPSA ($\theta_{h}$) of the LED in radians},
ymode=log,
ymin=1e-18,
ymax=1,
yminorticks=true,
ylabel style={font=\color{white!15!black}},
ylabel={Interference},
axis background/.style={fill=white},
title style={font=\bfseries},
title={},
xmajorgrids,
ymajorgrids,
yminorgrids,
legend style={font=\fontsize{7}{5}\selectfont,at={(0.3815,0.308)}, anchor=south west, legend cell align=left, align=left, draw=white!15!black}
]
\addplot [color=black, dashed, line width=1.0pt]
  table[row sep=crcr]{%
0.261799387799149	1.47103372518953e-18\\
0.278161849536596	1.1481862342124e-16\\
0.294524311274043	4.47686452836803e-15\\
0.31088677301149	1.00422037101669e-13\\
0.327249234748937	1.44218045020432e-12\\
0.343611696486384	1.43960736136347e-11\\
0.35997415822383	1.06490426941754e-10\\
0.376336619961277	6.13964874149571e-10\\
0.392699081698724	2.87222632432445e-09\\
0.409061543436171	1.12617412678585e-08\\
0.425424005173618	3.79958078577056e-08\\
0.441786466911065	1.12713439498407e-07\\
0.458148928648511	2.99267993609752e-07\\
0.474511390385958	7.21809526965353e-07\\
0.490873852123405	1.60121407551202e-06\\
0.507236313860852	3.30125625775653e-06\\
0.539961237335746	1.16564059777336e-05\\
0.572686160810639	3.36197866056146e-05\\
0.605411084285533	8.25853981300831e-05\\
0.638136007760427	0.000178296449490818\\
0.67086093123532	0.000346543477970743\\
0.687223392972767	0.000467263976247277\\
0.703585854710214	0.000617807402665058\\
0.719948316447661	0.000802373953963926\\
0.752673239922555	0.00129037817654013\\
0.769035701660001	0.00160200402039324\\
0.785398163397448	0.00196390140249179\\
0.801760625134895	0.00237967782894044\\
0.834485548609789	0.00338586382671195\\
0.850848010347236	0.00398196545755304\\
0.867210472084682	0.00464329271335614\\
0.883572933822129	0.00537181641236949\\
0.916297857297023	0.00703657120673012\\
0.93266031903447	0.00797500952276171\\
0.949022780771917	0.00898509160608491\\
0.965385242509364	0.0100671506897093\\
0.998110165984257	0.0124472359121995\\
1.0144726277217	0.0137447177180561\\
1.03083508945915	0.0151131487721609\\
1.0471975511966	0.0165518333404043\\
1.07992247467149	0.0196366545512726\\
1.09628493640894	0.0212809728026119\\
1.11264739814639	0.0229919742435761\\
1.12900985988383	0.0247687358062902\\
1.16173478335873	0.0285161638128047\\
1.17809724509617	0.0304854105870895\\
1.19445970683362	0.0325176564161934\\
1.21082216857107	0.0346126414252318\\
1.24354709204596	0.0389911663463146\\
1.25990955378341	0.041275742615426\\
1.27627201552085	0.0436252761726089\\
1.2926344772583	0.0460415075622543\\
1.32535940073319	0.0510846025188619\\
1.34172186247064	0.053719035963288\\
1.35808432420809	0.0564357791124501\\
1.37444678594553	0.0592421289002725\\
1.40717170942043	0.0651645566920352\\
1.42353417115788	0.0683095403159115\\
1.43989663289532	0.0716047565901068\\
1.45625909463277	0.0750809549567169\\
1.48898401810766	0.0827744061821478\\
1.50534647984511	0.0871636214313759\\
1.52170894158256	0.0921371844074186\\
1.53807140332	0.0980840065208722\\
1.5707963267949	0.15030927546806\\
};
\addlegendentry{$\hat{\mathtt{I}}_{1,1}(d_{x},d_{y})$}

\addplot [color=black, dashdotted, line width=1.0pt]
  table[row sep=crcr]{%
0.261799387799149	1.47058900271588e-18\\
0.278161849536596	1.14803615288118e-16\\
0.294524311274043	4.47660344115122e-15\\
0.31088677301149	1.00419332127004e-13\\
0.327249234748937	1.44216184384288e-12\\
0.343611696486384	1.43959813003491e-11\\
0.35997415822383	1.06490074805095e-10\\
0.376336619961277	6.13963788715714e-10\\
0.392699081698724	2.87222351296839e-09\\
0.409061543436171	1.12617349568099e-08\\
0.425424005173618	3.79957952712007e-08\\
0.441786466911065	1.1271341674876e-07\\
0.458148928648511	2.99267955735722e-07\\
0.474511390385958	7.21809468097598e-07\\
0.490873852123405	1.6012139890158e-06\\
0.507236313860852	3.30125613575967e-06\\
0.539961237335746	1.16564057298448e-05\\
0.556323699073193	2.0244398450728e-05\\
0.572686160810639	3.36197857364528e-05\\
0.589048622548086	5.36442678403928e-05\\
0.62177354602298	0.000123116226812443\\
0.638136007760427	0.000178296415672219\\
0.654498469497873	0.000251535838173289\\
0.67086093123532	0.000346543308578172\\
0.703585854710214	0.00061780671167056\\
0.719948316447661	0.000802372649933278\\
0.736310778185108	0.00102517621596422\\
0.752673239922555	0.00129037404807235\\
0.785398163397448	0.00196388999668795\\
0.801760625134895	0.00237965969379658\\
0.818123086872342	0.00285263248657338\\
0.834485548609789	0.00338582139869689\\
0.867210472084682	0.00464320205744454\\
0.883572933822129	0.00537168780063106\\
0.899935395559576	0.00616897346069072\\
0.916297857297023	0.0070363256761599\\
0.949022780771917	0.00898465126056677\\
0.965385242509364	0.0100665731275865\\
0.98174770424681	0.01122050699019\\
0.998110165984257	0.0124462785100471\\
1.03083508945915	0.0151116304054569\\
1.0471975511966	0.0165499493841225\\
1.06356001293404	0.0180576475525587\\
1.07992247467149	0.0196338303613226\\
1.11264739814639	0.0229878735285229\\
1.12900985988383	0.0247638473150611\\
1.14537232162128	0.0266046011859129\\
1.16173478335873	0.0285093496162983\\
1.19445970683362	0.0325083752706303\\
1.21082216857107	0.0346018936768183\\
1.22718463030851	0.0367579825177434\\
1.24354709204596	0.0389769523508766\\
1.25990955378341	0.0412594994237208\\
1.2926344772583	0.0460205303841197\\
1.30899693899575	0.0485031411579287\\
1.32535940073319	0.051057858335498\\
1.34172186247064	0.0536889616031371\\
1.37444678594553	0.0592043482796888\\
1.39080924768298	0.0621053421749699\\
1.40717170942043	0.0651173972076567\\
1.42353417115788	0.0682569207986941\\
1.45625909463277	0.0750154258579269\\
1.47262155637022	0.0787088917738026\\
1.48898401810766	0.0826922561543533\\
1.50534647984511	0.0870710382534679\\
1.52170894158256	0.0920319297471495\\
1.55443386505745	0.105960377445425\\
1.5707963267949	0.149985059637861\\
};
\addlegendentry{Number of interferers $(n) = 8$}

\addplot [color=black, line width=1.0pt, draw=none, mark=o, mark size=3pt]
  table[row sep=crcr]{%
0.261799387799149	1.47058900271588e-18\\
0.278161849536596	1.14803615288118e-16\\
0.294524311274043	4.47660344115122e-15\\
0.31088677301149	1.00419332127004e-13\\
0.327249234748937	1.44216184384288e-12\\
0.343611696486384	1.43959813003491e-11\\
0.35997415822383	1.06490074805095e-10\\
0.376336619961277	6.13963788715715e-10\\
0.392699081698724	2.87222351296842e-09\\
0.409061543436171	1.12617349568112e-08\\
0.441786466911065	1.12713416749169e-07\\
0.474511390385958	7.21809468139065e-07\\
0.507236313860852	3.30125613755779e-06\\
0.539961237335746	1.16564057705209e-05\\
0.572686160810639	3.36197862902369e-05\\
0.605411084285533	8.25853976923096e-05\\
0.638136007760427	0.00017829644892111\\
0.67086093123532	0.000346543477242935\\
0.703585854710214	0.000617807401622568\\
0.736310778185108	0.00102517857746111\\
0.769035701660001	0.00160200401425546\\
0.801760625134895	0.00237967780929411\\
0.834485548609789	0.00338586376731222\\
0.867210472084682	0.00464329255037043\\
0.899935395559576	0.00616915225282501\\
0.93266031903447	0.00797500859785292\\
0.965385242509364	0.0100671487389452\\
0.998110165984257	0.012447232068405\\
1.03083508945915	0.0151131416363659\\
1.06356001293404	0.0180599514546102\\
1.09628493640894	0.0212809516426507\\
1.12900985988383	0.0247687016026475\\
1.16173478335873	0.0285161104522465\\
1.19445970683362	0.0325175757251429\\
1.22718463030851	0.0367702510740758\\
1.25990955378341	0.0412755720246603\\
1.2926344772583	0.046041267438368\\
1.32535940073319	0.0510842703334092\\
1.35808432420809	0.0564353260152265\\
1.39080924768298	0.0621469663077009\\
1.42353417115788	0.0683087203383014\\
1.45625909463277	0.0750798561503159\\
1.48898401810766	0.0827729208091241\\
1.52170894158256	0.0921351176065053\\
1.55443386505745	0.106102962436711\\
};
\addlegendentry{Number of interferers $(n) = 15$}

\addplot [color=black, draw=none, mark=asterisk, mark size=3pt]
  table[row sep=crcr]{%
0.261799387799149	1.47058900271588e-18\\
0.278161849536596	1.14803615288118e-16\\
0.294524311274043	4.47660344115122e-15\\
0.31088677301149	1.00419332127004e-13\\
0.327249234748937	1.44216184384288e-12\\
0.343611696486384	1.43959813003491e-11\\
0.35997415822383	1.06490074805095e-10\\
0.376336619961277	6.13963788715715e-10\\
0.392699081698724	2.87222351296842e-09\\
0.409061543436171	1.12617349568112e-08\\
0.425424005173618	3.79957952712278e-08\\
0.441786466911065	1.12713416749169e-07\\
0.458148928648511	2.99267955740372e-07\\
0.474511390385958	7.21809468139065e-07\\
0.490873852123405	1.60121398931521e-06\\
0.507236313860852	3.3012561375578e-06\\
0.539961237335746	1.16564057705211e-05\\
0.572686160810639	3.36197862902496e-05\\
0.589048622548086	5.3644269588129e-05\\
0.62177354602298	0.000123116240239272\\
0.638136007760427	0.000178296448924811\\
0.67086093123532	0.00034654347727635\\
0.687223392972767	0.00046726397549026\\
0.719948316447661	0.000802373953086452\\
0.736310778185108	0.00102517857864294\\
0.769035701660001	0.00160200401933714\\
0.785398163397448	0.00196390140136452\\
0.818123086872342	0.00285266056554228\\
0.834485548609789	0.00338586382515087\\
0.867210472084682	0.00464329271088878\\
0.883572933822129	0.00537181640897474\\
0.916297857297023	0.00703657119946672\\
0.93266031903447	0.00797500951182428\\
0.965385242509364	0.0100671506649178\\
0.98174770424681	0.0112212549710455\\
1.0144726277217	0.0137447176396488\\
1.03083508945915	0.0151131486601575\\
1.06356001293404	0.0180599638083177\\
1.07992247467149	0.0196366542503682\\
1.11264739814639	0.022991973697415\\
1.12900985988383	0.0247687350828445\\
1.16173478335873	0.028516162581925\\
1.17809724509617	0.030485409004158\\
1.21082216857107	0.0346126388732971\\
1.22718463030851	0.0367703666017547\\
1.25990955378341	0.0412757376758079\\
1.27627201552085	0.0436252700956335\\
1.30899693899575	0.0485268525823887\\
1.32535940073319	0.0510845915611784\\
1.35808432420809	0.0564357632320597\\
1.37444678594553	0.0592421098793467\\
1.40717170942043	0.0651645295994499\\
1.42353417115788	0.0683095080521833\\
1.45625909463277	0.0750809091862007\\
1.47262155637022	0.0787820997542101\\
1.50534647984511	0.0871635420666611\\
1.52170894158256	0.0921370870682187\\
1.55443386505745	0.106105988470199\\
1.5707963267949	0.150308694229602\\
};
\addlegendentry{Number of interferers $(n) = 24$}

\addplot [color=black, draw=none, mark=+, mark size=3pt]
  table[row sep=crcr]{%
0.261799387799149	1.47058900271588e-18\\
0.278161849536596	1.14803615288118e-16\\
0.294524311274043	4.47660344115122e-15\\
0.31088677301149	1.00419332127004e-13\\
0.327249234748937	1.44216184384288e-12\\
0.343611696486384	1.43959813003491e-11\\
0.35997415822383	1.06490074805095e-10\\
0.376336619961277	6.13963788715715e-10\\
0.392699081698724	2.87222351296842e-09\\
0.409061543436171	1.12617349568112e-08\\
0.425424005173618	3.79957952712278e-08\\
0.441786466911065	1.12713416749169e-07\\
0.458148928648511	2.99267955740372e-07\\
0.474511390385958	7.21809468139065e-07\\
0.490873852123405	1.60121398931521e-06\\
0.507236313860852	3.3012561375578e-06\\
0.539961237335746	1.16564057705211e-05\\
0.556323699073193	2.02443986094842e-05\\
0.572686160810639	3.36197862902496e-05\\
0.589048622548086	5.3644269588129e-05\\
0.62177354602298	0.000123116240239273\\
0.638136007760427	0.000178296448924813\\
0.67086093123532	0.000346543477276372\\
0.687223392972767	0.000467263975490336\\
0.703585854710214	0.000617807401847322\\
0.736310778185108	0.00102517857864574\\
0.752673239922555	0.00129037817555323\\
0.785398163397448	0.00196390140140433\\
0.801760625134895	0.00237967782780678\\
0.818123086872342	0.00285266056571719\\
0.834485548609789	0.00338586382549318\\
0.867210472084682	0.00464329271205744\\
0.883572933822129	0.00537181641103024\\
0.899935395559576	0.00616915265718966\\
0.916297857297023	0.00703657120529402\\
0.949022780771917	0.00898509160448835\\
0.965385242509364	0.0100671506879825\\
0.998110165984257	0.012447235910002\\
1.0144726277217	0.0137447177154512\\
1.03083508945915	0.015113148768965\\
1.0471975511966	0.0165518333363651\\
1.07992247467149	0.0196366545443812\\
1.11264739814639	0.0229919742312965\\
1.12900985988383	0.0247687357898469\\
1.14537232162128	0.0266103905701674\\
1.17809724509617	0.0304854105484533\\
1.21082216857107	0.0346126413590002\\
1.22718463030851	0.0367703697184861\\
1.24354709204596	0.0389911662358922\\
1.25990955378341	0.0412757424742926\\
1.2926344772583	0.046041507336014\\
1.30899693899575	0.0485268613452274\\
1.34172186247064	0.0537190355226242\\
1.35808432420809	0.0564357785670876\\
1.37444678594553	0.0592421282277533\\
1.40717170942043	0.0651645556775981\\
1.42353417115788	0.068309539073045\\
1.43989663289532	0.071604755067809\\
1.47262155637022	0.0787821521315878\\
1.50534647984511	0.0871636178907117\\
1.52170894158256	0.092137179917689\\
1.53807140332	0.0980840006556504\\
1.5707963267949	0.150309239592377\\
};
\addlegendentry{Number of interferers $(n) = 35$}
\end{axis}
\end{tikzpicture}%
\caption{\textit{(Two Dimension Model)} Here the variation of interference $\mathbb{I}_{n}(d_{x},d_{y})$, is drawn with respect to a linear variation of the half-power-semi-angle (HPSA) $\theta_{h}$ of the LED for different number interferers $n$ in the network. The graph for the proposed interference expression $\hat{\mathtt{I}}_{1,1}(d_{x},d_{y})$ from approximation is also drawn. We consider the attocell length $a=0.5$m, the height $h$ of the LED as $2.5$m and $d_{x}=d_{y}=0$.}
\label{two5t1}
\end{figure}
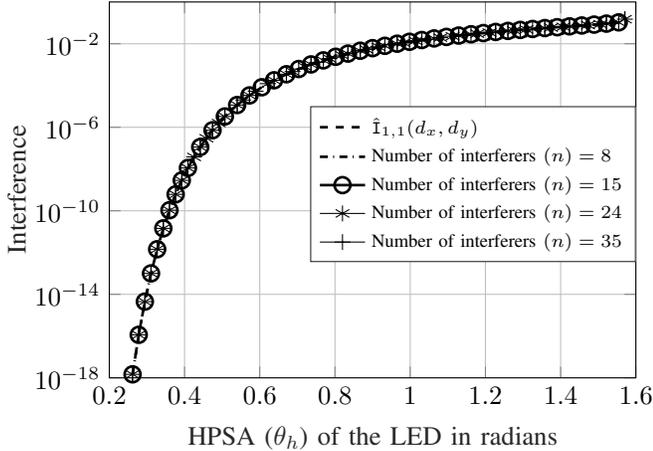 

\begin{figure}[ht]
\centering
\begin{tikzpicture}
\begin{axis}[%
width=7cm,
height=5cm,
scale only axis,
xmin=0,
xmax=0.355,
xlabel style={font=\color{white!15!black}},
xlabel={Distance ($z$) of the PD from origin in metres},
ymin=0.0165,
ymax=0.01656,
ylabel style={font=\color{white!15!black}},
ylabel={Interference},
axis background/.style={fill=white},
title style={font=\bfseries},
title={},
xmajorgrids,
ymajorgrids,
legend style={font=\fontsize{7}{5}\selectfont,at={(0.0815,0.4900)}, anchor=south west, legend cell align=left, align=left, draw=white!15!black}
]
\addplot [color=black, dashed, line width=1.0pt]
  table[row sep=crcr]{%
0	0.0165019246788051\\
0.014142135623731	0.0165020085581746\\
0.0282842712474619	0.0165022601157784\\
0.0424264068711929	0.0165026791102562\\
0.0565685424949238	0.0165032651398558\\
0.0707106781186548	0.0165040176432026\\
0.0848528137423857	0.0165049359003746\\
0.0989949493661167	0.0165060190342784\\
0.113137084989848	0.0165072660123221\\
0.127279220613579	0.0165086756483805\\
0.14142135623731	0.0165102466050449\\
0.15556349186104	0.016511977396153\\
0.169705627484771	0.0165138663895889\\
0.183847763108502	0.0165159118103459\\
0.197989898732233	0.0165181117438417\\
0.212132034355964	0.0165204641394776\\
0.226274169979695	0.0165229668144291\\
0.240416305603426	0.016525617457658\\
0.254558441227157	0.0165284136341327\\
0.268700576850888	0.0165313527892464\\
0.282842712474619	0.0165344322534176\\
0.29698484809835	0.0165376492468614\\
0.311126983722081	0.016541000884518\\
0.325269119345812	0.0165444841811228\\
0.339411254969543	0.0165480960564057\\
0.353553390593274	0.0165518333404043\\
};
\addlegendentry{$\hat{\mathtt{I}}_{1,1}(d_{x},d_{y})$}

\addplot [color=black, dashdotted, line width=1.0pt]
  table[row sep=crcr]{%
0	0.0165018894638775\\
0.014142135623731	0.0165019733430563\\
0.0282842712474619	0.0165022249000881\\
0.0424264068711929	0.0165026438936133\\
0.0565685424949238	0.0165032299218805\\
0.0707106781186548	0.0165039824235162\\
0.0848528137423857	0.0165049006786005\\
0.0989949493661167	0.0165059838100415\\
0.113137084989848	0.0165072307852497\\
0.127279220613579	0.0165086404181025\\
0.14142135623731	0.0165102113711934\\
0.15556349186104	0.0165119421583631\\
0.169705627484771	0.0165138311474989\\
0.183847763108502	0.0165158765635965\\
0.197989898732233	0.0165180764920771\\
0.212132034355964	0.0165204288823447\\
0.226274169979695	0.0165229315515778\\
0.240416305603426	0.0165255821887409\\
0.254558441227157	0.0165283783588049\\
0.268700576850888	0.0165313175071655\\
0.282842712474619	0.0165343969642435\\
0.29698484809835	0.0165376139502553\\
0.311126983722081	0.0165409655801431\\
0.325269119345812	0.0165444488686435\\
0.339411254969543	0.016548060735487\\
0.353553390593274	0.0165517980107116\\
};
\addlegendentry{Number of interferers $n = 8$}

\addplot [color=black, line width=1.0pt, draw=none, mark=o, mark size=3pt]
  table[row sep=crcr]{%
0	0.0165019151714481\\
0.014142135623731	0.016501999050774\\
0.0282842712474619	0.0165022506082475\\
0.0424264068711929	0.0165026696025087\\
0.0565685424949238	0.0165032556318062\\
0.0707106781186548	0.0165040081347666\\
0.0848528137423857	0.0165049263914702\\
0.0989949493661167	0.0165060095248249\\
0.113137084989848	0.0165072565022414\\
0.127279220613579	0.016508666137597\\
0.14142135623731	0.0165102370934852\\
0.15556349186104	0.0165119678837469\\
0.169705627484771	0.0165138568762695\\
0.183847763108502	0.0165159022960485\\
0.197989898732233	0.0165181022285054\\
0.212132034355964	0.0165204546230443\\
0.226274169979695	0.0165229572968436\\
0.240416305603426	0.0165256079388679\\
0.254558441227157	0.0165284041140884\\
0.268700576850888	0.0165313432679006\\
0.282842712474619	0.0165344227307257\\
0.29698484809835	0.0165376397227798\\
0.311126983722081	0.0165409913590056\\
0.325269119345812	0.0165444746541398\\
0.339411254969543	0.0165480865279126\\
0.353553390593274	0.0165518238103626\\
};
\addlegendentry{Number of interferers $n = 15$}

\addplot [color=black, draw=none, mark=asterisk, mark size=3pt]
  table[row sep=crcr]{%
0	0.0165019245241961\\
0.014142135623731	0.0165020084035531\\
0.0282842712474619	0.0165022599611198\\
0.0424264068711929	0.0165026789555368\\
0.0565685424949238	0.0165032649850527\\
0.0707106781186548	0.0165040174882939\\
0.0848528137423857	0.0165049357453405\\
0.0989949493661167	0.0165060188791009\\
0.113137084989848	0.0165072658569852\\
0.127279220613579	0.0165086754928713\\
0.14142135623731	0.016510246449352\\
0.15556349186104	0.0165119772402685\\
0.169705627484771	0.0165138662335088\\
0.183847763108502	0.016515911654068\\
0.197989898732233	0.016518111587367\\
0.212132034355964	0.0165204639828107\\
0.226274169979695	0.0165229666575775\\
0.240416305603426	0.0165256173006313\\
0.254558441227157	0.0165284134769442\\
0.268700576850888	0.0165313526319111\\
0.282842712474619	0.0165344320959535\\
0.29698484809835	0.0165376490892873\\
0.311126983722081	0.0165410007268554\\
0.325269119345812	0.0165444840233943\\
0.339411254969543	0.0165480958986343\\
0.353553390593274	0.016551833182614\\
};
\addlegendentry{Number of interferers $n = 24$}

\addplot [color=black, draw=none, mark=+, mark size=3pt]
  table[row sep=crcr]{%
0	0.016501924677864\\
0.014142135623731	0.0165020085572221\\
0.0282842712474619	0.016502260114789\\
0.0424264068711929	0.016502679109207\\
0.0565685424949238	0.016503265138724\\
0.0707106781186548	0.0165040176419663\\
0.0848528137423857	0.0165049358990144\\
0.0989949493661167	0.0165060190327765\\
0.113137084989848	0.0165072660106623\\
0.127279220613579	0.0165086756465509\\
0.14142135623731	0.0165102466030343\\
0.15556349186104	0.0165119773939532\\
0.169705627484771	0.0165138663871966\\
0.183847763108502	0.0165159118077595\\
0.197989898732233	0.0165181117410621\\
0.212132034355964	0.0165204641365094\\
0.226274169979695	0.0165229668112802\\
0.240416305603426	0.0165256174543381\\
0.254558441227157	0.0165284136306558\\
0.268700576850888	0.0165313527856275\\
0.282842712474619	0.016534432249675\\
0.29698484809835	0.0165376492430143\\
0.311126983722081	0.0165410008805877\\
0.325269119345812	0.0165444841771328\\
0.339411254969543	0.0165480960523788\\
0.353553390593274	0.0165518333363651\\
};
\addlegendentry{Number of interferers $n = 35$}

\end{axis}
\end{tikzpicture}%
\caption{\textit{(Two Dimension Model)} Here the variation of interference $\mathbb{I}_{n}(d_{x},d_{y})$ is drawn with respect to a linear variation of the position $z=\sqrt{d_{x}^{2}+d_{y}^{2}}$ of the receiver photodiode (PD), radially inside the square attocell for different number interferers $n$ in the network. The graph for the proposed interference expression $\hat{\mathtt{I}}_{1,1}(d_{x},d_{y})$ from approximation is also drawn. We consider $a=0.5$m, the half-power-semi-angle (HPSA) $\theta_{h}$ of the LED as $\frac{\pi}{3}$ radians and the height $h$ of the LED as $2.5$m.}
\label{two5z1}
\end{figure}
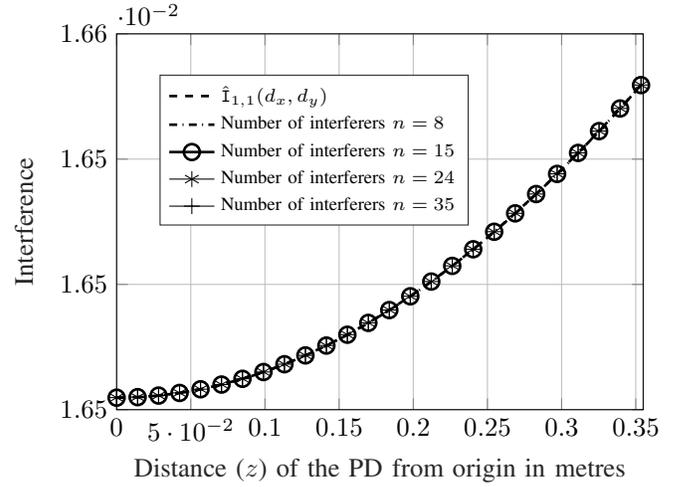 

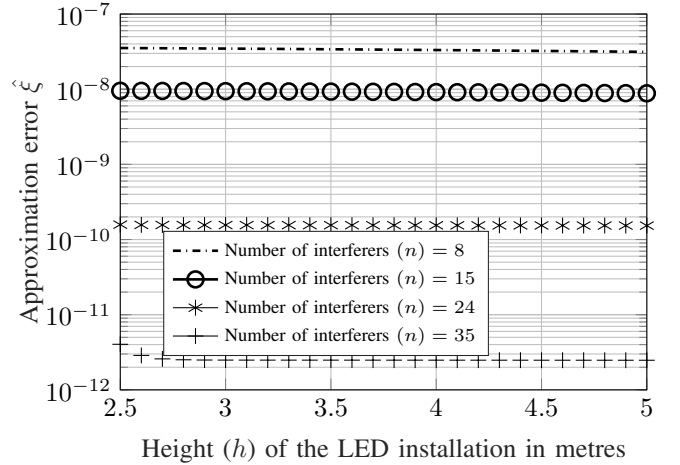
\begin{figure}[ht]
\centering
\begin{tikzpicture}
\begin{axis}[%
width=7cm,
height=5cm,
scale only axis,
xmin=2.5,
xmax=5,
xlabel style={font=\color{white!15!black}},
xlabel={Height ($h$) of the LED installation in metres},
ymode=log,
ymin=1e-12,
ymax=1e-07,
yminorticks=true,
ylabel style={font=\color{white!15!black}},
ylabel={Approximation error $\hat{\xi}$},
axis background/.style={fill=white},
title style={font=\bfseries},
title={},
xmajorgrids,
ymajorgrids,
yminorgrids,
legend style={font=\fontsize{7}{5}\selectfont,at={(0.0815,0.0900)}, anchor=south west, legend cell align=left, align=left, draw=white!15!black}
]
\addplot [color=black, dashdotted, line width=1.0pt]
  table[row sep=crcr]{%
2.5	3.53296927130264e-08\\
2.6	3.52127924613949e-08\\
2.7	3.50927774443732e-08\\
2.8	3.49690442297268e-08\\
2.9	3.48414908478076e-08\\
3	3.47101399350222e-08\\
3.1	3.45750454334645e-08\\
3.2	3.44362705495191e-08\\
3.3	3.42938813817882e-08\\
3.4	3.4147946058067e-08\\
3.5	3.39985339629575e-08\\
3.6	3.38457156660952e-08\\
3.7	3.36895628685885e-08\\
3.8	3.35301483253905e-08\\
3.9	3.33675457329756e-08\\
4	3.32018297041856e-08\\
4.1	3.30330756582922e-08\\
4.2	3.28613597747008e-08\\
4.3	3.26867589808084e-08\\
4.4	3.25093507597737e-08\\
4.5	3.23292132101485e-08\\
4.6	3.2146424877122e-08\\
4.7	3.19610647706264e-08\\
4.8	3.17732122301915e-08\\
4.9	3.15829468981103e-08\\
5	3.13903486552549e-08\\
};
\addlegendentry{Number of interferers $(n) = 8$}

\addplot [color=black, line width=1.0pt, draw=none, mark=o, mark size=3pt]
  table[row sep=crcr]{%
2.5	9.530041745337e-09\\
2.6	9.50868957877438e-09\\
2.7	9.48748647873476e-09\\
2.8	9.46577643004543e-09\\
2.9	9.44340348019851e-09\\
3	9.42033506236284e-09\\
3.1	9.39656823965157e-09\\
3.2	9.37210856361351e-09\\
3.3	9.34696300003077e-09\\
3.4	9.32113950434521e-09\\
3.5	9.29464633384058e-09\\
3.6	9.26749173344073e-09\\
3.7	9.23968437307673e-09\\
3.8	9.21123296583094e-09\\
3.9	9.18214640042644e-09\\
4	9.15243370696629e-09\\
4.1	9.12210410930047e-09\\
4.2	9.09116692278562e-09\\
4.3	9.05963164698435e-09\\
4.4	9.02750784542716e-09\\
4.5	8.99480530997757e-09\\
4.6	8.9615338465395e-09\\
4.7	8.92770343714549e-09\\
4.8	8.89332414568539e-09\\
4.9	8.85840614745078e-09\\
5	8.82295970371051e-09\\
};
\addlegendentry{Number of interferers $(n) = 15$}

\addplot [color=black, draw=none, mark=asterisk, mark size=3pt]
  table[row sep=crcr]{%
2.5	1.57790371047017e-10\\
2.6	1.56539014389834e-10\\
2.7	1.5616379370198e-10\\
2.8	1.56002216353257e-10\\
2.9	1.55890558274641e-10\\
3	1.55789083022828e-10\\
3.1	1.55687692772466e-10\\
3.2	1.55583991436753e-10\\
3.3	1.55477164129336e-10\\
3.4	1.55367245978366e-10\\
3.5	1.55254278183525e-10\\
3.6	1.55138063203178e-10\\
3.7	1.55018750006702e-10\\
3.8	1.54896313874289e-10\\
3.9	1.54770707751564e-10\\
4	1.54641958960422e-10\\
4.1	1.54510146647621e-10\\
4.2	1.54375237744996e-10\\
4.3	1.54237257839716e-10\\
4.4	1.54096129086068e-10\\
4.5	1.5395201140387e-10\\
4.6	1.53804794096081e-10\\
4.7	1.53654573819325e-10\\
4.8	1.5350130959076e-10\\
4.9	1.53345042718488e-10\\
5	1.53185769082541e-10\\
};
\addlegendentry{Number of interferers $(n) = 24$}

\addplot [color=black, draw=none, mark=+, mark size=3pt]
  table[row sep=crcr]{%
2.5	4.03923769431991e-12\\
2.6	2.87181736724484e-12\\
2.7	2.58352714221921e-12\\
2.8	2.51246766447277e-12\\
2.9	2.49393561357891e-12\\
3	2.48942793462659e-12\\
3.1	2.48785107098692e-12\\
3.2	2.48721399379037e-12\\
3.3	2.48665064234155e-12\\
3.4	2.48615971559785e-12\\
3.5	2.48563626278897e-12\\
3.6	2.48521537550561e-12\\
3.7	2.48479253665834e-12\\
3.8	2.48435321793805e-12\\
3.9	2.48375256993449e-12\\
4	2.48319572369871e-12\\
4.1	2.4826601278255e-12\\
4.2	2.48214144550618e-12\\
4.3	2.48160086230298e-12\\
4.4	2.48101441734788e-12\\
4.5	2.48042916501517e-12\\
4.6	2.47978536576515e-12\\
4.7	2.47920054711331e-12\\
4.8	2.4785611930922e-12\\
4.9	2.47793631317009e-12\\
5	2.47727630509759e-12\\
};
\addlegendentry{Number of interferers $(n) = 35$}

\end{axis}
\end{tikzpicture}%
\caption{\textit{(Two Dimension Model)} Here the variation of interference approximation error $\hat{\xi}=|\mathbb{I}_{n}(d_{x},d_{y})-\hat{\mathtt{I}}_{1,1}(d_{x},d_{y})|$ is drawn for a linear variation of the height $h$ of installation of the LED for different number interferers $n$ in the network. We consider $a=0.5$m, the half-power-semi-angle (HPSA) $\theta_{h}$ of the LED as $\frac{\pi}{3}$ radians and $d_{x}=d_{y}=0$.}
\label{two5h3}
\end{figure} 

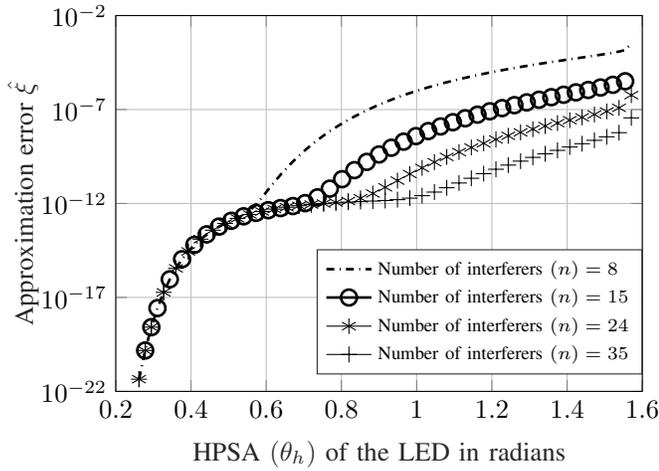
\begin{figure}[ht]
\centering
\begin{tikzpicture}
\begin{axis}[%
width=7cm,
height=5cm,
scale only axis,
xmin=0.2,
xmax=1.6,
xlabel style={font=\color{white!15!black}},
xlabel={HPSA $(\theta_{h})$ of the LED in radians},
ymode=log,
ymin=1e-22,
ymax=0.01,
yminorticks=true,
ylabel style={font=\color{white!15!black}},
ylabel={Approximation error $\hat{\xi}$},
axis background/.style={fill=white},
title style={font=\bfseries},
title={},
xmajorgrids,
ymajorgrids,
yminorgrids,
legend style={font=\fontsize{7}{5}\selectfont,at={(0.3815,0.0450)}, anchor=south west, legend cell align=left, align=left, draw=white!15!black}
]
\addplot [color=black, dashdotted, line width=1.0pt]
  table[row sep=crcr]{%
0.261799387799149	4.44722473647548e-22\\
0.278161849536596	1.50081331221462e-20\\
0.294524311274043	2.61087216809121e-19\\
0.31088677301149	2.70497466472752e-18\\
0.327249234748937	1.8606361435362e-17\\
0.343611696486384	9.23132855691209e-17\\
0.35997415822383	3.52136658841444e-16\\
0.376336619961277	1.0854338569287e-15\\
0.392699081698724	2.81135606160138e-15\\
0.409061543436171	6.31104851282943e-15\\
0.425424005173618	1.25865048997386e-14\\
0.441786466911065	2.27496468168469e-14\\
0.458148928648511	3.7874030244962e-14\\
0.474511390385958	5.88677547599123e-14\\
0.490873852123405	8.64962180469709e-14\\
0.507236313860852	1.21996863667792e-13\\
0.523598775598299	1.69874935990254e-13\\
0.539961237335746	2.47888799873204e-13\\
0.556323699073193	4.17817408953193e-13\\
0.572686160810639	8.69161804974063e-13\\
0.589048622548086	2.12290290733811e-12\\
0.605411084285533	5.48046617853196e-12\\
0.62177354602298	1.39282292838973e-11\\
0.638136007760427	3.38185989628981e-11\\
0.654498469497873	7.77751050411515e-11\\
0.67086093123532	1.6939257093361e-10\\
0.687223392972767	3.50399700791218e-10\\
0.703585854710214	6.90994497189699e-10\\
0.719948316447661	1.30403064832908e-09\\
0.736310778185108	2.36361452133661e-09\\
0.752673239922555	4.128467778362e-09\\
0.769035701660001	6.97013609893618e-09\\
0.785398163397448	1.14058038443686e-08\\
0.801760625134895	1.81351438510377e-08\\
0.818123086872342	2.80803211429731e-08\\
0.834485548609789	4.24280150639403e-08\\
0.850848010347236	6.26721463831484e-08\\
0.867210472084682	9.06559116014075e-08\\
0.883572933822129	1.28611738424086e-07\\
0.899935395559576	1.79197882640286e-07\\
0.916297857297023	2.45530570222585e-07\\
0.93266031903447	3.31210838212456e-07\\
0.949022780771917	4.40345518144467e-07\\
0.965385242509364	5.77562122839465e-07\\
0.98174770424681	7.48017715664073e-07\\
0.998110165984257	9.57402152373249e-07\\
1.0144726277217	1.21193637587363e-06\\
1.03083508945915	1.51836670392314e-06\\
1.0471975511966	1.88395628179794e-06\\
1.06356001293404	2.31647507406563e-06\\
1.07992247467149	2.82418994993969e-06\\
1.09628493640894	3.41585658286334e-06\\
1.11264739814639	4.10071505323184e-06\\
1.12900985988383	4.88849122915955e-06\\
1.14537232162128	5.78940622131957e-06\\
1.16173478335873	6.81419650640919e-06\\
1.17809724509617	7.97414770613986e-06\\
1.19445970683362	9.28114556303372e-06\\
1.21082216857107	1.07477484135796e-05\\
1.22718463030851	1.2387286558696e-05\\
1.24354709204596	1.42139954379533e-05\\
1.25990955378341	1.62431917052272e-05\\
1.27627201552085	1.84915044231099e-05\\
1.2926344772583	2.09771781346096e-05\\
1.30899693899575	2.37204712635955e-05\\
1.32535940073319	2.67441833638957e-05\\
1.34172186247064	3.00743601509348e-05\\
1.35808432420809	3.37412494792882e-05\\
1.37444678594553	3.77806205837244e-05\\
1.39080924768298	4.22356242085475e-05\\
1.40717170942043	4.71594843785283e-05\\
1.42353417115788	5.26195172173421e-05\\
1.43989663289532	5.870336195965e-05\\
1.45625909463277	6.55290987900681e-05\\
1.47262155637022	7.32626529155594e-05\\
1.48898401810766	8.21500277945059e-05\\
1.50534647984511	9.25831779079872e-05\\
1.52170894158256	0.000105254660269088\\
1.53807140332	0.000121607350669212\\
1.55443386505745	0.000145774417870517\\
1.5707963267949	0.000324215830198787\\
};
\addlegendentry{Number of interferers $(n) = 8$}

\addplot [color=black, line width=1.0pt, draw=none, mark=o, mark size=3pt]
  table[row sep=crcr]{%
0.278161849536596	1.50081331221462e-20\\
0.294524311274043	2.61087216809121e-19\\
0.31088677301149	2.70497466472752e-18\\
0.343611696486384	9.23132855691209e-17\\
0.376336619961277	1.08543385641171e-15\\
0.409061543436171	6.31104724062564e-15\\
0.441786466911065	2.27492386528455e-14\\
0.474511390385958	5.88262879443421e-14\\
0.507236313860852	1.20198739352197e-13\\
0.539961237335746	2.07212773415538e-13\\
0.572686160810639	3.15377674604995e-13\\
0.605411084285533	4.37773474697319e-13\\
0.638136007760427	5.69707828005153e-13\\
0.67086093123532	7.27807245894108e-13\\
0.703585854710214	1.04248955388325e-12\\
0.736310778185108	2.11763094039308e-12\\
0.769035701660001	6.13778060494508e-12\\
0.801760625134895	1.96463223293974e-11\\
0.834485548609789	5.93997303600935e-11\\
0.867210472084682	1.62985713467179e-10\\
0.899935395559576	4.05748348114576e-10\\
0.93266031903447	9.24908795385537e-10\\
0.965385242509364	1.95076416999174e-09\\
0.998110165984257	3.84379448062355e-09\\
1.03083508945915	7.13579494798777e-09\\
1.06356001293404	1.25730226264997e-08\\
1.09628493640894	2.11599611305013e-08\\
1.12900985988383	3.42036427514181e-08\\
1.16173478335873	5.33605582844765e-08\\
1.19445970683362	8.06910504386193e-08\\
1.22718463030851	1.18730226296437e-07\\
1.25990955378341	1.70590765756484e-07\\
1.2926344772583	2.4012388628869e-07\\
1.32535940073319	3.32185452714406e-07\\
1.35808432420809	4.53097223560017e-07\\
1.39080924768298	6.11491477564563e-07\\
1.42353417115788	8.19977610005962e-07\\
1.45625909463277	1.09880640104809e-06\\
1.48898401810766	1.48537302366303e-06\\
1.52170894158256	2.06680091333067e-06\\
1.55443386505745	3.18942658460164e-06\\
};
\addlegendentry{Number of interferers $(n) = 15$}

\addplot [color=black, draw=none, mark=asterisk, mark size=3pt]
  table[row sep=crcr]{%
0.261799387799149	4.44722473647548e-22\\
0.278161849536596	1.50081331221462e-20\\
0.294524311274043	2.61087216809121e-19\\
0.327249234748937	1.8606361435362e-17\\
0.35997415822383	3.52136658841444e-16\\
0.392699081698724	2.81135602355107e-15\\
0.425424005173618	1.25864778674762e-14\\
0.458148928648511	3.78693797164443e-14\\
0.490873852123405	8.61968054526437e-14\\
0.523598775598299	1.60690163872855e-13\\
0.556323699073193	2.59061238986388e-13\\
0.589048622548086	3.75166667480539e-13\\
0.62177354602298	5.01399919847212e-13\\
0.654498469497873	6.30566616716816e-13\\
0.687223392972767	7.57017170303909e-13\\
0.719948316447661	8.77474416911828e-13\\
0.752673239922555	9.94297092576923e-13\\
0.785398163397448	1.12727535639401e-12\\
0.818123086872342	1.35224036829085e-12\\
0.850848010347236	1.90239230618561e-12\\
0.883572933822129	3.39474975907805e-12\\
0.916297857297023	7.2633947467704e-12\\
0.949022780771917	1.65104440535968e-11\\
0.98174770424681	3.6860106980563e-11\\
1.0144726277217	7.84072563819738e-11\\
1.0471975511966	1.57790371047017e-10\\
1.07992247467149	3.0090435698682e-10\\
1.11264739814639	5.46161092046527e-10\\
1.14537232162128	9.48306835835133e-10\\
1.17809724509617	1.58293149421507e-09\\
1.21082216857107	2.55193469134163e-09\\
1.24354709204596	3.99054578964497e-09\\
1.27627201552085	6.07697547838626e-09\\
1.30899693899575	9.04680359792298e-09\\
1.34172186247064	1.32160341212439e-08\\
1.37444678594553	1.90209257927565e-08\\
1.40717170942043	2.70925852863257e-08\\
1.43989663289532	3.84112256307523e-08\\
1.47262155637022	5.4672508048581e-08\\
1.50534647984511	7.93647148139565e-08\\
1.53807140332	1.22479299008438e-07\\
1.5707963267949	5.81238457253441e-07\\
};
\addlegendentry{Number of interferers $(n) = 24$}

\addplot [color=black, draw=none, mark=+, mark size=3pt]
  table[row sep=crcr]{%
0.261799387799149	4.44722473647548e-22\\
0.278161849536596	1.50081331221462e-20\\
0.294524311274043	2.61087216809121e-19\\
0.327249234748937	1.8606361435362e-17\\
0.35997415822383	3.52136658841444e-16\\
0.392699081698724	2.81135602355107e-15\\
0.425424005173618	1.25864778674762e-14\\
0.458148928648511	3.78693797164443e-14\\
0.490873852123405	8.61968054526437e-14\\
0.523598775598299	1.60690163872855e-13\\
0.556323699073193	2.59061238986388e-13\\
0.589048622548086	3.75166660704276e-13\\
0.62177354602298	5.01399567481506e-13\\
0.654498469497873	6.30560436764432e-13\\
0.687223392972767	7.56941601412486e-13\\
0.719948316447661	8.76438895416887e-13\\
0.752673239922555	9.8689741274971e-13\\
0.785398163397448	1.08746778942903e-12\\
0.818123086872342	1.17733687382549e-12\\
0.850848010347236	1.26004501971311e-12\\
0.883572933822129	1.33924902417926e-12\\
0.916297857297023	1.43609690111957e-12\\
0.949022780771917	1.59655795528568e-12\\
0.98174770424681	1.91814099326226e-12\\
1.0144726277217	2.60485730207982e-12\\
1.0471975511966	4.03923769431991e-12\\
1.07992247467149	6.89136941955937e-12\\
1.11264739814639	1.2279659927783e-11\\
1.14537232162128	2.19668172540821e-11\\
1.17809724509617	3.86361914683775e-11\\
1.21082216857107	6.62316243515981e-11\\
1.24354709204596	1.10422317123327e-10\\
1.27627201552085	1.79230817687337e-10\\
1.30899693899575	2.83964844693951e-10\\
1.34172186247064	4.40663873446656e-10\\
1.37444678594553	6.72519206990074e-10\\
1.40717170942043	1.01443711764126e-09\\
1.43989663289532	1.52229782945934e-09\\
1.47262155637022	2.29513030802764e-09\\
1.50534647984511	3.54066423402521e-09\\
1.53807140332	5.8652218498656e-09\\
1.5707963267949	3.58756823415352e-08\\
};
\addlegendentry{Number of interferers $(n) = 35$}

\end{axis}
\end{tikzpicture}%
\caption{\textit{(Two Dimension Model)} Here the variation of interference approximation error $\hat{\xi}=|\mathbb{I}_{n}(d_{x},d_{y})-\hat{\mathtt{I}}_{1,1}(d_{x},d_{y})|$ is drawn for a linear variation of the half-power-semi-angle (HPSA) $\theta_{h}$ of the LED for different number interferers $n$ in the network. We consider the attocell length $a=0.5$m, the height $h$ of the LED as $2.5$m and $d_{x}=d_{y}=0$.}
\label{two5t3}
\end{figure} 

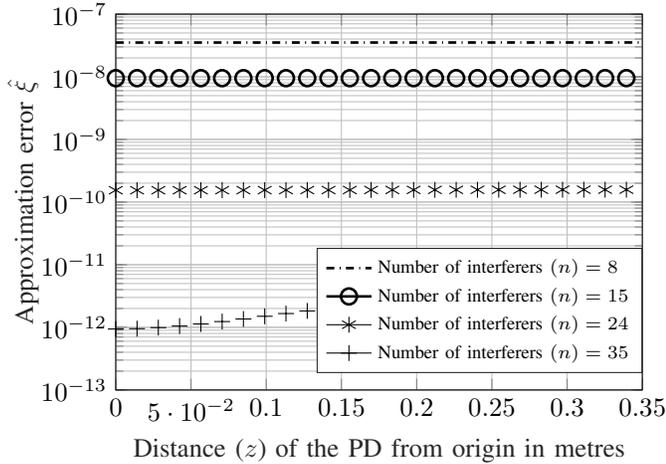
\begin{figure}[ht]
\centering
\begin{tikzpicture}
\begin{axis}[%
width=7cm,
height=5cm,
scale only axis,
xmin=0,
xmax=0.35,
xlabel style={font=\color{white!15!black}},
xlabel={Distance ($z$) of the PD from origin in metres},
ymode=log,
ymin=1e-13,
ymax=1e-07,
yminorticks=true,
ylabel style={font=\color{white!15!black}},
ylabel={Approximation error $\hat{\xi}$},
axis background/.style={fill=white},
title style={font=\bfseries},
title={},
xmajorgrids,
ymajorgrids,
yminorgrids,
legend style={font=\fontsize{7}{5}\selectfont,at={(0.3815,0.0450)}, anchor=south west, legend cell align=left, align=left, draw=white!15!black}
]
\addplot [color=black, dashdotted, line width=1.0pt]
  table[row sep=crcr]{%
0	3.52149275260893e-08\\
0.014142135623731	3.52151183109772e-08\\
0.0282842712474619	3.521569023196e-08\\
0.0424264068711929	3.52166428727041e-08\\
0.0565685424949238	3.52179753033977e-08\\
0.0707106781186548	3.5219686393001e-08\\
0.0848528137423857	3.52217740424987e-08\\
0.0989949493661167	3.52242368710509e-08\\
0.113137084989848	3.52270724396364e-08\\
0.127279220613579	3.52302779588198e-08\\
0.14142135623731	3.52338514406081e-08\\
0.15556349186104	3.52377898318879e-08\\
0.169705627484771	3.52420899962791e-08\\
0.183847763108502	3.52467493386355e-08\\
0.197989898732233	3.52517646046158e-08\\
0.212132034355964	3.52571328486595e-08\\
0.226274169979695	3.52628513368425e-08\\
0.240416305603426	3.52689170785014e-08\\
0.254558441227157	3.52753277838014e-08\\
0.268700576850888	3.52820808818821e-08\\
0.282842712474619	3.52891740516836e-08\\
0.29698484809835	3.52966060962467e-08\\
0.311126983722081	3.53043748506365e-08\\
0.325269119345812	3.53124792462634e-08\\
0.339411254969543	3.53209187661796e-08\\
0.353553390593274	3.53296927130264e-08\\
};
\addlegendentry{Number of interferers $(n) = 8$}

\addplot [color=black, line width=1.0pt, draw=none, mark=o, mark size=3pt]
  table[row sep=crcr]{%
0	9.50735697113902e-09\\
0.014142135623731	9.50740058555666e-09\\
0.0282842712474619	9.50753085288136e-09\\
0.0424264068711929	9.50774754066019e-09\\
0.0565685424949238	9.50804960112017e-09\\
0.0707106781186548	9.50843597607998e-09\\
0.0848528137423857	9.50890437223517e-09\\
0.0989949493661167	9.50945345731813e-09\\
0.113137084989848	9.51008073679649e-09\\
0.127279220613579	9.51078345245993e-09\\
0.14142135623731	9.51155969958206e-09\\
0.15556349186104	9.51240609892157e-09\\
0.169705627484771	9.51331944470946e-09\\
0.183847763108502	9.51429738119125e-09\\
0.197989898732233	9.51533629320322e-09\\
0.212132034355964	9.5164333080433e-09\\
0.226274169979695	9.51758557382609e-09\\
0.240416305603426	9.51879004090772e-09\\
0.254558441227157	9.52004429108366e-09\\
0.268700576850888	9.52134583329101e-09\\
0.282842712474619	9.52269182952215e-09\\
0.29698484809835	9.52408160670437e-09\\
0.311126983722081	9.5255123060134e-09\\
0.325269119345812	9.52698298375965e-09\\
0.339411254969543	9.52849313340387e-09\\
0.353553390593274	9.530041745337e-09\\
};
\addlegendentry{Number of interferers $(n) = 15$}

\addplot [color=black, draw=none, mark=asterisk, mark size=3pt]
  table[row sep=crcr]{%
0	1.54608943703227e-10\\
0.014142135623731	1.54621485753959e-10\\
0.0282842712474619	1.54658525569618e-10\\
0.0424264068711929	1.54719424771965e-10\\
0.0565685424949238	1.54803080076871e-10\\
0.0707106781186548	1.54908700450429e-10\\
0.0848528137423857	1.55034075854932e-10\\
0.0989949493661167	1.55177443811327e-10\\
0.113137084989848	1.55336948165496e-10\\
0.127279220613579	1.55509188859426e-10\\
0.14142135623731	1.55692816278252e-10\\
0.15556349186104	1.55884475466772e-10\\
0.169705627484771	1.5608008635537e-10\\
0.183847763108502	1.56277893403889e-10\\
0.197989898732233	1.56474683904451e-10\\
0.212132034355964	1.56666884326695e-10\\
0.226274169979695	1.56851656663015e-10\\
0.240416305603426	1.57026638220037e-10\\
0.254558441227157	1.57188501798133e-10\\
0.268700576850888	1.5733528369033e-10\\
0.282842712474619	1.57464062622292e-10\\
0.29698484809835	1.5757408919348e-10\\
0.311126983722081	1.57662546212967e-10\\
0.325269119345812	1.57728510807864e-10\\
0.339411254969543	1.57771400111084e-10\\
0.353553390593274	1.57790371047017e-10\\
};
\addlegendentry{Number of interferers $(n) = 24$}

\addplot [color=black, draw=none, mark=+, mark size=3pt]
  table[row sep=crcr]{%
0	9.41021566225331e-13\\
0.014142135623731	9.52463802272874e-13\\
0.0282842712474619	9.89409942864228e-13\\
0.0424264068711929	1.04918504439944e-12\\
0.0565685424949238	1.1317821679846e-12\\
0.0707106781186548	1.23627497128354e-12\\
0.0848528137423857	1.36016198304389e-12\\
0.0989949493661167	1.50184378822082e-12\\
0.113137084989848	1.65981464683718e-12\\
0.127279220613579	1.82956774730236e-12\\
0.14142135623731	2.01058614202054e-12\\
0.15556349186104	2.19972651205325e-12\\
0.169705627484771	2.39229816512143e-12\\
0.183847763108502	2.58637208871981e-12\\
0.197989898732233	2.77961334504973e-12\\
0.212132034355964	2.96822982859268e-12\\
0.226274169979695	3.14893250363824e-12\\
0.240416305603426	3.31983399104452e-12\\
0.254558441227157	3.47694789626374e-12\\
0.268700576850888	3.61892807387854e-12\\
0.282842712474619	3.74252018264798e-12\\
0.29698484809835	3.84707543599205e-12\\
0.311126983722081	3.93020685440781e-12\\
0.325269119345812	3.98992297534484e-12\\
0.339411254969543	4.02698013823866e-12\\
0.353553390593274	4.03923769431991e-12\\
};
\addlegendentry{Number of interferers $(n) = 35$}

\end{axis}
\end{tikzpicture}%
\caption{\textit{(Two Dimension Model)} Here the variation of interference approximation error $\hat{\xi}=|\mathbb{I}_{n}(d_{x},d_{y})-\hat{\mathtt{I}}_{1,1}(d_{x},d_{y})|$ is drawn for a linear variation of the position $z=\sqrt{d_{x}^{2}+d_{y}^{2}}$ of the receiver photodiode (PD), radially inside the attocell for different number interferers $n$ in the network. We consider $a=0.5$m, the half-power-semi-angle (HPSA) $\theta_{h}$ of the LED as $\frac{\pi}{3}$ radians and the height $h$ of the LED as $2.5$m.}
\label{two5z3}
\end{figure}    

From Fig. \ref{two5h1} and it's corresponding approximation error plot in Fig. \ref{two5h3}, we observe that for any given height $h$, as the number of interferers increase, the error $\hat{\xi}$, decreases. We observe a maximum error $\hat{\xi}_{max}$ in the order of $10^{-8}$ with respect to $\hat{\mathtt{I}}_{1,1}(d_{x},d_{y})$, that is in the order of $10^{-3}$. This error further reduces as the number of interferers is increased. The same can be observed with the variation of HPSA in graphs of Fig. \ref{two5t1} and the error plot in Fig. \ref{two5t3}, where $\hat{\xi}_{max}$ is in the order of $10^{-5}$, for $\hat{\mathtt{I}}_{1,1}(d_{x},d_{y})$ in the order of $10^{-1}$. Again, this error reduces as the number of interferers increases. Similarly, in graphs of Fig. \ref{two5z1} and Fig. \ref{two5z3}, we observe $\hat{\xi}_{max}$ in the order of $10^{-7}$, for $\hat{\mathtt{I}}_{1,1}(d_{x},d_{y})$ in the order of $10^{-2}$. So, when compared with the interference values, these errors are small, which numerically validates the approximation to $\hat{\mathtt{I}}_{1,1}(d_{x},d_{y})$. \\ \par

As seen in the above example, Prop. \ref{prop2} essentially implies that for a given value of $\frac{h}{a}$ the approximation to $\hat {\mathtt{I}}_{j,l}(d_{x},d_{y})$ is tight and very close to the actual interference $\mathbb{I}_{\infty}(d_{x},d_{y})$ in \eqref{eqn:interf2a}, with an approximation error bounded by an exponential decay.
Hence, the above discussion can be summarized as   
\begin{align*}
\mathbb{I}_{n}(d_{x},d_{y}) < \mathbb{I}_{\infty}(d_{x},d_{y}) \approx \hat{\mathtt{I}}_{j,l}(d_{x},d_{y}). 
\end{align*}
Similar to the one dimension model, this also implies that our characterization provides closed form analytical bounds for interference in finite LED networks.  
\subsection{Two dimension model with FOV $\theta_{f}<\frac{\pi}{2}$ radians} 
We now look at the interference characterization when $\theta_{f}<\frac{\pi}{2}$ radians. Here we show that, the Fourier analysis method can be used to give a suitable interference approximation for such cases as well.
The infinite summation in \eqref{eqn:interf2a} becomes a finite summation, when the FOV constraint function $\rho(D_{u,v})$, acts on every interferer. From the proof of Thm. \ref{theorem2}, we can modify the function $q(.)$ in \eqref{eqn:funcRad2d} as 
\[q'(r) = (r^{2} + h^{2})^{-\beta}\rho(D_{u,v}).\]
The Poisson summation theorem can be used to obtain a similar result as in the previous subsection if the Fourier transform of $q'(r)$ can be obtained.

Using the Hankel transform \cite{hankel}, the Fourier transform of  $q'(r) $ equals
\begin{align*} 
Q'(w,k) =&2\pi \int_{0}^{\infty}\frac{J_{0}(2\pi r\sqrt{w^{2}+k^{2}})\rho(D_{u,v})}{(r^{2}+h^{2})^{\beta}}r\d r, \nonumber\\
=&2\pi \int_{0}^{h\tan(\theta_{f})}\frac{J_{0}(2\pi r\sqrt{w^{2}+k^{2}})}{(r^{2}+h^{2})^{\beta}}r\d r.
\end{align*}
Hence we have the following Lemma. 
\begin{lemma}
For an FOV $\theta_{f}<\frac{\pi}{2}$ radians and finite integers $j\geq0$ and $l\geq0$ we have 
\begin{align}
\mathbb{I}_{\infty}(d_{x},d_{y}) \approx &\hat{\mathtt{I}}_{j,l}(d_{y},d_{y})\nonumber\\
&=\frac{1}{a^{2}}\bigg[Q'(0,0)+4\sum_{(w,k)\in \mathbb{A}}Q'\bigg(\frac{w}{a},\frac{k}{a}\bigg) \nonumber\\
&\ \ \ \ \ \ \ \ \cos\bigg(\frac{2\pi w d_{x}}{a}\bigg)\cos\bigg(\frac{2\pi k d_{y}}{a}\bigg)\bigg] \nonumber\\
&\ \ \ \ -\frac{1}{(d_{x}^{2}+d_{y}^{2}+h^{2})^{\beta}},
\label{eqn:fovfov2}
\end{align}
where $\mathbb{A} \triangleq (\mathbb{Z}^{2}\cap([0,j]\times[0,l]))\setminus\{(0,0)\}$ over the set of integers $\mathbb{Z}^{2}$. 
\label{lem2}
\end{lemma}
\begin{proof}
Follows from the Poisson summation theorem and approximations. 
\end{proof}

The constant term evaluated at $w=k=0$ is 
\begin{align*}
Q'(0,0)&=2\pi\int_{0}^{\infty}\frac{J_{0}(0)\rho(D_{u,v})}{(r^{2}+h^{2})^{\beta}}r\d r,\nonumber\\
&=2\pi\int_{0}^{h\tan(\theta_{f})}\frac{J_{0}(0)}{(r^{2}+h^{2})^{\beta}}r\d r,\nonumber\\
&=\frac{h^{2-2\beta}\pi}{\beta-1}(1-\cos(\theta_{f})^{2\beta-2}).
\end{align*}
As earlier, this represents the average spatial interference seen at all locations. A closed form expression for $Q'(\frac{w}{a},\frac{k}{a})$ can be simply obtained from numerical integration. \\ \par

Similar to the one dimension model, we consider $h=2.5$m and $a=0.5$m, leading to $\frac{h}{a}=5$ to numerically validate \eqref{eqn:fovfov2} for $j=l=1$ over various values of $\theta_{f}$ and compare it with $\mathbb{I}_{\infty}(d_{x},d_{y})$ in \eqref{eqn:interf2}. In Li-Fi attocell networks, if the FOV $\theta_{f}<\theta_{o}\big(=\tan^{-1}\big(\frac{a}{h}\big)\big)$, the PD does not experience any interference. Here the ratio $\frac{a}{h}=0.2$ and $\theta_{o}=0.197$ radians. So, in Fig. \ref{twofov2}, we observe that both $\mathbb{I}_{\infty}(d_{x},d_{y})$ and $\hat{\mathtt{I}}'_{1,1}(d_{x},d_{y})$ drop down to zero once $\theta_{f}<\theta_{o}=0.197$ radians. Also, for $\theta_{f}>\theta_{o}$, both the graphs, $\mathbb{I}_{\infty}(d_{x},d_{y})$ and $\hat{\mathtt{I}}'_{1,1}(d_{x},d_{y})$ are tightly bounded, which numerically validates our proposition in Lem. \ref{lem2} for $j=l=1$. Also, as $\theta_{f}\to 1.57 (=\frac{\pi}{2})$ radians, the interference values converge to the earlier case of $\theta_{f}=\frac{\pi}{2}$ radians, for the two dimension model, giving similar validation results as in the one dimension model.    \\ \par

So, the approximation above in Lem. \ref{lem2} is a good approximation for various practical parameter values based on the choice of $(j,l)$. As shown above, if we choose $h=2.5$m and $a=0.5$m, considering $j=l=1$ is sufficient. When $\frac{h}{a}$ becomes small, a few more terms are necessary to improve the approximation accuracy. \\ \par

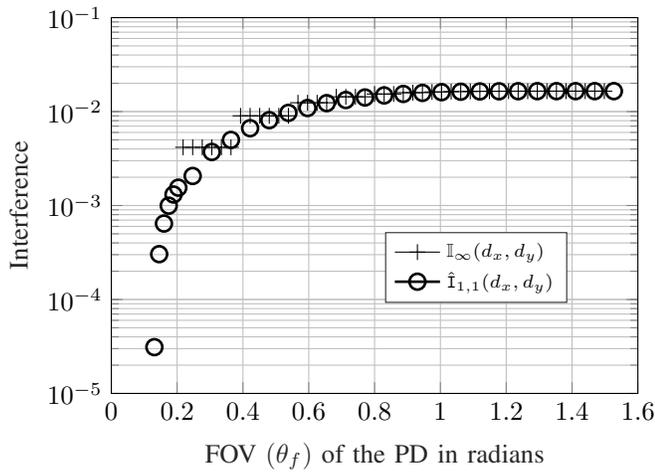
\begin{figure}[ht]
\centering
\definecolor{mycolor1}{rgb}{0.00000,0.44700,0.74100}%
\definecolor{mycolor2}{rgb}{0.85000,0.32500,0.09800}%
\begin{tikzpicture}
\begin{axis}[%
width=7cm,
height=5cm,
scale only axis,
xmin=0,
xmax=1.6,
xlabel style={font=\color{white!15!black}},
xlabel={FOV $(\theta_{f})$ of the PD in radians},
ymode=log,
ymin=1e-05,
ymax=0.1,
yminorticks=true,
ylabel style={font=\color{white!15!black}},
ylabel={Interference},
axis background/.style={fill=white},
xmajorgrids,
ymajorgrids,
yminorgrids,
legend style={font=\fontsize{7}{5}\selectfont,at={(0.52,0.241)}, anchor=south west, legend cell align=left, align=left, draw=white!15!black}
]
\addplot [color=black, draw=none, mark=+, mark size=3pt]
  table[row sep=crcr]{%
0	0\\
0.0145444104332861	0\\
0.0290888208665722	0\\
0.0436332312998582	0\\
0.0581776417331443	0\\
0.0727220521664304	0\\
0.0872664625997165	0\\
0.101810873033003	0\\
0.116355283466289	0\\
0.130899693899575	0\\
0.145444104332861	0\\
0.159988514766147	0\\
0.174532925199433	0\\
0.189077335632719	0\\
0.218166156499291	0.00416765\\
0.247254977365863	0.00416765\\
0.276343798232435	0.00416765\\
0.305432619099008	0.00416765\\
0.33452143996558	0.00416765\\
0.363610260832152	0.00416765\\
0.392699081698724	0.00900731\\
0.421787902565296	0.00900731\\
0.450876723431868	0.00900731\\
0.479965544298441	0.00900731\\
0.509054365165013	0.00900731\\
0.538143186031585	0.00900731\\
0.567232006898157	0.0124201\\
0.596320827764729	0.0124201\\
0.625409648631301	0.0124201\\
0.654498469497874	0.0124201\\
0.683587290364446	0.0143647\\
0.712676111231018	0.0143647\\
0.74176493209759	0.0143647\\
0.770853752964162	0.0143647\\
0.799942573830734	0.0153815\\
0.829031394697307	0.0153815\\
0.858120215563879	0.0153815\\
0.887209036430451	0.0159013\\
0.916297857297023	0.0159013\\
0.945386678163595	0.0159013\\
0.974475499030167	0.0161697\\
1.00356431989674	0.0161697\\
1.03265314076331	0.0163117\\
1.06174196162988	0.0163117\\
1.09083078249646	0.0163892\\
1.11991960336303	0.0164329\\
1.1490084242296	0.0164584\\
1.17809724509617	0.0164737\\
1.20718606596274	0.0164831\\
1.23627488682932	0.0164891\\
1.26536370769589	0.016493\\
1.29445252856246	0.0164974\\
1.32354134942903	0.0164995\\
1.35263017029561	0.0165008\\
1.38171899116218	0.0165015\\
1.41080781202875	0.0165017\\
1.43989663289532	0.0165019\\
1.46898545376189	0.0165019\\
1.49807427462847	0.0165019\\
1.52716309549504	-0.00065536\\
};
\addlegendentry{$\mathbb{I}_{\infty}(d_{x},d_{y})$}

\addplot [color=black, line width=1.0pt, draw=none, mark=o, mark size=3pt]
  table[row sep=crcr]{%
0	-0.00065536\\
0.0145444104332861	-0.000623821\\
0.0290888208665722	-0.000541791\\
0.0436332312998582	-0.000440641\\
0.0581776417331443	-0.000354377\\
0.0727220521664304	-0.000301473\\
0.0872664625997165	-0.000273576\\
0.101810873033003	-0.000236986\\
0.116355283466289	-0.000146985\\
0.130899693899575	3.10946e-05\\
0.145444104332861	0.000303765\\
0.159988514766147	0.000643184\\
0.174532925199433	0.000996801\\
0.189077335632719	0.00130953\\
0.203621746066005	0.00154888\\
0.247254977365863	0.00206315\\
0.305432619099008	0.00371714\\
0.363610260832152	0.0049932\\
0.421787902565296	0.00665095\\
0.479965544298441	0.00810102\\
0.538143186031585	0.00968992\\
0.596320827764729	0.0109111\\
0.654498469497874	0.0122977\\
0.712676111231018	0.0132738\\
0.770853752964162	0.0141267\\
0.829031394697307	0.0148587\\
0.887209036430451	0.0154089\\
0.945386678163595	0.0158038\\
1.00356431989674	0.016088\\
1.06174196162988	0.0162716\\
1.11991960336303	0.0163847\\
1.17809724509617	0.0164479\\
1.23627488682932	0.0164806\\
1.29445252856246	0.0164948\\
1.35263017029561	0.0165002\\
1.41080781202875	0.0165016\\
1.46898545376189	0.0165019\\
1.52716309549504	0.0165019\\
};
\addlegendentry{$\hat{\mathtt{I}}_{1,1}(d_{x},d_{y})$}

\end{axis}
\end{tikzpicture}%
\caption{\textit{(Two Dimension Model)} ($\theta_{f}<\frac{\pi}{2}$ radians) Here the variation of $\hat{\mathtt{I}}_{1,1}(d_{x},d_{y})$ is drawn for a linear variation of the FOV $\theta_{f}$ of the receiver photodiode (PD). $\mathbb{I}_{\infty}(d_{x},d_{y})$ from \eqref{eqn:interf2} (or $\mathbb{I}_{n}(d_{x},d_{y})$ for $n=24$) is also drawn to validate the same. We consider $a=0.5$m, the half-power-semi-angle (HPSA) $\theta_{h}$ of the LED as $\frac{\pi}{3}$ radians, the height $h$ of the LED as $2.5$m and $d_{x}=d_{y}=0$.}
\label{twofov2}
\end{figure}    

\section{ Conclusion }
In this work, the Poisson summation theorem has been used to provide a simple closed form approximation to co-channel-interference in Li-Fi attocell networks for both one and two dimensions.  We also show that the approximation has an error that is tight with respect to an exponential decay for a given set of system parameters. The advantage of this characterization is, it can be used to compute interference power with a high degree of accuracy for any given finite separation between the LEDs and provide upper bounds for interference in finite attocell networks. Using this characterization, large scale network interference summations can be circumvented and important metrics like probability of coverage, area spectral efficiency, optimal LED spacing etc. can be analytically computed in an easy way. Further, we show that our method of Fourier analysis can be extended to characterize interference when the user PDs have limited FOVs as well.
\bibliographystyle{IEEEtran}
\bibliography{IEEEfull,Reference1.bib}

\begin{appendices}

\section{}
\section*{Proof of Theorem \ref{theorem1}}
\begin{proof}
From \eqref{eqn:interf1}, for $\theta_{f}=\frac{\pi}{2}$, we can write the interference term $\mathbb{I}_{\infty}(z)$ as  
\begin{equation}
\mathbb{I}_{\infty}(z) =  \sum_{i = -\infty \setminus 0}^{+\infty}( (ia + z)^{2} + h^{2} )^{-\beta}.
\label{eqn:inta}  
\end{equation}
We scale and shift the function $q(i)$ in \eqref{eqn:intb1} as $q(z+ia)$ and using the time shifting \cite[Prop.4.3.2]{oppenheim}  and scaling \cite[Prop.4.3.5]{oppenheim} property of Fourier transform, we obtain
\begin{equation}
a\sum_{i= -\infty}^{+ \infty}  q(z + ia) = \sum_{w=-\infty}^{+\infty} Q\Big(\frac{w}{a}\Big)e^{\frac{\iota2\pi wz}{a}}.
\label{eqn:intb}  
\end{equation}
In \eqref{eqn:intb}, from \eqref{eqn:inta}, we consider a real and even function $q(i)$ given as
\begin{equation}
q(i) = ( i^{2} + h^{2} ) ^{-\beta}.
\label{eqn:intc}  
\end{equation}
Correspondingly it's Fourier transform $Q(w)=\int_{-\infty}^{\infty}q(x)e^{-\iota 2\pi w x}\d x$ will also be real and even \cite[Prop.4.3.3]{oppenheim} and is given as
 \begin{equation}  
Q(w) = \left\{
                \begin{array}{l}                 
                \frac{h^{1-2\beta}\sqrt{\pi}\Gamma(\beta-0.5)}{\Gamma(\beta)}        ;\ \ \ \ \ \ \ \ \ \ \ \ \ \ \ \ \ \ \ \ \ \ \  w=0, \\ \\
                 \frac{2^{1-\beta}\sqrt{2\pi}h^{0.5-\beta}(2\pi w)^{\beta-0.5} \mathbb{K}_{\beta-0.5}(2\pi hw)}{\Gamma(\beta)}; w\neq0.

                \end{array}
                \right.  
\label{eqn:intd}  
\end{equation} 
 Now, substituting \eqref{eqn:intc} and \eqref{eqn:intd} into \eqref{eqn:intb} we get  
\begin{align}
&\sum_{i = -\infty}^{+ \infty}( (ia + z)^{2} + h^{2} )^{-\beta} \nonumber\\
&=\frac{1}{a}\bigg(Q(0)+\sum_{w=-\infty \setminus 0}^{+\infty}Q\bigg(\frac{w}{a}\bigg)e^{\frac{j2\pi wz}{a}}\bigg), \nonumber\\ 
&= \frac{h^{1-2\beta}\sqrt{\pi}\Gamma(\beta-0.5)}{a\Gamma(\beta)} \  \nonumber\\
&+\sum_{w = -\infty \setminus 0}^{+\infty} \frac{2^{1-\beta}\sqrt{2\pi}h^{0.5-\beta}(2\pi w)^{\beta-0.5} \mathbb{K}_{\beta-0.5}(\frac{2\pi hw}{a})}{a^{0.5+\beta}\Gamma(\beta)}e^{\frac{j2\pi wz}{a}}.
\label{eqn:intf}   
\end{align}
We remove the redundant addition of $i=0$ term from both sides of \eqref{eqn:intf}, which refers to the signal power from the tagged LED source at origin and we get
\begin{align}
&\mathbb{I}_{\infty}(z) \nonumber\\
&=  \sum_{i = -\infty \setminus 0}^{+\infty}( (ia + z)^{2} + h^{2} )^{-\beta}, \nonumber\\
&=\sum_{i = -\infty}^{+ \infty}( (ia + z)^{2} + h^{2} )^{-\beta} -\frac{1}{(z^{2}+h^{2})^{\beta}}, \nonumber\\
&=\frac{1}{a}\bigg(Q(0)+\sum_{w=-\infty \setminus 0}^{+\infty}Q\bigg(\frac{w}{a}\bigg)e^{\frac{j2\pi wz}{a}}\bigg)-\frac{1}{(z^{2}+h^{2})^{\beta}}.
\label{eqn:intf1}    
\end{align}  
Now, that the Fourier transform $Q(w)$ is real and even, we can modify \eqref{eqn:intf1} as
\begin{align*}
&\mathbb{I}_{\infty}(z) \nonumber\\
&\stackrel{(a)}{=}\frac{1}{a}\bigg(Q(0)+\sum_{w=-\infty\setminus 0}^{+\infty}Q\bigg(\frac{w}{a}\bigg)\cos\bigg(\frac{2\pi w z}{a}\bigg)\bigg)-\frac{1}{(z^{2}+h^{2})^{\beta}}, \nonumber\\
&\stackrel{(b)}{=}\frac{1}{a}\bigg(Q(0)+\sum_{w=1}^{+\infty}2Q\bigg(\frac{w}{a}\bigg)\cos\bigg(\frac{2\pi w z}{a}\bigg)\bigg)-\frac{1}{(z^{2}+h^{2})^{\beta}}, \nonumber\\
&= \frac{h^{1-2\beta}\sqrt{\pi}\Gamma(\beta-0.5)}{a\Gamma(\beta)}\nonumber\\
&+\sum_{w =1}^{+\infty} \frac{2^{2-\beta}\sqrt{2\pi}h^{0.5-\beta}(2\pi w)^{\beta-0.5} \mathbb{K}_{\beta-0.5}(\frac{2\pi hw}{a})\cos(\frac{2\pi wz}{a})}{a^{0.5+\beta}\Gamma(\beta)} \nonumber\\
&-\frac{1}{(z^{2}+h^{2})^{\beta}},
\end{align*}
where $(a)$ follows from the fact that $Q(w)$ is real, $(b)$ follows from the fact that $Q(w)$ is even and hence proving the theorem. \par
\end{proof}
\label{app:theorem1}
\section{}
\section*{Proof of Proposition \ref{prop1}}
\begin{proof}
The interference power obtained in Thm. \ref{theorem1} can be written as 
\begin{align*}
\mathbb{I}_{\infty}(z) = &\frac{h^{1-2\beta}\sqrt{\pi}\Gamma(\beta-0.5)}{a\Gamma(\beta)}+\sum_{w=1}^k g(w) + \sum_{w=k+1}^\infty g(w)\nonumber\\
&-\frac{1}{(z^{2}+h^{2})^{\beta}}.
\end{align*}
 Let $E(k)= \sum_{w=k+1}^\infty g(w)$, then 
\begin{align}
g(w) = \frac{M}{a}\cos\bigg(\frac{2\pi z w}{a}\bigg)r^{\beta-0.5}\mathbb{K}_{\beta-0.5}\bigg(\frac{2\pi h w}{a}\bigg),
\label{eqn:prop1}  
\end{align} 
where $M = \frac{2^{2-\beta}(2\pi)^{\beta}h^{0.5-\beta}}{\Gamma(\beta)}$. From \cite{asybesl}, for large $w$, the modified Bessel function $\mathbb{K}_{\beta-0.5}\big(\frac{2\pi h w}{a}\big)$ can be expanded as 
\begin{equation*}
\mathbb{K}_{\beta-0.5}\bigg(\frac{2\pi h w}{a}\bigg) = \sqrt{\frac{a\pi}{4\pi h w}}e^{\frac{-2\pi h r}{a}}\Big(1+\Theta\Big(\frac{1}{w}\Big)\Big),
\end{equation*} 
where $\Theta(.)$ is an asymptotic notation\footnote{The asymptotic notation $f(n)=\Theta(g(n))$ is defined as $\exists k_{1}>0, k_{2}>0, n_{o}>0 \ni \forall n>n_{o}, k_{1}\times g(n)\leq f(n) \leq k_{2}\times g(n).$}. The cosine term $\cos\big(\frac{2\pi w z}{a}\big)$ in \eqref{eqn:prop1} is bounded by $\Theta(1)$, as $w$ becomes large. So, for large $w$,
the asymptotic bound on $g(w)$ in \eqref{eqn:prop1}, can be shown as
\begin{align}
g(w)\leq &\frac{M}{a}\Theta(1)\bigg(\frac{w}{a}\bigg)^{\beta-0.5}\sqrt{\frac{a}{4h w}}e^{\frac{-2\pi h w}{a}}\Big(1+\Theta\Big(\frac{1}{w}\Big)\Big),\nonumber\\
\in&\Theta(w^{\beta-2}e^{\frac{-2\pi h w}{a}}).
\label{eqn:prop41}
\end{align} 
Summing \eqref{eqn:prop41} over large $w$, we write the total error $E(k)$ as 
\begin{align}
E(k) &= \sum_{w=k+1}^{\infty} g(w), \nonumber\\
&\in \sum_{w=k+1}^{\infty} \Theta(w^{\beta-2}e^{\frac{-2\pi h w}{a}}), \nonumber\\
&= \Theta \bigg(\sum_{w=k+1}^{\infty}|w^{\beta-2}e^{\frac{-2\pi h w}{a}}|\bigg).
\label{eqn:newstar}
\end{align}
Observe that summand $h(w) =w^{\beta-2}e^{\frac{-2\pi h w}{a}}$ increases and then decreases with respect to $w$ and attains its maximum at $w_{0}=\frac{a (\beta-2)}{2\pi h}$. Let $w_{1}=\max\{k+1, \lceil w_{0} \rceil\}$. Hence we have
\begin{align}
\sum_{w=k+1}^{\infty}|h(w)|&\leq (w_1-(k+1))|w_1^{\beta-2}e^{\frac{-2\pi h w_1}{a}}|\nonumber\\
&\ \ \ \ \ \ \ \ \ \ \ \ \ +\int_{w_1}^\infty |w^{\beta-2}e^{\frac{-2\pi h w}{a}}| \d w,\nonumber\\
&=(w_1-(k+1)) |w_1^{\beta-2}e^{\frac{-2\pi h w_{1}}{a}}|\nonumber\\
&\ \ \ \ \ \ \ \ \ +\bigg(\frac{a}{2\pi h}\bigg)^{\beta-1}\Gamma\bigg(\beta-1,\frac{2\pi h w_1}{a}\bigg),
\end{align}
where $\Gamma(x,t)=\int_{t}^{\infty}w^{x-1}e^{-w}\d w$ is the incomplete gamma function. Now, for large $t$, from \cite{gammap}, we can asymptotically bound $\Gamma(x,t)$ as 
\[ \Gamma(x,t) \leq \Theta(t^{x-1}e^{-t}). \]
Using this result and the fact that for large $k$ or large $\frac{h}{a}$, $w_1=k+1$, we have
 \begin{align*}
E(k)&\leq \Theta\bigg(\bigg(\frac{a}{2 \pi h}\bigg)^{\beta-1} \bigg( \frac{2\pi h}{a}(k+1)\bigg)^{\beta-2}e^{\frac{-2\pi h(k+1)}{a}}\bigg),\nonumber\\
&\in \Theta \big((k+1)^{\beta-2}e^{\frac{-2\pi h (k+1)}{a}}\big),
\end{align*}
proving the proposition. 
\end{proof}
\label{app:theorem1a}

\section{}
\section*{Proof of Theorem \ref{theorem2}}
\begin{proof}
From \eqref{eqn:interf2}, for $\theta_{f}=\frac{\pi}{2}$, we can write interference term $\mathbb{I}_{\infty}(d_{x},d_{y})$ as  
\begin{align}
\mathbb{I}_{\infty}(d_{x},d_{y}) =  \sum_{u = -\infty}^{+\infty}\sum_{v = -\infty \setminus (0,0)}^{+\infty} ((ua+d_{x})^{2}+(va&+d_{y})^{2}\nonumber\\
&+h^{2})^{-\beta}.
\label{eqn:inter2d}
\end{align} 

For two dimensions, we can scale and shift the function $q(u,v)$ in \eqref{eqn:intb1} as $q(d_{x}+ua,d_{y}+va)$ and from \cite{poissonNd} we extend the time shifting \cite[Prop.4.3.2]{oppenheim} and scaling \cite[Prop.4.3.5]{oppenheim} property of Fourier transform for two dimensions to obtain   

\begin{equation*}
a^{2}\sum_{u= -\infty}^{+ \infty}\sum_{v=-\infty}^{+ \infty} q(d_{x} + ua, d_{y} + va) =
\end{equation*}
\begin{equation}
\sum_{w=-\infty}^{+\infty}\sum_{k=-\infty}^{+\infty}Q\Big(\frac{w}{a},\frac{k}{a}\Big)e^{\frac{j2\pi wd_{x}}{a}}e^{\frac{j2\pi kd_{y}}{a}}.
\label{eqn:poisson2d}  
\end{equation}
Now, from \eqref{eqn:inter2d}, $q(u,v)$ can be expressed as a real and even function, given as     
\begin{equation}
q(u,v) = ( u^{2} + v^{2} + h^{2} )^{-\beta}.
\label{eqn:func2d}
\end{equation}
We define a parameter $s =2\pi \sqrt{w^{2}+k^{2}}$ and $r = \sqrt{u^{2}+v^{2}}$. So, $q(u,v)$ can be expressed as a radially symmetric function $q(r)$ as
\begin{equation}
q(r) = ( r^{2} + h^{2} )^{-\beta}.
\label{eqn:funcRad2d}
\end{equation}
Let $Q(s)$ be the radial Fourier transform of $q(r)$. We evaluate this using the Hankel function \cite{hankel} for two dimensions. The Hankel function for $n$ dimensions is defined as
\[s^{\frac{n-2}{2}}Q_{n}(s) = (2\pi)^{\frac{n}{2}}\int_{0}^{\infty}J_{\frac{n-2}{2}}(sr)r^{\frac{n-2}{2}}q(r)r\d r.\]
For $n=2$, 
\begin{align}
&Q(s) \nonumber\\
&= (2\pi)\int_{0}^{\infty}J_{0}(sr)q(r)r\d r, \nonumber\\
&= (2\pi)\int_{0}^{\infty}\frac{J_{0}(sr)r}{(r^{2} + h^{2})^{\beta}}\d r, \nonumber\\
&= \left\{
               \begin{array}{l}
                 \frac{h^{2-2\beta}\pi}{a^{2}(\beta-1)}; \ \ \ \ \ \ \ \ \ \ \ \ \ \ \ \ \ \ \ s=0, \\ \\
                 \frac{2^{2-\beta}\pi}{\Gamma(\beta)}\big(\frac{h}{s}\big)^{1-\beta}\mathbb{K}_{\beta-1}(hs); s\neq0. 
              \end{array}
              \right. 
\label{eqn:rect3}   
\end{align} 
Considering the dimension of the attocell as $a$, we obtain the scaled Radial Fourier transform of $q(ar)$, from \eqref{eqn:rect3} as 
\begin{equation}
\frac{1}{a^{2}}Q\Big(\frac{s}{a}\Big) = \left\{
               \begin{array}{l}
                 \frac{h^{2-2\beta}\pi}{a^{2}(\beta-1)}; \ \ \ \ \ \ \ \ \ \ \ \ \ \ \ \ \ \ \ \ \ \ \ s=0, \\ \\
                 \frac{2^{2-\beta}\pi}{a^{\beta+1}\Gamma(\beta)}\big(\frac{h}{s}\big)^{1-\beta}\mathbb{K}_{\beta-1}\big(\frac{hs}{a}\big); s\neq0. \\ 
              \end{array}
              \right.
\label{eqn:rect4}   
\end{equation}

Substituting $s = 2\pi\sqrt{w^{2} + k^{2}}$ in \eqref{eqn:rect4}, we get $\frac{1}{a^{2}}Q\big(\frac{w}{a},\frac{k}{a}\big)$ as
\begin{equation}
\frac{1}{a^{2}}Q\Big(\frac{w}{a},\frac{k}{a}\Big) = \left\{
               \begin{array}{l}
                 \frac{h^{2-2\beta}\pi}{a^{2}(\beta-1)}; \ \ \ \ \ \ \ \ \ \ \ \ \ \ w=k=0, \\ \\
                 \frac{2^{2-\beta}\pi}{a^{\beta+1}\Gamma(\beta)}\Big(\frac{h}{2 \pi\sqrt{w^{2} + k^{2}}}\Big)^{1-\beta}\\ \mathbb{K}_{\beta-1}\Big(\frac{2 \pi h\sqrt{w^{2} + k^{2}}}{a}\Big); w\neq0, k\neq0. \\ 
              \end{array}
              \right. 
\label{eqn:rect5}   
\end{equation}

Now, substituting \eqref{eqn:rect5} in \eqref{eqn:poisson2d} and after removing the redundant term of $u=v=0$, which represents the signal power from the tagged LED, we get
\begin{align}
&\mathbb{I}_{\infty}(d_{x},d_{y})\nonumber\\
&=\sum_{u = -\infty}^{+ \infty}\sum_{v=-\infty \setminus (0,0)}^{+ \infty} ( (ua + d_{x})^{2} + (va + d_{y})^{2}  + h^{2} )^{-\beta}, \nonumber\\
&=\sum_{u = -\infty}^{+ \infty}\sum_{v=-\infty}^{+ \infty} ( (ua + d_{x})^{2} + (va + d_{y})^{2}  + h^{2} )^{-\beta}\nonumber\\
&\ \ \ \ -\frac{1}{(d_{x}^{2} + d_{y}^{2} + h^{2})^{\beta}}, \nonumber\\
&=\frac{1}{a^{2}}\bigg(Q(0,0) \nonumber\\
&\ \ \ \ \ \ \ \ +\sum_{w=-\infty}^{+\infty}\sum_{k=-\infty \setminus(0,0)}^{+\infty}Q\bigg(\frac{w}{a},\frac{k}{a}\bigg)e^{\frac{j2\pi wd_{x}}{a}}e^{\frac{j2\pi kd_{y}}{a}}\bigg) \nonumber\\
&\ \ \ \ -\frac{1}{(d_{x}^{2}+d_{y}^{2}+h^{2})^{\beta}}.
\label{eqn:fourierEqn}  
\end{align} 

Now, because $q(u,v)$ in \eqref{eqn:func2d} is a real and even signal, it's Fourier transform is also real and even \cite[Prop.4.3.3]{oppenheim}. So, we modify \eqref{eqn:fourierEqn} as

\begin{align*}
&\mathbb{I}_{\infty}(d_{x},d_{y}) \nonumber\\
&\stackrel{(a)}{=}\frac{1}{a^{2}}\bigg(Q(0,0) \nonumber\\
&+\sum_{w=-\infty}^{+\infty}\sum_{k=-\infty \setminus(0,0)}^{+\infty}Q\bigg(\frac{w}{a},\frac{k}{a}\bigg)\cos\bigg(\frac{2\pi wd_{x}}{a}\bigg)\nonumber\\
&\ \ \ \ \ \ \ \ \ \ \ \ \ \ \ \ \ \  \cos\bigg(\frac{2\pi kd_{y}}{a}\bigg)\bigg)  -\frac{1}{(d_{x}^{2} + d_{y}^{2} + h^{2})^{\beta}}, \nonumber\\
&\stackrel{(b)}{=}\frac{1}{a^{2}}\bigg(Q(0,0) \nonumber\\
&+\sum_{w=0}^{+\infty}\sum_{k=0 \setminus(0,0)}^{+\infty}4Q\bigg(\frac{w}{a},\frac{k}{a}\bigg)\cos\bigg(\frac{2\pi wd_{x}}{a}\bigg)\nonumber\\
&\ \ \ \ \ \ \ \ \ \ \ \ \ \ \ \ \ \  \cos\bigg(\frac{2\pi kd_{y}}{a}\bigg)\bigg)  -\frac{1}{(d_{x}^{2} + d_{y}^{2} + h^{2})^{\beta}}, \nonumber\\
&= \frac{h^{2-2\beta}\pi}{a^{2}(\beta-1)} \nonumber\\
&\ \ \ +\sum_{w=0}^{+\infty}\sum_{k=0\setminus(0,0)}^{+\infty}\bigg(\frac{2^{4-\beta}\pi}{a^{\beta+1}\Gamma(\beta)}\Big(\frac{h}{2 \pi \sqrt{w^{2} + k^{2}}}\Big)^{1-\beta} \nonumber\\ 
&\ \ \ \ \ \ \ \mathbb{K}_{\beta-1}\Big(\frac{2 \pi h\sqrt{w^{2} + k^{2}}}{a}\Big)\cos\Big(\frac{2\pi wd_{x}}{a}\Big) \nonumber\\
&\ \ \ \ \ \ \ \ \ \ \ \ \ \ \ \ \ \  \cos\bigg(\frac{2\pi kd_{y}}{a}\bigg)\bigg)  -\frac{1}{(d_{x}^{2} + d_{y}^{2} + h^{2})^{\beta}},
\end{align*}
where $(a)$ follows from the fact that $Q(w,k)$ is real, $(b)$ follows from the fact that $Q(w,k)$ is even and hence proving the theorem. 
\end{proof}
\label{app:theorem2}
\section{}
\section*{Proof of Proposition \ref{prop2}}
\begin{proof}
 Consider the set $\mathbb{A} \triangleq (\mathbb{Z}^{2}\cap([0,j]\times[0,l]))\setminus\{(0,0)\}$, over the set of integers $\mathbb{Z}$.
The interference power obtained in Thm. \ref{theorem2} can be written as 
\begin{align*}
\mathbb{I}_{\infty}(d_{x},d_{y}) =&\frac{h^{2-2\beta}\pi}{a^{2}(\beta-1)}-\frac{1}{(d_{x}^{2} + d_{y}^{2} + h^{2})^{\beta}} \nonumber\\
&+\sum_{(w,k)\in\mathbb{A}}g(w,k)+\sum_{(w,k)\notin\mathbb{A}\cup\{(0,0)\}}g(w,k).
\end{align*}

Let $E(j,l)=\sum_{(w,k)\notin\mathbb{A}\cup\{(0,0)\}}g(w,k)$ and $r=\sqrt{w^{2}+k^{2}}$, then
\begin{align}
&g(w,k)=\psi(r) \nonumber\\
&= \frac{M}{a^{\beta+1}}\cos\bigg(\frac{2\pi w d_{x}}{a}\bigg)\cos\bigg(\frac{2\pi k d_{y}}{a}\bigg)r^{\beta-1}\mathbb{K}_{\beta-1}\bigg(\frac{2\pi h r}{a}\bigg),
\label{eqn:prop26}
\end{align} 
where $M=\frac{2^{3}\pi^{\beta}h^{1-\beta}}{\Gamma(\beta)}$. From \cite{asybesl}, for large $r$, the modified Bessel function $\mathbb{K}_{\beta-1}\big(\frac{2\pi h r}{a}\big)$ can be expanded as
\begin{equation*}
\mathbb{K}_{\beta-1}\bigg(\frac{2\pi h r}{a}\bigg) = \sqrt{\frac{a}{4hr}}e^{\frac{-2\pi h r}{a}}\Big(1+\Theta\Big(\frac{1}{r}\Big)\Big),
\end{equation*} 

where $\Theta(.)$ is the same asymptotic notation introduced for the one dimension model. The cosine terms $\cos\big(\frac{2\pi w d_{x}}{a}\big)$ and $\cos\big(\frac{2\pi k d_{y}}{a}\big)$ in \eqref{eqn:prop26} are bounded by $\Theta(1)$, as $w$ or $k$ becomes large. So, for large $r$, the asymptotic bound on $\psi(r)$ in \eqref{eqn:prop26}, can be shown as

\begin{align}
\psi(r)\leq &\frac{M}{a^{\beta+1}}\Theta(1)r^{\beta-1}\sqrt{\frac{a}{4h r}}e^{\frac{-2\pi h r}{a}}\Big(1+\Theta\Big(\frac{1}{r}\Big)\Big),\nonumber\\
\in&\Theta(r^{\beta-2.5}e^{\frac{-2\pi h r}{a}}).
\label{eqn:prop42}
\end{align} 

Summing \eqref{eqn:prop42} over large $r$, we get the total error $E(j,l)$ as 
\begin{align} 
E(j,l) &= \sum_{r:(w,k)\notin \mathbb{A}\cup\{(0,0)\}} \psi(r), \nonumber\\
&\in \sum_{r:(w,k)\notin \mathbb{A}\cup\{(0,0)\}} \Theta(r^{\beta-2.5}e^{\frac{-2\pi h r}{a}}), \nonumber\\
&= \Theta\Bigg(\sum_{r:(w,k)\notin \mathbb{A}\cup\{(0,0)\}} |r^{\beta-2.5}e^{\frac{-2\pi h r}{a}}|\Bigg).
\label{eqn:newstar2}
\end{align}
The summand inside the summation in \eqref{eqn:newstar2} is symmetric with $r$. So $E(j,l)$ can be bounded by a symmetric summation as 
\[E(j,l)\leq \Theta\Bigg(\sum_{r\geq \sqrt{j^{2}+l^{2}}+1} |r^{\beta-2.5}e^{\frac{-2\pi h r}{a}}|\Bigg).\]
Observe that summand $h(r) =r^{\beta-2.5}e^{\frac{-2\pi h r}{a}}$ increases and then decreases with respect to $r$ and attains its maximum at $r_{0}=\frac{a (\beta-2.5)}{2\pi h}$. Let $r_{1}=\max\{\sqrt{j^{2}+l^{2}}+1, \lceil r_{0} \rceil\}$. Hence we have
\begin{align}
\sum_{r=\sqrt{j^{2}+l^{2}}+1}^{\infty}|h(r)|&\leq (r_1-(\sqrt{j^{2}+l^{2}}+1))|r_1^{\beta-2.5}e^{\frac{-2\pi h r_1}{a}}|\nonumber\\
&\ \ \ \ \ \ \ \ \ \ \ \ \ +\int_{r_1}^\infty |r^{\beta-2.5}e^{\frac{-2\pi h r}{a}}| \d r,\nonumber\\
&=(r_1-(\sqrt{j^{2}+l^{2}}+1)) |r_1^{\beta-2.5}e^{\frac{-2 \pi h r_{1}}{a}}|\nonumber\\
&\ \ \ \ \ +\bigg(\frac{a}{2\pi h}\bigg)^{\beta-1.5}\Gamma\bigg(\beta-1.5,\frac{2\pi h r_1}{a}\bigg),
\end{align}
where $\Gamma(x,t)=\int_{t}^{\infty}r^{x-1}e^{-r}\d r$ is the incomplete gamma function. Now, for large $t$, from \cite{gammap}, we can asymptotically bound $\Gamma(x,t)$ as 
\[ \Gamma(x,t) \leq \Theta(t^{x-1}e^{-t}). \]
Using this result and the fact that for large $\sqrt{j^{2}+l^{2}}$ or large $\frac{h}{a}$, $r_1=\sqrt{j^{2}+l^{2}}+1$, we have  
\begin{align*}
E(j,l)&\leq \Theta\bigg(\bigg(\frac{a}{2 \pi h}\bigg)^{\beta-1.5} \bigg(\frac{2\pi h}{a}(\sqrt{j^{2}+l^{2}}+1)\bigg)^{\beta-2.5} \nonumber\\
& \ \ \ \ \ \ \ \ \ \ \ \ \ e^{\frac{-2\pi h(\sqrt{j^{2}+l^{2}}+1)}{a}}\bigg),\nonumber\\
&\in \Theta \big((\sqrt{j^{2}+l^{2}}+1)^{\beta-2.5}e^{\frac{-2\pi h (\sqrt{j^{2}+l^{2}}+1)}{a}}\big), 
\end{align*}
proving the proposition.       
\end{proof}
\label{app:theorem2a}

\end{appendices}

\end{document}